\documentclass[letterpaper,12pt,titlepage,oneside,final]{report}

\usepackage{ifthen}
\newboolean{shortdraft}

\ifthenelse{\boolean{shortdraft}}{
\usepackage[paperheight=8in, paperwidth=6.5in, margin=0.2in]{geometry}
}{}

\usepackage{amsmath,amssymb,amstext,amsthm} 
\usepackage[pdftex]{graphicx} 

\usepackage{makeidx}
\makeindex
\usepackage{nomencl}
\makenomenclature

\usepackage[letterpaper=true,pdftex,bookmarks,pagebackref,
	plainpages=false, 
        pdfpagelabels=true, 
        pdftitle=Quantum Information Processing with Adversarial Devices, 
        pdfauthor=Matthew\ McKague 
        ]{hyperref}

\usepackage{minitoc}
\dominitoc

\input{Qcircuit}

\xyoption{all}

\newcommand{\proj}[2]{\ket{#1}\!\!\bra{#2}}
\newcommand{\braket}[2]{\left \langle #1 | #2 \right \rangle}
\newcommand{\tr}[2][ ]{\text{Tr}_{#1}\!\left( #2 \right)}

\newtheorem{definition}{Definition}[chapter]

\newtheorem{theorem}{Theorem}[chapter]
\newtheorem{lemma}{Lemma}[chapter]
\newtheorem{corollary}{Corollary}[chapter]
\newtheorem{assumption}{Assumption}[chapter]
\newtheorem*{futurework}{Future work}
\newtheorem{conjecture}{Conjecture}[chapter]

\ifthenelse{\boolean{shortdraft}}{
	\renewcommand{\cite}[1]{[??]}
}{
	\bibliographystyle{halphads}
}

\begin{document}

\ifthenelse{\boolean{shortdraft}}{\tableofcontents
}{
\pagestyle{empty}
\pagenumbering{roman}

\begin{titlepage}
        \begin{center}
        \vspace*{1.0cm}

        \Huge
        {\bf Quantum Information Processing with Adversarial Devices}

        \vspace*{1.0cm}

        \normalsize
        by \\

        \vspace*{1.0cm}

        \Large
       Matthew McKague \\

        \vspace*{3.0cm}

        \normalsize
        A thesis \\
        presented to the University of Waterloo \\ 
        in fulfillment of the \\
        thesis requirement for the degree of \\
        Doctor of Philosophy \\
        in \\
        Combinatorics \& Optimization\\

        \vspace*{2.0cm}

        Waterloo, Ontario, Canada, 2010 \\

        \vspace*{1.0cm}

        \copyright\ Matthew McKague 2010 \\
        \end{center}
\end{titlepage}

\pagestyle{plain}
\setcounter{page}{2}


  \noindent
  I hereby declare that I am the sole author of this thesis.  This is a true copy of the thesis, including any required final revisions, as accepted by my examiners.

  \bigskip
  
  \noindent
  I understand that my thesis may be made electronically available to the public.
  
  \newpage


\begin{center}\textbf{Abstract}\end{center}

We consider several applications in black-box quantum computation in which untrusted physical quantum devices are connected together to produce an experiment.  By examining the outcome statistics of such an experiment, and comparing them against the desired experiment, we may hope to certify that the physical experiment is implementing the desired experiment.  This is useful in order to verify that a calculation has been performed correctly, that measurement outcomes are secure, or that the devices are producing the desired state.

First, we introduce constructions for a family of simulations, which duplicate the outcome statistics of an experiment but are not exactly the same as the desired experiment.  This places limitations on how strict we may be with the requirements we place on the physical devices.  We identify many simulations, and consider their implications for quantum foundations as well as security related applications. 

The most general application of black-box quantum computing is self-testing circuits, in which a generic physical circuit may be tested against a given circuit.  Earlier results were restricted to circuits described on a real Hilbert space.  We give new proofs for earlier results and begin work extending them to circuits on a complex Hilbert space with a test that verifies complex measurements.

For security applications of black-box quantum computing, we consider device independent quantum key distribution (DIQKD).  We may consider DIQKD as an extension of QKD (quantum key distribution) in which the model of the physical measurement devices is replaced with an adversarial model.  This introduces many technical problems, such as unbounded dimension, but promises increased security since the many complexities hidden by traditional models are implicitly considered.  We extend earlier work by proving security with fewer assumptions.

Finally, we consider the case of black-box state characterization.  Here the emphasis is placed on providing robust results with operationally meaningful measures.  The goal is to certify that a black box device is producing high quality maximally entangled pairs of qubits using only untrusted measurements and a single statistic, the CHSH value, defined using correlations of outcomes from the two parts of the system.  We present several measures of quality and prove bounds for them.

\newpage


\begin{center}\textbf{Acknowledgements}\end{center}
Thanks to my supervisor, Michele Mosca, for guidance throughout my studies.

I am also grateful to my co-authors:  Valerio Scarani, Serge Massar, Timothy Liew, Charles-Edouard Bardyn, and Nicolas Gisin in addition to Michele Mosca.  

During my studies I had interesting discussions about my work with many people, most notably Douglas Stebila, Jamie Sikora, Bill Wootters, Frederic Magniez, and Sandu Popescu.  Thanks to each of them for providing new perspectives, questions, and a patient ear.  Also, thanks to my \LaTeX  guru, Niel de Beaudrap.

During my studies I was financially supported by NSERC, QuantumWorks, Ontario Centres of Excellence, MITACS, CIFAR, ORF, the Government of Canada, and Ontario-MRI.

Finally, thanks to Douglas Stebila, Lana Sheridan, Paul Dickinson, Zach Olesh and David Clark for their friendship during my graduate studies.

\newpage


\begin{center}\textbf{Dedication}\end{center}
To my beautiful wife, for her never-failing support during my studies.

\newpage


\tableofcontents
\newpage


\addcontentsline{toc}{section}{List of Figures}
\listoffigures
\newpage


\pagenumbering{arabic}

}

\chapter{Introduction}
\minitoc
\section{States and statistics}
The main thrust behind the development of quantum formalism, and its main usefulness, arises from its ability to predict the outcomes of experiments.  Indeed, the kinds of predictions that quantum formalism makes has revolutionized physics.  In this thesis, however, we take a much different approach.  Instead of using quantum states and operations to determine the outcomes of an experiment, we will use the outcomes of an experiment (or rather the distribution of outcomes) to determine the quantum state and operations.

Why would we want to do this?  As quantum formalism moves from the role of a theoretical tool in describing the functioning of the universe to an integral part of technological developments it becomes more important that physical devices are operating as we believe them to be.  This is no more apparent than in the case of quantum key distribution in which we posit the existence of an adversary who actively subverts the functioning of physical devices and modifies quantum states.  In experiments intended for academic applications we may rely on academic integrity and repeatability to establish that results are correct, but in the world of security we must be sure that each time we use a device it actually behaves as we believe it to behave, without the benefit of time, repetition, and expert opinion.  Indeed, we may have little to go on besides the assurances of the device manufacturer.

Into this context we introduce the concepts of self-testing, device independent quantum key distribution and black box state characterization.  The aim of all of these techniques is to replace assumptions about a physical devices with a test.  The test will rely solely on the classical data available about the devices:  how are the devices connected to each other?  What are the measurement setting?  What are the outcomes?  When we look at the probability distributions for these data we may be able to certify that the devices are behaving properly, or raise a red flag if the functioning of the devices stray from the ideal.

\section{Terminology}
\subsection{\index{black box model}Black box model}

\subsubsection{Devices}
In this thesis we have a particular model in mind for black box computing.  In particular, each black box device is a physical device with some combination of quantum and classical inputs and outputs.  Devices do not communicate unless we allow them to by physically connecting them together.  Each device is labeled with its intended function, but this may bear no resemblance to what the device actually does.

We will concern ourselves with three main types of black box devices
 \begin{itemize}
	\item \emph{Sources}:  The only sources we will use are bipartite state sources with two quantum outputs.  They are labeled with the state they produce (always $\frac{1}{\sqrt{2}}(\ket{00} + \ket{11})$ in this thesis.)
	\item \emph{Gates}:  Gates will have an equal number of quantum inputs and outputs.  Gates will be labeled with a matrix corresponding to the unitary they are supposed to apply.
	\item \emph{Measurement devices}:  We will be concerned with single system measurements in a small number of bases.  They have one quantum input and one classical output (one bit in all the cases we consider).  We may consider the different measurement bases to be implemented in different physical devices, or there may be a classical input which specifies the basis in which to measure.  Measurement devices are labeled with the basis (or bases) in which they measure.  We will model measurements as Hermitian observables.  We may easily translate the more general POVM formalism into such a description using Naimark's theorem \cite{Naimark:1940:Spectral-functi}.  Since the dimension of the Hilbert space is not fixed this poses no problems.
\end{itemize}

\subsubsection{Device interaction}
We need the devices to be able to interact in order to form circuits, but not in an unlimited fashion.  If unlimited communication were allowed, then the various devices could operate as a single device with a classical conspiracy.  An ideal interaction, from the verifier's point of view, would allow only one-way communication.  This would greatly simplify models.  Two way communication could lead to a situation equivalent unlimited communication depending on which devices were connected simultaneously.  For example, if all devices were simultaneously connected and the graph of their connections were connected then any two devices could communicate through intermediate devices.

One way to enforce one-way communication would be to use a number of identity gate devices and pairwise communication.  If we wish for device $D_{1}$ to send a state to $D_{2}$ we first interact $D_{1}$ with an identity device $I$ and the interact $I$ with $D_{2}$.  This guarantees that no communication occurs from $D_{2}$ back to $D_{1}$.

\subsubsection{State preparation}
With black box devices we have no control over what state the device has at its disposal.  We might interact a device with a source, but the device may disregard the source and use a state that it already had in memory.  In this way, gate and measurement devices may share an arbitrary amount of entanglement.  However, this does not give the adversary any more power since we can replace a state in memory with one prepared in the source.  For this reason, we will generally speak of the state as being prepared by the source device without loss of generality.

\subsection{Reference experiments and simulations}

\subsubsection{Experiments}
The usual way of arranging quantum gates is into a quantum circuit, which consists of a state preparation followed by several unitary gates and finally a measurement.  For our purposes this will not be sufficient.  In particular we will need to measure in several different bases chosen classically by an agent outside the circuit.  We will refer to this type of quantum apparatus as an \emph{\index{experiment}experiment}.  An important distinction to be made is between the experiment we wish to implement, and the experiment that is actually implemented by the quantum devices.  The first we call the \emph{\index{reference experiment}reference experiment}, and the second the \emph{\index{physical experiment}physical experiment}.  

Each experiment will have several classical inputs which determine measurements to be made, which we call the \emph{\index{measurement settings}measurement settings}.  The \emph{\index{statistics generated by experiment}statistics generated by an experiment} is the probability distribution of the measurement outcomes, conditioned on the measurement setting.  Another important aspect of an experiment is the \emph{\index{topology}topology}.  This is the division of the experiment into distinct elements, each implemented by a single device, and the connections between these elements.  A typical topology consists of a bipartite state preparation device connected to two measurement devices.  Importantly, we assume that no signalling occurs between devices that are not explicitly connected in the topology, and if there is some physical channel then this cannot be used to communicate backwards.
\begin{assumption}
No communication occurs between quantum devices that are not explicitly connected.  Further, quantum communication from the output of one device to the input of another device is one-way.
\end{assumption}

\subsubsection{Simulations}
The verifier has two pieces of information about an experiment.  First, the topology is known, since the verifier is responsible for it.  Second, the verifier may estimate the statistics generated by the experiment, provided the experiment behaves the same each time it is used.
\begin{assumption}
The quantum devices behave the same each time they are used.\footnote{We will consider less restrictive behaviour in Chapter~\ref{chapter:diqkd}.}
\end{assumption}
The verifier has a particular reference experiment in mind and wishes to implement it.  The verifier then builds a physical experiment with the same topology as the reference experiment using quantum devices and estimates the statistics it generates.  Then the verifier may check if the statistics match those generated by the reference experiment.
\begin{definition}
If a physical experiment generates the same statistics as a reference experiment with the same topology, then the physical experiment \emph{\index{simulates}simulates} the reference experiment.
\end{definition}

We will wish to compare the physical and reference systems in various ways.  Our goal will always be to show that they are the same in some way or another (or determine they are not and abort.)  The raw data that we obtain from any interaction with the physical system will be in the form of classical outcomes from measurements.  These classical outcomes are obtained from probability distributions that are determined by the physical system.

Another important concept that we will sometimes refer to is that of a \index{conspiracy}conspiracy.  A conspiracy is any behaviour of the black box devices that attempts to defeat a test without actually implementing the reference circuit.  One example is for a circuit to classically calculate statistics for a circuit, rather than implementing the circuit.  This classical conspiracy immediately implies that we must use non-locality in our testing, since any circuit that is implemented locally has a classical conspiracy.

\section{Summary of new results}

\subsubsection{Simulations}
In chapter 2 we consider simulations of circuits, which are other circuits which produce the same outcome statistics as a reference circuit.  The first main result of this section will be to develop a means of transforming a reference circuit into a simulation circuit which is described using only real numbers.  Although such a simulation was known in limited contexts, the simulation we develop works for multi-partite systems general operations such as completely positive maps, POVMs (positive operator valued measure), and Hamiltonians.  The second main result is to describe a family of simulations that generalizes the real simulation.  

The real simulation developed is interesting from the perspective of quantum foundations.  In particular, the simulation proves that there are no experiments which distinguish between quantum physics over real Hilbert spaces and quantum physics over complex Hilbert spaces (which we will from now on refer to as real quantum physics, and complex quantum physics, respectively.)  From the point of view of quantum foundations this eliminates the need to specify which of these two fields to use.  In this spirit we also consider other number systems that might be used, specifically  quaternions.  We show that if we use the quaternions it is possible to implement a non-local box, which allows for stronger non-local effects than can be achieved within complex quantum physics.  In principle this allows for an experiment which can distinguish quaternionic quantum physics from complex quantum physics (although not the reverse, since complex quantum physics is contained within quaternionic quantum physics.)

Finally, in chapter 2 we consider the security of cryptographic protocols the are implemented using simulations from the continuum described above.  We show that for a limited class of protocols (where all operations of the honest parties are measurements) security is not compromised as compared to the reference protocol.

\subsubsection{Self-testing}
Self-testing is the workhorse of black box computing.  It allows us to test generic quantum circuits, and can be adapted to use with cryptographic protocols such as QKD.  In this area our major contribution is new proofs for the major results (state and measurement testing, and gate testing.)  The new proofs are simpler and much easier to understand.  In the case of gate testing, we provide new proofs for two technical lemmas whose original proofs relied on incorrect assumptions.  

Besides new proofs, we also introduce a new test for states and measurements.  The previously known test used only measurements with real entries, whereas our new test introduces extra measurements which have complex entries.  This paves the way for gate testing of arbitrary gates, where the previous test was only applicable to gates with real entries.

\subsubsection{Device independent quantum key distribution}
Device independent quantum key distribution represents a new approach to QKD security proofs by using an adversarial device model for the participants' quantum devices.  The goal is to provide higher security by removing untestable assumptions and replacing them with physical tests.  Previous work in this area identified the main problems, introduced a usable protocol and produced a limited security proof (analogous to security against collective attacks in the usual QKD model.)  Our contribution is to push the security boundary further by proving security against a larger class of attacks.  These attacks allow arbitrary states and relax the assumptions on the devices.  We also introduce new proof techniques to DIQKD, adapting security proofs used for traditional QKD using the quantum de Finetti theorem.

\subsubsection{Black box state characterization}
We finally consider black box state characterization.  The goal is very much similar to that of state and measurement testing we discuss in chapter 3:  to test a state preparation device using only untrusted measurements.  However, the focus is placed on finding a robust result, measuring the quality of the state with an operationally meaningful definition (in this case, using fidelity.)  Our contributions are to find suitable definitions for measures of quality (we introduce several) and to prove bounds using these definitions.  We also discuss the relationship between the various definitions.


\chapter{Simulations}\index{simulation}\label{chap:simulations}
\minitoc

\section{Introduction}
Suppose that a physical experiment simulates a reference experiment.  What can we conclude?  As we shall see, the answer to this question will be quite sensitive to the topology and specific construction of the reference experiment and will vary from nothing to quite strong statements about the structure of the physical experiment.  Before we consider this question, however, we look at the complementary question:  what can we never hope to conclude?

We will have to be quite specific about the topologies that we consider.  For example, take a topology with only one part (i.e. a single-partite state) with a simple reference experiment that implements a search by iterating through a list.  This experiment can be simulated by a Grover search, or a procedure that sorts the list and performs a binary search.  The structures of these experiments are quite different.  We will not be interested in these types of specific simulations.  Instead, we wish to find general procedures for producing a simulation of a given circuit.

\subsubsection{Contributions}
In this chapter we make many original contributions.  First, we extend the well known real circuit simulation (described below) to allow POVMs, mixed states, CP maps, and continuous time evolution, completing the suite of general tools used in quantum formalism.  Next we show that the real simulation may be applied in the case of multipartite evolution through the use of entangled states.

The real simulation is then generalized to a large family of simulations.  The simulations share a common construction for operators and differ in their states.  Roughly, the simulations are formed as a mixture of the reference experiment and its complex conjugate.

Later, in chapter~\ref{chap:selftesting}, we will discuss how we can ensure that a physical experiment is one of the generalized simulations.  We may wish to use the physical experiment in a cryptographic protocol, in which case the security properties of the experiment become important.  We show that, for a restricted set of protocols, the security properties if implemented by a generalized simulation are identical to if implemented by the reference experiment.

Finally, we consider some issues in quantum foundations.  Specifically, we discuss how the real simulation implies that the quantum formalism produces identical predictions whether the real numbers or complex numbers are used as the underlying number system.  As well, we briefly consider the possibility of using the quaternions instead of the real or complex numbers and show that there exists an experiment for which the quantum formalism over the quaternions allows for different predictions than are possible with the usual quantum formalism.

Some of the original material in this chapter was presented in \cite{Matthew-McKague:2007:Simulating-Quan} and published in \cite{McKague:2009:Simulating-Quan}.  Further material is available in \cite{McKague:2009:Quaternionic-qu} and is to appear in \cite{McKague:2010:Generalized-sel}.

\section{Literature review}
The concept of a simulation is widespread in the field of quantum computing, although the notion is usually more general, requiring only that the statistics of the measurements are the same as, or approximately the same as, the reference system.  As an example, Aharonov et al. \cite{Aharonov:2004:Adiabatic-Quant} showed that adiabatic quantum computing is equivalent to the circuit model by describing how to construct an adiabatic simulation for any given quantum circuit.  Another family of examples is the various constructions for universal sets of quantum gates (see Nielsen and Chuang \cite{Michael-A.-Nielsen:2000:Quantum-Computa}) where a new circuit containing only gates from a restricted set is constructed from a given circuit.  In some of these constructions the outcome statistics are only approximately the same as in the given circuit, although the error can be made arbitrarily small.  In all these constructions the form of the system is changed, and only the outcome statistics are the same.  However, despite the difference in the system the final state is approximately the same as in the reference system (with an ancilla, in the case of adiabatic computing).

The simple real circuit simulation, described below in section~\ref{sec:simplerealsimulation}, generalizes work by Rudolph and Grover \cite{Rudolph:2002:A-2-rebit-gate-}.  Their simulation construction, like those for universal gate sets, has a restricted set of gates described using only real numbers, but the final state is in general not the same as in the reference circuit.  The construction is well known outside of Rudolph and Grover's work.  It also appears, for example, in work by Stueckelberg \cite{Stueckelberg:1960:Quantum-theory-}.

The simulation constructions considered so far have not aimed at preserving any of the topology of the reference system.  This makes them unsuitable in situations where the division into multiple physical systems is important, as in the case of Bell inequalities.  However, P\'{a}l and V\'{e}rtesi \cite{Pal:2008:Efficiency-of-h} considered exactly this situation.  In particular, they considered the scenario of two physical systems with local measurements.  Their result, derived independently of the results presented here, provided a simulation construction that duplicated the outcome statistics with measurement observables and states described using only real numbers.  Importantly, the measurements in the simulation are local, so that the division into two physical systems is respected in the simulation.

\section{\index{unitarily equivalent simulations}Unitarily equivalent simulations}
We first explore various simple modifications on the reference system that do not modify the outcome statistics.  We call the results of such modifications \index{unitarily equivalent simulations}\emph{unitarily equivalent simulations}.  The idea we wish to convey with this terminology is that the experiment has been modified in a way that is analogous to a change of basis, or changing the description of the Hilbert space.

One obvious and simple modification of a circuit is to make a change of basis.  We must be careful in order to respect the division of the circuit into multiple systems.  For this reason we consider only local changes of basis.  We apply the change of basis (i.e. a unitary) to the initial state and conjugate all unitaries and measurement observables in the circuit by the change of basis operation.  The outcome statistics are unchanged and the circuit is not the same.

A local change of basis may also be made between operations.  The idea here is that a wire between two unitaries represents a quantum channel that carries the state from one physical device to another.  This quantum channel may be anything so long as it faithfully preserves the state.  In particular, the channel may apply some arbitrary change of basis at the beginning, and then reverse the change of basis at the end.  We may incorporate this change of basis into the two unitaries, as shown in figure~\ref{fig:change_of_basis_between_gates}.  As an extension of this principle, if multiple wires leave one gate and enter another, we may apply a change of basis to all the wires simultaneously.

 \begin{figure}
\[ \Qcircuit {
\ket{\psi} & & \gate{A} & \gate{B} & \measureD{M}
 }
\]
 \[ = \]
  \[\Qcircuit{
  U\ket{\psi} & & \gate{UAU^{\dagger}} & \gate{UBU^{\dagger}} & \measureD{UMU^{\dagger}}
  }
\]
\caption{Local change of basis on a wire}
\label{fig:wire_change_of_basis}
\end{figure}
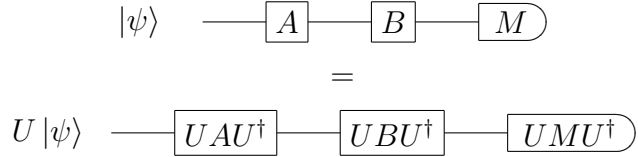

 \begin{figure}
\[ \Qcircuit {
  & \gate{A}& \qw & \qw & \qw & \gate{B} & \qw \\
 }
\]
\[ = \]
\[ \Qcircuit {
  & \gate{A}& \gate{U} & \qw & \gate{U^{\dagger}} & \gate{B} & \qw \\
 }
\]
\[ = \]
\[ \Qcircuit {
  & \gate{UA}& \gate{BU^{\dagger}} & \qw \\
 }
\]
\caption{Local change of basis between gates}
\label{fig:change_of_basis_between_gates}
\end{figure}
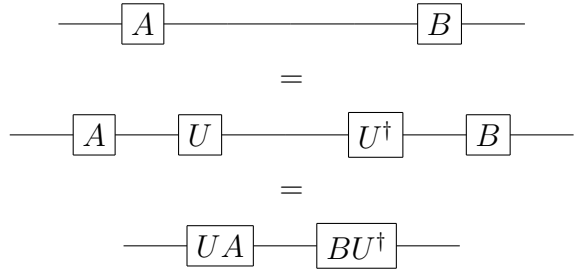

An experiment may be described on a particular Hilbert space, but the entire space may not be necessary to describe the state (i.e. if the support of the state is not the entire Hilbert space.)  In this case the operation of the gates outside the support of the state may be changed without modifying the functioning of the experiment.  In addition, the Hilbert space may be changed, either embedded in a larger Hilbert space or reduced to only a subspace without modifying the functionality of the experiment.  The idea here is that our model of the system need not consider dimensions that the state does not visit, or our model may be changed arbitrarily on those dimensions without changing how the experiment operates.  

Finally, we consider ancillas.  The presence of ancillas does not change the outcome statistics provided that the unitaries operate independently on the ancillas and the measurements do not operate on the ancillas.  The state of the ancilla does not matter, so we may consider ancillas added to each system in the experiment, prepared in any state.

Of course, any combination of the modifications would not change the outcomes, and so would produce a simulation.  After many such modifications the simulation may not look very simple, but the underlying operation of the experiment is essentially untouched.  Rather, the way that we are describing it has changed since the choice of basis, division into systems and ancillas, etc. are all a product of our description of the experiment rather than inherent in the experiment itself.

One final modification of a experiment is to take the complex conjugate of every state and operation in the circuit.  Obviously this does not change the outcome statistics, and the experiment will not be the same unless every element in the circuit has only real entries (in the same basis in which we apply the complex conjugation.)  However, there is no unitary operation that implements complex conjugation, so we will not consider it as unitarily equivalent.  We will consider it more fully later in section~\ref{sec:generalsimulations}.

\section{\index{real simulation}Real simulation}\label{sec:real_simulation}

\subsection{Simple real circuit simulation}\label{sec:simplerealsimulation}
We first show a way of arriving at a real simulation of a given reference circuit.\footnote{This simulation has been present in the folklore for some time.}  This is based on the observation that the complex numbers are a two dimensional real vector space.  We can thus, at least formally, represent an arbitrary quantum state as a real vector as follows
 \[
 \sum_{x} \left(a^{R}_{x} + i a^{I}_{x}\right) \ket{x} \mapsto \sum_{x} \left(a^{R}_{x} \ket{0} + a^{I}_{x} \ket{1}\right) \ket{x}.
\]
($a^{R}_{x}$ and $a^{I}_{x}$ are real numbers representing the real and imaginary part of the complex number $a_{x}$.)  Here we have simply replaced the two dimensional vector space spanned by $1$ and $i$ by a qubit.  The new state has the correct norm as a consequence of the formula
 \[
 |a_{x}|^{2} = \left(a^{R}_{x}\right)^{2} + \left(a^{I}_{x}\right)^{2}.
\]
Summing over $x$ on the left side gives us the norm of the reference state, while summing over $x$ on the right side gives us the norm of the simulation state.  This also shows that we can measure the $\ket{x}$ register on the simulation state in the computational basis and obtain outcome $x$ with the same probability as in the reference state.

So far we have shown how to simulate states and measurements in the computational basis.  Now we show how to simulate unitaries.  This is more complicated since we need to simulate, not only the norm of complex numbers, but the multiplication of complex numbers. Intuitively we can make a unitary work by making sure that the multiplication involved with each of the entries in the unitary matrix works.

Suppose that we have two complex numbers $a$, and $b$ that we want to multiply together.  We represent $a$ by the vector
$
\left(
	\begin{matrix}
	a^{R} \\ a^{I}
	\end{matrix}
\right)$.  
We want to multiply this vector by a suitable matrix derived from $b$ so that the result is the vector
 \[
 \left(
	\begin{matrix}
	(ab)^{R} \\ (ab)^{I}
	\end{matrix}
\right).
\]
To get the real part we need $a^{R}b^{R} - a^{I} b^{I}$, so the first row of the matrix should be $(b^{R}, -b^{I})$.  For the imaginary part we want $a^{R}b^{I} + a^{I} b^{R}$, so the second row should be $(b^{I}, b^{R})$.  Thus the matrix we want is
 \[
 \left(
	\begin{matrix}
	b^{R} & -b^{I} \\
	b^{I} & b^{R} 
	\end{matrix}
\right).
\]
To transform a matrix of complex numbers, we will replace each complex number by a $2 \times 2$ submatrix corresponding to the complex number as above.  It is easily shown that this transformation takes unitary matrices to unitary matrices and, as intuition suggests, the resulting unitary maps the simulation state to correctly track the evolution of the reference state.

\subsection{Complex numbers as matrices}

The intuition in the previous section can be made much more rigorous by using a simple field isomorphism.

\begin{lemma}
Define $R : \mathbb{C} \rightarrow M_{2}(\mathbb{R})$ by
 \[
R(a) = a^{R} I + a^{I} XZ
\]
where 
 \[
 XZ = 
 \left(
	\begin{matrix}
	0 & -1 \\
	1 & 0 \\
	\end{matrix}
\right)
\]
is the product of the Pauli matrices $X$ and $Z$.  Then $R$ is an isomorphism between $\mathbb{C}$ and its image under $R$.
\end{lemma}

The map $R$ naturally extends to matrices and vectors, mapping $n \times n$ matrices to $2n \times 2n$ matrices, and $n$ vectors to $2n \times 2$ matrices.  To be more rigorous, $R$ maps $n \times n$ matrices over the complex numbers to $n \times n$ matrices over $2 \times 2$ matrices, which is naturally mapped to $2n \times 2n$ matrices.  
 \[
 R\left(
 \left(
	\begin{matrix}
	a & b \\
	c & d \\
	\end{matrix}
\right)
\right)
= 
\left(
	\begin{matrix}
	\left(
	\begin{matrix}
	a^{R} & -a^{I} \\
	a^{I} & a^{R}
	\end{matrix}
	\right)
	&
	\left(
	\begin{matrix}
	b^{R} & -b^{I} \\
	b^{I} & b^{R}
	\end{matrix}
	\right) \\
	\left(
	\begin{matrix}
	c^{R} & -c^{I} \\
	c^{I} & c^{R}
	\end{matrix}
	\right)
	&
	\left(
	\begin{matrix}
	d^{R} & -d^{I} \\
	d^{I} & d^{R}
	\end{matrix}
	\right)
	\end{matrix}
\right)
\mapsto
\left(
	\begin{matrix}
	a^{R} & -a^{I} & b^{R} & -b^{R} \\
	a^{I} & a^{R} & b^{I} & b^{R} \\
	c^{R} & -c^{I} & d^{R} & -d^{R} \\
	c^{I} & c^{R} & d^{I} & d^{R} \\
	\end{matrix}
\right)
\]

It will frequently be useful to have the matrices output from $R$ to be defined over a tensor product space.  Specifically, if we are dealing with the Hilbert space $\mathcal{H}$ then we consider the output of $R$ to be $\mathcal{H}_{2} \otimes \mathcal{H}$.  Using Dirac notation we can then describe $R$ by

 \[
 R\left(\sum_{x,y} u_{xy} \proj{x}{y}\right) = \sum_{x,y} \left(u_{xy}^{R} I + u_{xy}^{I} XZ\right) \otimes \proj{x}{y}.
\]

We now give some properties of $R$, which are easily verified.

\begin{lemma}
Let operators $M$ and $N$ (which may have 1 column) be given.  Then
\begin{enumerate}
	\item $R(M^{\dagger}) = R(M)^{T} = R(M)^{\dagger}$
	\item $R(A + B) = R(A) + R(B)$
	\item $R(AB) = R(A)R(B)$
	\item If $M$ is normal and has (non necessarily distinct) eigenvalues $\lambda_{k}$ then $R(M)$ has eigenvalues $\lambda_{k}$ and $\lambda^{*}_{k}$.
	\item If $M$ is positive semi-definite, then $R(M)$ is positive semi-definite
	\item If $M$ is unitary, then $R(M)$ is unitary
	\item If $M$ is Hermitian then $R(M)$ is Hermitian and $\tr{R(M)} = 2 \tr{M}$
\end{enumerate}
\end{lemma}

\begin{proof}
Items 1 to 3 are easily verified.  We prove item 6 first.  If $M$ is unitary then $M M^{\dagger} = I$.  Applying item 3 we find 
\begin{equation}
R(M) R(M)^{\dagger} = R(I) = I \otimes I
\end{equation}
so $R(M)$ is also unitary.

We now prove item 4.  The remaining items follow immediately.  Diagonalize $M$ as $UMU^{\dagger} = D$.  Then 
\begin{equation}
R(U) R(M) R(U)^{\dagger} = R(D)
\end{equation}
Since $R(U)$ is unitary the eigenvalues of $R(M)$ are the same as the eigenvalues of $R(D)$.  $R(D)$ is $2 \times 2$-block diagonal, with the $k$th block equal to $R(\lambda_{k})$.  We may calculate the eigenvalues of $R(\lambda_{k})$ directly, using the characteristic equation
\begin{equation}
(\lambda_{k}^{R} - \lambda)^{2} + (\lambda_{k}^{I})^{2} = 0
\end{equation}
which has roots $\lambda = \lambda_{k}, \lambda_{k}^{*}$.  Hence these are the eigenvalue of $R(M)$.
\end{proof}

\subsection{Real simulation states and measurements}

The map $R$ lets us, at least formally, describe everything quantum using matrices instead of complex numbers.  We now show that we can modify $R$ slightly so that it maps quantum states to quantum states, quantum evolution (both discrete and continuous) to valid quantum evolution on the simulation states, and quantum measurements to valid measurements that output the correct statistics from the simulation states.

\subsubsection{Density Matrices}
Let $\rho$ be a density matrix.  Then $\rho^{\prime} = R(\rho)$ is positive.  We now come across a small problem, since the trace of $\rho^{\prime}$ is 2 instead of 1.  We can easily deal with this by multiplying $\rho^{\prime}$ by $1/2$.  Below we will show that this is necessary in order to obtain the proper outcome statistics for measurements.

\subsubsection{POVMs}
Let $\{P_{k}\}$ be a collection of positive matrices with $\sum_{k} P_{k} = I_{\mathcal{H}}$.  Then the matrices $\{ P^{\prime}_{k} = R(P_{k})\}$ are all positive and
 \[
\sum_{k} P^{\prime}_{k} = \sum_{k}R(P_{k}) =  R(I_{\mathcal{H}}) = I_{\mathcal{H}_{2}} \otimes I_{\mathcal{H}}.
\]
Thus $\{P^{\prime}_{k}\}$ is a valid POVM.

\subsubsection{Measurement statistics}
We will verify the measurement statistics for POVMs, since any measurement can be expressed as such.  Let $\rho$ be a density matrix and  $\{P_{k}\}$ be a POVM.  As previously discussed, our simulation state will be given by the density matrix $\rho^{\prime} = R(\rho) /2$ and the simulation POVM will be given by $\{ P^{\prime}_{k} = R(P_{k})\}$.  In the reference system the probability of outcome $k$ is given by
 \[
\tr{\rho P_{k}} = p_{k}
\]
Applying $R$, we obtain
 \[
p_{k} = \tr{\rho P_{k}} = \frac{1}{2}\tr{ R(\rho) P^{\prime}_{k}} =  \tr{ \rho^{\prime}P^{\prime}_{k}}.
\]

\subsubsection{Pure states}
We began our discussion of real simulations by giving a transformation that takes pure states to pure simulation states.  However, our transformation $R$ does not have this property, since $R$ takes vectors to matrices with two columns.  Nevertheless we can define a pure simulation state for each pure reference state.

Let $\ket{\psi}$ be a pure reference state.  We then have
$\braket{\psi}{\psi} = 1$.  Transforming by $R$ we obtain
 \[
 R(\ket{\psi})^{T} R(\ket{\psi}) = I_{\mathcal{H}_{2}}
\]
Thus the two columns of $R(\ket{\psi})$, which we will denote $u$ and $v$, are orthogonal and norm 1 and each represents a valid pure state.  In fact, these two vectors span the 1 eigenspace of $R(\proj{\psi}{\psi})$.  Now consider a reference POVM element $P_{k}$.  The probability of outcome $k$ is $\bra{\psi}P_{k}\ket{\psi} = p_{k}$.  Transforming by $R$ we obtain
 \[
\tr[\mathcal{H}]{R(\ket{\psi})^{T}R(P_{k})R(\ket{\psi})} = p_{k}I_{\mathcal{H}_{2}}
\]
This tells us that 
 \[
 u^{\dagger}R(P_{k})u = p_{k}, v^{\dagger}R(P_{k}) v = p_{k}
\]
so the two pure states $u$ and $v$ give the same probability for outcome $k$ as the reference experiment.  Now if we want our simulation to use a pure state instead of the rank 2 density matrix, we can choose either $u$, $v$, or in fact any normalized linear combination of the two and still obtain the correct outcome statistics.

This curious splitting of a pure state into two pure simulation states has a nice interpretation.   We return briefly to our original pure simulation state derived from the reference state given by
 \[
 \sum_{x} (a^{R}_{x} + i a^{I}_{x}) \ket{x} \rightarrow \sum_{x} (a^{R}_{x} \ket{0} + a^{I}_{x} \ket{1}) \ket{x}.
\]
We now multiply our reference state by the global phase $i$ and apply the transformation, giving the state
 \[
 \sum_{x} \left(ia^{R}_{x} - i a^{I}_{x}\right) \ket{x} \rightarrow \sum_{x} \left(a^{R}_{x} \ket{1} - a^{I}_{x} \ket{0}\right) \ket{x}.
\]
A quick calculation shows that this state is orthogonal to the previous state.  In fact, we can multiply by any global phase and obtain a linear combination of these two states which are exactly the two columns of $R(\ket{\psi})$.  This means that we can consider the ambiguity in this 2-dimensional subspace to be equivalent to the ambiguity in assigning a global phase.

\subsection{State evolution}\label{sec:real_sim_state_evolution}
We have already shown that $R(\cdot)$ maps unitary matrices to unitary matrices.  We thus turn our attention to completely positive maps and continuous time evolution.

\subsubsection{Completely positive maps}
Let $\Phi(\rho) = \sum_{k} M_{k}\rho M_{k}^{\dagger}$ be a completely positive trace preserving map.  Then
 \[
\sum_{k}M^{\dagger}_{k}M_{k} = I.
\]
Let $\Phi^{\prime}(\rho) = \sum_{k} R(M_{k}) \rho R(M_{k})^{\dagger}$.  Then $\Phi^{\prime}$ is completely positive, by its form.  It is also trace preserving, since
 \[
\sum_{k}R(M_{k})^{\dagger}R(M_{k}) = R\left(\sum_{k}M^{\dagger}_{k}M_{k} \right) =  R(I_{\mathcal{H}}) = I_{\mathcal{H}_{2}} \otimes I_{\mathcal{H}} 
\]
where $I_{\mathcal{H}_{2}}$ is the identity operating on the qubit added by $R$.

\subsubsection{Hamiltonians}
For Hamiltonians, as with density matrices, we must make a small departure.  Recall that a quantum state evolves according to
 \[
 \ket{\psi(t)} = e^{-iHt} \ket{\psi(t_{0})}.
\]
In order to maintain proper normalization, $e^{-iHt}$ must be unitary.  We ensure this by requiring that $H$ be Hermitian.  Thus $-iH$ has imaginary eigenvalues and $e^{-iHt}$ will have eigenvalues which have absolute value $1$.

We now apply $R(-iH) = R(-i) R(H)$ to obtain
 \[
\ket{\psi^{\prime}(t)} = e^{-R(i) R(H)t} \ket{\psi^{\prime}(t_{0})}.
\]
Here we have replaced $R(t) = tI$ with $t$.  When considering everything over the field $R(\mathbb{C})$ the value $R(-i) = XZ$ is a scalar.  Here we can replace this with the matrix $-XZ \otimes I_{\mathcal{H}}$ to obtain
 \[
\ket{\psi^{\prime}(t)} = e^{-XZ \otimes I_{\mathcal{H}} R(H)t} \ket{\psi^{\prime}(t_{0})}.
\]
By the properties of $R$ we see that $-XZ \otimes I_{\mathcal{H}} R(H)$ will have imaginary eigenvalues and thus $e^{-XZ \otimes I_{\mathcal{H}} R(H)t}$ will be unitary.  We can also use the Taylor expansion and see that each term has only real entries.  Thus $e^{-XZ \otimes I_{\mathcal{H}} R(H)t}$ will have only real entries for each $t$.

\subsubsection{Choi-Jamiolkowski representation}
Interestingly, despite the success so far, $R$ does not correctly transform the Choi-Jamiolkowski representation $J(\cdot)$ of a superoperator, which we discuss in section~\ref{sec:selftestingtechnicalbackground}.  The salient feature is that $J(\cdot)$ is rank 1 if the superoperator is unitary.  Since $R(J(U))$ for some unitary $U$ will have double the rank of $J(U)$, $R(J(U))$ does not represent a unitary operation despite the fact that $R(U)$ is unitary.  Thus $R(J(U)) \neq J(R(U))$.  This is analogous to the fact that pure density matrices are mapped by $R$ to rank 2 matrices.

\subsubsection{Other encoded operations}
A curious result of the real simulation is that it allows us to perform encoded operations which are not possible on the reference system.  For example, we can apply $X$ to the extra qubit and effect a complex conjugation!  In this way we can apply any anti-unitary.  We can also perform global phase rotations and measurements by suitably manipulating the extra qubit.  So even though we cannot apply these operations to the reference system, we can simulate what the outcome would be if we could.

\subsection{Real simulations and locality}
The simulation presented in the previous section deals with single systems only.  If we are dealing with a multi-part system, then $R$ will not generally map a local operation to a local operation.  This is because we added a qubit in order to store the phase information and this qubit must be stored somewhere.  Any operation involving complex numbers will be mapped to an operation that acts non-trivially on this extra qubit.

We now present a solution to this problem.  We suppose that the reference system has $k$ subsystems.  The real simulation will consist of $k$ subsystems, each of which corresponds to a reference subsystem with one added qubit.

A simple idea would be to add an extra qubit to each subsystem and perform the real simulation as though each subsystem were isolated.  In this situation, a local operation on a particular subsystem would be mapped using $R$ as in the previous section to an operation on the subsystem combined with its extra qubit.  This idea quickly fails since the state would be free to move about the entire $2^{k}$ dimensional space of the $k$ extra qubits.  The isomorphism with $\mathbb{C}$ no longer works because it only has 2 dimensions as a real vector space.

We can fix the naive solution above by constraining the extra qubits to be in a suitable 2-dimensional subspace.  It turns out that we can do this merely by choosing a suitable initial state for the qubits.  The operators will be as described in the naive solution.

Consider a phase change of $i$ applied to one of the subsystems.  We require that this phase change combined another phase change of $i$ to a different subsystem results in an overall phase change of $-1$.  This means that we need the state to be a $+1$ eigenvector of $-(XZ)_{m} (XZ)_{n}$ where $m$ and $n$ are two different subsystems. Each of these operators can be generated using operators of the form $-(XZ)_{1}(XZ)_{m}$ for different $m$, so there are $k-1$ independent operators.  Thus there is a 2-dimensional subspace stabilized by these operators.  The space is spanned by the vectors
 \[
\ket{\overline{0}} = \frac{1}{\sqrt{2^{k}}}\sum_{h(x) \text{even}} (-1)^{\frac{h(x)}{2}} \ket{x}
\]
 \[
\ket{\overline{1}} = \frac{1}{\sqrt{2^{k}}} \sum_{h(x) \text{odd}} (-1)^{\frac{h(x) - 1}{2}}\ket{x}
\]
where $x$ ranges over all $k$ bit strings.

With a bit of accounting it can be seen that if we apply $XZ$ to any one of the $k$ qubits, the effect is to take $\ket{\overline{0}}$ to $\ket{\overline{1}}$ and $\ket{\overline{1}}$ to $-\ket{\overline{0}}$.  Thus these non-local states, together with the local $XZ$ operations, behave exactly as a qubit with its $XZ$ operation.  We can now create local simulation operators from local reference operations by replacing $(XZ)_{m}$ for $XZ$ in our definitions of $R$.

\section{General simulations}\label{sec:generalsimulations}
The real simulation revealed that there exist non-unitarily equivalent simulations, but is this the only non-unitarily equivalent simulation?  The answer to this is no.  In this section we will develop a large family of simulations.  As we shall see, in section~\ref{sec:extmayersyao}, there exist experiments for which the general simulations described here (along with the unitarily equivalent simulations) are all possible simulations.

Interestingly, if we confine ourself to experiments on a real Hilbert space, then the general simulations all collapse to unitarily equivalent simulations.  This allows the self-testing Theorems to be successful for such experiments, as we shall see in chapter~\ref{chap:selftesting}.

\subsection{States and measurements}

Consider a reference state $\ket{\psi}$ measured according to a reference POVM $\{ P_{k}\}$.\footnote{We may consider mixed states as well, but it is not necessary for our discussion since we may consider the purification of a mixed state.}  We may duplicate the statistics of this experiment using the complex conjugate state $\ket{\psi^{*}}$ and POVM $\{P_{k}^{*}\}$.  In addition, we could do some combination of the two.  We may add an additional qubit register which records which of the two experiments to perform:  $\ket{0}$ for the reference experiment, and $\ket{1}$ for the complex conjugate.  This qubit may be in any state, and not necessarily pure.  We then arrive at new state
\begin{equation}\label{eq:contsimstate}
\rho^{\prime} = a \proj{0}{0} \otimes \proj{\psi}{\psi} + (1-a) \proj{1}{1} \otimes \proj{\psi^{*}}{\psi^{*}} + c \proj{0}{1}\otimes \proj{\psi}{\psi^{*}} + c^{*} \proj{1}{0} \otimes \proj{\psi^{*}}{\psi}
\end{equation}
with $a \geq 0$ and $|c| \leq \sqrt{a(1-a)}$.  The important feature is that when we project onto $\proj{0}{0}$ or $\proj{1}{1}$ we get either $\ket{\psi}$ or $\ket{\psi^{*}}$, respectively.  For the measurement, we form the POVM
 \begin{equation}\label{eq:cont_sim_povm}
\left\{ \proj{0}{0} \otimes P_{k} + \proj{1}{1} \otimes P^{*}_{k}\right\}.
\end{equation}
This POVM measurement is equivalent to measuring the added qubit, collapsing the state into either $\ket{\psi}$ or $\ket{\psi^{*}}$ and then measuring either $\{P_{k}\}$ or $\{P_{k}^{*}\}$ as appropriate; thus the statistics of the experiment are preserved.

\subsection{Operators}
We can extend the measurement operator defined in \ref{eq:cont_sim_povm} to arbitrary operators.  We define
\begin{equation}\label{eq:cont_sim_op}
C(M) =  \proj{0}{0} \otimes M + \proj{1}{1} \otimes M^{*}.
\end{equation}
Note that $C(M)$ can be expressed differently as
\begin{equation}\label{eq:cont_sim_op2}
C(M) = I \otimes Re(M) + iZ \otimes Im(M)
\end{equation}
where $Re(M)$ and $Im(M)$ are the real and imaginary parts of $M$ (both real matrices).

We summarize some of the properties of $C(M)$ here

\begin{lemma}
Let $M$ and $N$ be matrices.  Then we have the following:
\begin{enumerate}
	\item $C(MN) = C(M)C(N)$.
	\item $C(M + N) = C(M) + C(N)$.
	\item Let $a$ be a real number, then $C(aM) = aC(M)$.
	\item If $\ket{\psi}$ is an eigenvector of $M$ (normal) with eigenvalue $\lambda$, then $\ket{0}\ket{\psi}$ and $\ket{1}\ket{\psi^{*}}$ are eigenvectors of $C(M)$ with eigenvalues $\lambda$ and $\lambda^{*}$, respectively.
	\item $C(M)$ is Hermitian if and only if $M$ is.
	\item $C(M)$ is unitary if and only if $M$ is.
	\item $C(M)$ is positive semi-definite if and only if $M$ is.
	\item When $M$ is Hermitian, $\tr{C(M)} = 2 \tr{M}$.
\end{enumerate}
\end{lemma}

These properties can be derived easily.  In fact, $R(\cdot)$ and $C(\cdot)$ are related by a unitary transformation, as will be seen in section~\ref{sec:real_sim_in_cont_sim}.

\subsubsection{Discrete time evolution}
The properties of $C(\cdot)$ allow us to easily determine how the simulation states in the continuum evolve.  Let $U$ and $\ket{\psi}$ be a reference unitary operation and state and let $\rho^{\prime}$ be as in equation~\ref{eq:contsimstate}.  By the form of $C(U)$ we have

 \[
C(U)\rho^{\prime}C(U)^{\dagger} = 
a \proj{0}{0} \otimes U\proj{\psi}{\psi}U^{\dagger} + (1-a) \proj{1}{1} \otimes U^{*}\proj{\psi^{*}}{\psi^{*}}U^{T} +\]
\[ c \proj{0}{1}\otimes U\proj{\psi}{\psi^{*}}U^{T} + c^{*} \proj{1}{0} \otimes U^{*}\proj{\psi^{*}}{\psi}U^{\dagger}
 .
\]
But this is the simulation state for $U\ket{\psi}$, and hence $C(U)$ evolves the simulation state $\ket{\psi^{\prime}}$ to produce a new simulation state corresponding to $U\ket{\psi}$.  Compositions of unitaries will also evolve the state correctly so that the measurement statistics at the end of a circuit will be identical to that of the reference circuit.

General quantum operations may be mapped similarly.  It is easy to verify that in Krauss representation a completely positive map is mapped to a completely positive map if we apply $C(\cdot)$ to each of the Kraus operators.  The trace preserving property is also preserved.  Applying the properties of $C(\cdot)$ we see that
\begin{equation}
\frac{1}{2} \sum_{j} C(K_{j}) C(\rho) C(K_{j})^{\dagger} = \frac{1}{2}C\left(\sum_{j} K_{j} \rho K^{\dagger}_{j} \right).
\end{equation}

\subsubsection{Continuous time evolution}
We begin with a Hamiltonian $H$.  Instead of applying $C(\cdot)$ to obtain the simulation Hamiltonian, we make a small modification, introducing a $-1$ on the complex conjugated part, thus
\begin{equation}
H^{\prime} = \proj{0}{0} \otimes H - \proj{1}{1} \otimes H^{*}.
\end{equation}
This reflects the fact that complex conjugating a Hamiltonian corresponds to time reversal.  The evolution of the state will be according to the Schr\"{o}dinger equation
\begin{equation}
U(t) = e^{-iH^{\prime}t}
\end{equation}
We may take advantage of a property of the exponential function, namely $\exp(A + B) = \exp(A) +  \exp(B) - I$ when $AB = 0 = BA$, which may be easily verified by examining the Taylor expansion.  Since $\proj{0}{0} \proj{1}{1} = 0$ we obtain
\begin{equation}
e^{-iH^{\prime}t} = e^{-i \proj{0}{0} \otimes H t} + e^{i \proj{1}{1} \otimes H^{*}t} - I.
\end{equation}
We now use another property of the exponential function:  $\exp\left(P \otimes A\right) = P \otimes \exp( A) - P \otimes I + I$ when $P^{2} = P$.  We obtain
\begin{equation}
e^{-iH^{\prime}t}Ê= \proj{0}{0} \otimes e^{-iHt}  + \proj{1}{1} \otimes e^{iH^{*}t}.
\end{equation}
Finally, we use the Taylor expansion and the fact that $a^{*}b^{*} = (ab)^{*}$ to see
\begin{equation}
e^{iH^{*}t} = (e^{(iH^{*}t)^{*}})^{*} = \left(e^{-iHt}\right)^{*}.
\end{equation}
Thus
\begin{equation}
e^{-iH^{\prime}t} = C(e^{-iHt})
\end{equation}
and the simulation evolution tracks that of the reference system.

Another way to think of this is similar to that used in the real simulation in section~\ref{sec:real_simulation}.  There, rather than considering the Hamiltonian alone, the whole matrix in the exponent, $-iHt$, was considered.  Applying $C(\cdot)$ to this matrix we obtain
\begin{equation}
\proj{0}{0} \otimes (-iHt) + \proj{1}{1} \otimes (-iHt)^{*} =  i\left(\proj{0}{0} \otimes H - \proj{1}{1} \otimes H^{*} \right)t. 
\end{equation}
Here the fact that $a^{*}b^{*} = (ab)^{*}$ means $(-iH)^{*} = iH^{*}$ and the $-1$ factor is explained.  

\subsection{Non-local computations}
The above family of simulations suffer from the same problem as the real simulation:  every operator must have access to the extra qubit added in the simulation.  This makes the simulations unsuitable for multi-party computations.  This problem can be solved in the same manner as for the real simulation.  There an extra qubit was added for each party in the computation and an entangled state was prepared across these qubits.  New operators were then defined that interact with these qubits locally in order to have the same effect as the single qubit in the original real simulation.

In order to perform the simulation correctly, each party needs to know whether to apply the original operator or the complex conjugate.  For classical mixtures this is no problem, since this information can be stored in a local classical bit for each party.  For quantum mixtures, the extra qubit is encoded in a logical qubit which is stored across a number of qubits, each located with a different party.  More specifically, each party has a qubit and the combined state is prepared in the GHZ-like state
\begin{equation}\label{eq:cont_sim_multiparty_state}
\alpha\ket{00\dots0}\ket{\psi} + \beta\ket{11\dots1}\ket{\psi^{\prime}}.
\end{equation}
Each party can perform a desired unitary by applying the operator specified in equation \ref{eq:cont_sim_op2} with the $Z$ operator acting on their local qubit.  We see that the $Z$ operator introduces a $-1$ phase on the imaginary part of the operator when the local qubit is in the state $1$, thus applying the complex conjugate of the operator.

\subsection{Real simulation in the continuum}\label{sec:real_sim_in_cont_sim}

The real simulation (section~\ref{sec:real_simulation})  can be expressed as a simulation in the family defined above through a change of basis.  Starting with the state defined for $\alpha = \beta = \frac{1}{\sqrt{2}}$ we have
 \[
\ket{\psi^{\prime}} = \frac{1}{\sqrt{2}} \ket{0}\ket{\psi} + \frac{1}{\sqrt{2}} \ket{1}\ket{\psi^{*}}.
\]
We next apply a Hadamard gate followed by the relative phase rotation
 \[
\left(
	\begin{matrix}
	1 & 0 \\
	0 & -i \\
	\end{matrix}
\right)
\]
to the extra qubit.  This is the same as applying the unitary
\begin{equation}
U = \left(
	\begin{matrix}
	1  & 1 \\
	-i & i \\
	\end{matrix}
\right).
\end{equation}
The resulting state is
 \[
\frac{1}{2}\ket{0}(\ket{\psi} + \ket{\psi^{*}}) - \frac{i}{2}\ket{1}(\ket{\psi} - \ket{\psi^{*}})
\]
which can be rewritten as
 \[
\ket{0}Re(\ket{\psi}) + \ket{1}Im(\ket{\psi})
\]
and the real simulation is recovered.

Operators are transformed quite easily.  For operator $M$ we conjugate $C(M)$ by $U \otimes I$.  From \ref{eq:cont_sim_op2} we see that the resulting operator is 
\begin{equation}
(U \otimes I)\, C(M)\, (U^{\dagger} \otimes I) = I \otimes Re(M) + XZ \otimes Im(M)
\end{equation}
which is exactly the operator used in the real simulation for $M$.

For the multiparty real simulation the extra qubits added for each party must also be transformed correctly.  We start with the state
 \[
\frac{1}{\sqrt{2}}\ket{00\dots0}\ket{\psi} + \frac{1}{\sqrt{2}}\ket{11\dots1}\ket{\psi^{*}}.
\]
We now apply a Hadamard gate to \emph{each} extra qubit, resulting in the state
 \[
\frac{1}{\sqrt{2^{k+1}}} \sum_{x} \ket{x}\ket{\psi} + \frac{1}{\sqrt{2^{k+1}}} \sum_{x} (-1)^{h(x)} \ket{x}\ket{\psi^{*}},
\]
where $h(x)$ is the Hamming weight (number of 1s) of the bit string $x$ and $k$ is the number of parties.  Collecting terms gives
 \[
\frac{1}{\sqrt{2^{k+1}}}\sum_{h(x)\, \text{even}}\ket{x}(\ket{\psi} + \ket{\psi^{*}}) + \sum_{h(x)\, \text{odd}}\ket{x}(\ket{\psi} - \ket{\psi^{*}}).
\]
We now apply the relative phase rotation as above to each extra qubit, resulting in the state
\[
\frac{1}{\sqrt{2^{k+1}}}\sum_{h(x)\, \text{even}}(-1)^{\frac{h(x)}{2}}\ket{x}(Re(\ket{\psi}) + \sum_{h(x)\, \text{odd}}(-1)^{\frac{h(x)-1}{2}}\ket{x}Im(\ket{\psi})
\]
which is the correct state for the multi-party real simulation.

Another way to see this latter transformation is in terms of stabilizers.  In \cite{McKague:2009:Simulating-Quan} it is noted that the entangled states used in the simulation are stabilized by $Y_{s}\otimes Y_{t}$ for distinct $s,t$.  Also note that the states used in the simulations defined here are stabilized by $Z_{s}\otimes Z_{t}$ for distinct $s,t$.  The qubit-wise transformation applied transforms $Z$ into $Y$.  Collecting terms for the real and imaginary parts completes the transformation.

\section{Simulations in a cryptographic setting}\label{sec:selftestcrypto}

Suppose that two or more parties are engaged in a cryptographic protocol using self-tested apparatus.  Later, in section~\ref{sec:extmayersyao}, we will develop the extended Mayers-Yao test which allows them to determine that the devices are implementing a simulation from the family of simulations described in section~\ref{sec:generalsimulations}.  Suppose further that the adversary, Eve, knows how the devices are implemented (she provides them) and controls the preparation of the state.  The honest parties only perform operations as specified for the simulation.  Eve, on the other hand, is free to interact with the extra qubits in the simulation in any way she likes.  Does this give any advantage to Eve?

Eve can potentially perform many operations, including entangling a qubit of her own with the extra simulation qubits allowing her to perform simulation operations.  She may also interact in complex ways with the extra simulation qubits along with the original register, including performing encoded anti-unitary operations.  Despite this, we are able to prove that Eve can gain no advantage for some protocols.

We explore a restricted class of protocols that are especially easy to analyze.  These are protocols where the only operation that an honest party will do is a Pauli measurement.  This class includes the six-state quantum key distribution protocol (implemented as an entanglement based protocol) (see \cite{Bennett:1984:Eavesdrop-detec}, \cite{Bruss:1998:Optimal-Eavesdr}).  Briefly, the 6-state QKD protocol uses three measurements, $X$, $Y$, and $Z$ instead of only $X$ and $Z$ as in BB84.  We will demonstrate that these protocols do not leak any more information when implemented using one of the simulations.

The proof is a series of security reductions to protocols in which each reduction only increases Eve's power.  We will show that the final protocol in the reduction is just as secure as the reference protocol (without the simulation applied), hence the simulation protocol is also just as secure as the reference protocol.

For the first reduction we suppose that the participants in the protocol measure their simulation qubit in the $Z$ eigenbasis after the protocol is completed, and transmit the result to Eve.  This does not interfere with the intended protocol and only increases Eve's information.  Since the $Z$ measurement commutes with all simulation operations, the participants could just as well have performed the measurement before the protocol began.  If Eve is the one who prepares the initial state for the simulation (in other cases Eve has strictly less power) then Eve could also perform this measurement herself.  This measurement would collapse the state to an eigenvector of the $Z$ measurements and Eve's strategy would be a mixture of different strategies with the states each an eigenvector of the $Z$ measurements.

Let us examine the result of Eve choosing one of these eigenvector states.  Each of the parties will receive their extra qubit prepared in a $Z$ eigenvector.  The effect of this on their operations is either to perform the protocol's original operation (in the case of a $\ket{0}$) or the complex conjugate (in the case of a $\ket{1}$.)  For Pauli measurements, only the $Y$ measurement is affected: the output bit is flipped in the case of the complex conjugate.

If every party receives the same eigenvector in their extra qubit, then the protocol reduces to either the original or the complex conjugate.  In either case the security is identical to the original protocol.  If the extra qubits are not in the same eigenvector then some $Y$ measurements outcomes will be flipped and some will not.  This does not affect Eve's information since she controls which outcomes are flipped and can undo the flips in her reckoning of the final classical information.  Note that the bit flips may introduce errors into the protocol.  If the protocol does not explicitly check for such errors (as does the 6-state protocol) a test for these errors may be required, however the lack of such a test does not leak information to Eve.  Thus the reduced protocol is as secure as the reference protocol, and so is the simulated protocol.

\section{Quantum formalism over other rings}
In the previous sections we have shown that quantum physics over the complex numbers is indistinguishable from quantum physics over the real numbers.  That is to say, there is no experiment that can be designed where the predictions of complex quantum physics are different from those of real quantum physics.  The natural question to ask is whether the same is true if we move to quantum physics over the \index{quaternions}quaternions or other division rings.

\subsection{Why other rings?}

Although the amplitudes for quantum states are usually taken to be complex numbers, we might imagine other objects could be used.  Of course these objects must have certain properties in order to be suitable.  We will argue that the only objects that can be used are the real numbers, complex numbers, or quaternions.  The argument follows that of Adler \cite{Adler:1995:Quaternionic-Qu}.

Suppose that the amplitudes are drawn from a set $V$.  Adler argues that, since all measurable quantities will reduce to real numbers, we may take $V$ as a general division algebra over the real numbers.  In order to define normalized states, and to use the square rule for deriving probabilities we must have a modulus function, $N$, which assigns a non-negative ``size'' to each element.  That is, $N : V \mapsto \mathbb{R}^{+}$.  We would like $N$ to have the properties
\begin{enumerate}
	\item $N(r\phi) = |r|N(\phi)$
	\item $N(\phi + \theta) \leq N(\phi) + N(\theta)$
	\item $N(\phi) > 0$ if $\phi \neq 0$.
\end{enumerate}
These properties are required for $N$ to be a norm on $V$.  Adler also imposes one more constraint.  We begin by considering the inner product between two states, $\braket{a}{b} \in V$.  Inserting a resolution of the identity, we obtain
\begin{equation}
\braket{a}{b} = \sum_{c}\braket{a}{c}\braket{c}{b}
\end{equation}
where $\ket{c}$ runs over some basis (not necessarily including $\ket{a}$ or $\ket{b}$).  If there is only one non-zero term in the sum then we obtain
\begin{equation}
\braket{a}{b} = \braket{a}{c}\braket{c}{b}.
\end{equation}
Adler now imposes the restriction
\begin{equation}
N(\braket{a}{b})^{2} = N(\braket{a}{c})^{2}N(\braket{c}{b})^{2}
\end{equation}
which recovers classical probability superposition.  Taking the square root everywhere, we obtain the property
\begin{equation}
N(\phi\theta) = N(\phi)N(\theta).
\end{equation}

We now apply a result due to Albert (\cite{Albert:1947:Absolute-Valued}):  the only division algebras over the reals for which an $N$ exists with the properties above are the real numbers, the complex numbers, and the quaternions.  This means that the only number systems that can underly the quantum formalism are these three.  (Unless, of course, we modify the formalism in some important way.)

We have established a short list of number systems to consider.  Why should we consider anything other than the complex numbers?  From a foundations perspective we would like to eliminate all possibilities that are not the usual quantum formalism.  If we can argue that using something other than the complex numbers would result in a theory that is not viable for some reason, then we do not need to assume the use of the complex numbers as an axiom.  A particularly attractive outcome would be an experiment which eliminates other number systems.

\subsection{Bell inequalities and real quantum formalism}

In \cite{Gisin:2007:Bell-inequaliti}, Gisin notes that all known Bell inequalities can be maximally violated using states and measurements described using only real numbers.  Gisin then asks whether this is true for all possible Bell inequalities, or whether there exists a Bell inequality for which a higher violation can be achieved if complex numbers are used.  If the latter is the case then such an inequality, and the accompanying physical experiment, could experimentally prove that complex numbers are required.  P\'{a}l and V\'{e}rtesi prove in \cite{Pal:2008:Efficiency-of-h} that the former is true.  The results in this section provide a different and independent proof, since their construction is different.  Moreover, the present work applies much more generally.  Indeed, there is no experiment whatsoever whose outcomes cannot be duplicated by another experiment described using only real numbers.

More generally, the real simulation shows that quantum formalism over the real numbers is the same as over the complex numbers in the following sense:  there are no experimental outcomes predicted by complex quantum formalism that cannot be predicted by real quantum formalism.  The reference experiment in the complex formalism translates into the real formalism via the real simulation, giving the same predictions.  From a foundations perspective, then, the distinction is not problematic.  We can use either number system.  Of course there may be more practical reasons for preferring one over the other.  With the real numbers, calculations may be more cumbersome, since a larger Hilbert space will in general be necessary.  Also, many operations are permitted by the real quantum formalism, such as an encoded complex conjugation, that do not seem to occur in nature, so the complex formalism may offer a better fit\footnote{Thanks to Bill Wootters for pointing this out.}.  This need not be a problem, however, as the complex formalism also allows operations, such as energy non-conserving operations, that do not occur in nature as well.

\subsection{Quaternionic quantum physics}

\subsubsection{Quaternions}

The Quaternions ($\mathbb{H}$), first described by William Hamilton \cite{Hamilton:1844:On-quaternions-}, are a division ring formed by adjoining new elements, $i$, $j$, and $k$ to the real numbers $\mathbb{R}$.  Thus a quaternion looks like
\begin{equation}
q = a + ib + jc + kd.
\end{equation}
The new elements have the properties
\begin{equation}
i^{2} = j^{2} = k^{2} = ijk = -1.
\end{equation}
The multiplication of the elements $i, j, k$ is summarized in figure~\ref{fig:quatmult}.  When multiplying two elements along an arrow (eg. $ij$) the result is the third element in the cycle.  When multiplying backwards along an arrow (eg. $ji$) a $-1$ factor is added.  So $ij = k$ and $ji = -k$.  The elements $1$ and $-1$, of course, commute with the other elements.

 \begin{figure}
\[ \xymatrix {
 & i \ar@(r,ur)[ddr]& \\
  & & \\
 k \ar@(ul,l)[uur] & & \ar@(dl, dr)[ll]j \\}
\]
\caption{Multiplication in the quaternion group}
\label{fig:quatmult}
\end{figure}
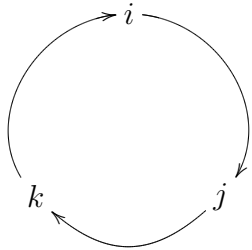

Just as complex numbers have a real and imaginary part, quaternions have a scalar and vector part.  The scalar part is the part which lies on the real axis.  We denote it by $Re(q)$.  The vector part, also called the pure imaginary part, is everything else, and in general is a vector in $\mathbb{R}^{3}$.  We denote it by $Im(q)$.  The scalar and vector parts of $q$, defined above, are $R(q) = a$ and $Im(q) = ib + jc + kd$.

Like in the complex numbers, we may define  the \emph{conjugate} of a quaternion, which multiplies each of the non-real parts by -1.  Thus the conjugate of $q$ is
\begin{equation}
q^{*} := a - ib - jc - kd.
\end{equation}
The norm on the quaternions is analogous to that for the complex numbers, i.e.
\begin{equation}
||q|| = \sqrt{q q^{*}}.
\end{equation}

The most important difference between the complex numbers and the quaternions is that the quaternions do not form a commutative algebra.  This property will be the basis for the rest of our discussion.

\subsubsection{Quaternionic quantum mechanics}
Quaternionic quantum mechanics is formed, roughly speaking, by replacing every complex number in the usual quantum mechanics by a quaternion.  Thus states are vectors over the quaternions, so amplitudes are now quaternions instead of complex numbers.  The usual norm-squared rule applies for deriving outcome probabilities, and discrete time evolution is described by unitary matrices $U$ over the quaternions with the usual property $U U^{\dagger} = I$.  Now the Hermitian conjugation ${(\cdot)}^{\dagger}$ is the matrix transpose, followed by quaternionic conjugation.  

Although many more aspects of quantum mechanics, such as continuous time evolution, may be considered, these few properties will suffice for our discussion.  For a comprehensive treatment of quaternionic quantum mechanics, see Stephen Adler's book \cite{Adler:1995:Quaternionic-Qu}.

\subsubsection{The tensor product problem}

The non-commutative nature of the quaternions introduces many new properties into quaternionic quantum mechanics.  The one we are most interested in here is the nature of multi-partite systems.

Consider a bipartite system in the state $\frac{1}{\sqrt{2}}\left(\ket{00} + \ket{11} \right)$.  Define the unitary matrices $R_{i}$ and $R_{j}$ as
\begin{equation}
R_{i} = \left(
	\begin{matrix}
	1 & 0 \\ 0 & i \\
	\end{matrix}
\right)
\end{equation}
\begin{equation}
R_{j} = \left(
	\begin{matrix}
	1 & 0 \\ 0 & j \\
	\end{matrix}
\right)
\end{equation}
We consider different ways that we may apply these matrices.  First we apply $R_{i}$ to the first subsystem, obtaining  $\frac{1}{\sqrt{2}}\left(\ket{00} + i\ket{11} \right)$.  Next we apply $R_{j}$ to the second subsytem, obtaining  
\begin{equation}\label{eq:stateij}
\frac{1}{\sqrt{2}}\left(\ket{00} - k\ket{11} \right).
\end{equation}  
Now consider the same operations, but applied in the reverse order.   We apply $R_{j}$ to the second subsystem, obtaining $\frac{1}{\sqrt{2}}\left(\ket{00} + j\ket{11} \right)$, followed by $R_{i}$ applied to the first subsystem, obtaining
\begin{equation}\label{eq:stateji}
\frac{1}{\sqrt{2}}\left(\ket{00} + k\ket{11} \right).
\end{equation}
Here we see the non-commutativity of $\mathbb{H}$ in action.  The two states in equations~\ref{eq:stateij} and~\ref{eq:stateji} are orthogonal, but all we have changed is the time-ordering of two local operations on separate subsystems.

The above problem may be stated as follows:  $R_{i} \otimes I$ and $I \otimes R_{j}$ do not commute.  This extends to the tensor product problem:  How do we define $R_{i} \otimes R_{j}$?  Evidently the evolution of subsystems cannot be considered without considering the context of the system as a whole.  Adler \cite{Adler:1995:Quaternionic-Qu} considers the same problem in the context of continuous evolution:

\begin{quote}
We conclude, then, that in quaternionic quantum mechanics, a sum of $N \geq 2$ 
one-body Hamiltonians gives a many-body Hamiltonian that does not describe 
N independent particles; the particle motions are coupled through the 
noncommutativity of the quaternion algebra.  (Adler \cite{Adler:1995:Quaternionic-Qu}, p. 245)
\end{quote}

What does this mean for locality?  Is there such thing as a local transformation?  Is it possible for Alice and Bob to actually perform the operations $R_{i} \otimes I$ and $I \otimes R_{j}$?  The formalism does not answer this question.  However, we may address this problem in another way.  If Alice and Bob can locally perform these operations, then they can implement a non-local box.

\subsection{Quaternionic non-local boxes}
The non-local box, first defined by Popescu and Rohrlich in \cite{Popescu:1994:Quantum-nonloca}, is an imaginary device which produces non-local correlations between data in the following way:  Two distant parties, Alice and Bob, each have half of the box.  They have one bit of input each, $a$ and $b$, and input their bit into their half of the box.  Each half of the box produces one bit of output, $x$ and $y$, obeying the property
\begin{equation}
x \oplus y = ab.
\end{equation}
The content of the famous CHSH inequality \cite{Clauser:1969:Proposed-Experi} is that this condition cannot be satisfied by a non-signalling classical local hidden variable theory with probability better than $0.75$ when $a$ and $b$ are chosen uniformly at random.  Quantumly, we can do better, but are bounded above by $\cos^{2} \pi / 8 \approx 0.85$, the Cirel'son bound \cite{Cirelson:1980:Quantum-general}.

Now we consider the case of quaternionic quantum mechanics.  Evidently it has stronger non-local behaviour than complex quantum mechanics, but how strong?  Clearly we can at least achieve the Cirel'son bound since any strategy in complex quantum mechanics also exists in quaternionic quantum mechanics, but can we do better?  The answer is that we simulate the non-local box perfectly.

Consider the two parties, Alice and Bob, as before.  Before receiving their inputs they synchronize clocks and choose times $t_{1} < t_{2}< t_{3}< t_{4} < t_{5}$ such that $t_{1}$ is after they receive their inputs and $t_{5}$ is before they require the outputs (we may arrange it so that the time elapsed is too short to allow signalling by moving Alice and Bob far enough away from each other.)  They also share the state $\frac{1}{\sqrt{2}}\left(\ket{00} + k\ket{11} \right)$.

Alice does the following:

\begin{enumerate}
	\item Receive input $a$.
	\item If $a = 0$ then apply operation $R_{i}$ at time $t_{1}$.
	\item If $a = 1$ then apply operation $R_{i}$ at time $t_{3}$.
	\item At time $t_{5}$ measure in the basis $\ket{+} / \ket{-}$ and output the result as $x$.
\end{enumerate}

Meanwhile, Bob does the following:

\begin{enumerate}
	\item Receive input $b$.
	\item If $b = 0$ then apply operation $R_{j}$ at time $t_{4}$.
	\item If $b = 1$ then apply operation $R_{j}$ at time $t_{2}$.
	\item At time $t_{5}$ measure in the basis $\ket{+} / \ket{-}$ and output the result as $y$.
\end{enumerate}

The protocol is depicted in figure~\ref{fig:quaternionic_nonlocal_box}

Roughly what is happening here is that Alice applies her operation before Bob in all cases except when both of their inputs are 1.  Alice and Bob then detect this event using local measurements that are correlated except when Bob goes first, in which case they are anti-correlated.

Consider the outputs that Alice and Bob generate.  First note that the final state before measuring will be $\ket{00} \pm \ket{11}$.  If Alice and Bob both measure in the basis $\ket{+} / \ket{-}$ their outcomes will be the same if the relative phase was $+$ and opposite if the relative phase was $-$.

Suppose that Bob's input is $0$.  He will wait until $t_{4}$ before applying $R_{j}$.  Regardless of her input, Alice will apply $R_{i}$ before Bob applies $R_{j}$.  Thus the combined effect on $\ket{\psi}$ is a relative phase change of $-k$, taking the state to $\ket{\phi_{+}}$.  Then Alice and Bob's measurements will agree and $x \oplus y = 0 = ab$.

Meanwhile, if Bob's input is $1$ the situation changes.  If Alice receive the input 0 then she applies $R_{i}$ at time $t_{1}$, Bob applies $R_{j}$ at time $t_{2}$ and the situation is the same as above.  However, if Alice receives the input 1 then she applies $R_{i}$ at time $t_{3}$, \emph{after} Bob applies $R_{j}$ at time $t_{2}$.  In this case the effect on $\ket{\psi}$ is a relative phase change of $k$, taking the state to $\ket{\phi_{-}}$.  In this case Alice and Bob's measurements will disagree and $x \oplus y = 1 = ab$.

 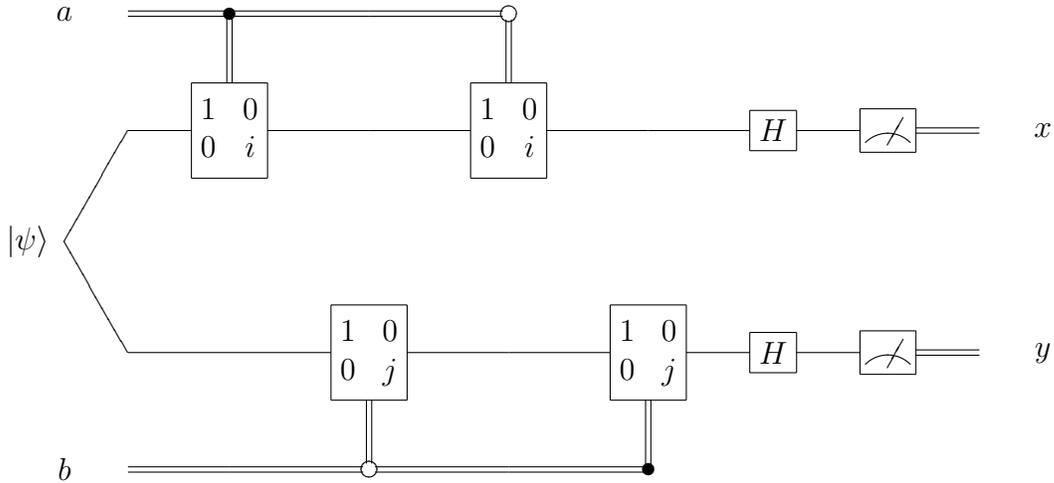
\begin{figure}
\[ \Qcircuit {
a  & &  \controlÊ\cw \cwx[1]  & \cw & \controlo \cw \cwx[1] & \\
  & & \gate{\begin{matrix}
			1 & 0 \\ 0 & i\\
			\end{matrix}}
& \qw &  \gate{ \begin{matrix}
			1 & 0 \\ 0 & i\\
			\end{matrix}}
& \qw & \gate{H} & \meter & \cw & x \\
\lstick{\ket{\psi}}   \ar@{-}[ur] \ar@{-}[dr]\\
  & & \qw & \gate{\begin{matrix}
  		1 & 0 \\ 0 & j \\
		\end{matrix}}
 & \qw & \gate{\begin{matrix}
 		1 & 0 \\ 0 & j\\
		\end{matrix}}
& \gate{H} & \meter & \cw & y \\
b & & \cw & \controlo \cw \cwx & \cw & \control \cw \cwx \\
 }
\]
\caption{Quaternionic non-local box}
\label{fig:quaternionic_nonlocal_box}
\end{figure}

\subsubsection{Communication complexity and information causality}

We now briefly consider \emph{communication complexity}.  Suppose two parties, Alice and Bob, receive two inputs, $a$ and $b$.  They wish to compute the value of some function $f(a,b)$.  How much communication is necessary between Alice and Bob? (for simplicity, we suppose that Alice receives the final answer.)  Here we are interested in boolean functions, whose output is a single classical bit.  It has been shown by van Dam \cite{Dam:2005:Implausible-Con} that the communication complexity of all Boolean functions is trivial if non-local boxes are available.  This means that Bob needs to send only one bit of information to Alice, and Alice does not have to send anything to Bob.  However, there exist Boolean functions, such as the inner product between two strings, for which the communication complexity in either a classical or quantum setting is maximal (i.e. the optimal strategy is for Bob to transmit his entire input to Alice) \cite{Brassard:2006:Limit-on-Nonloc}.  Coupled with the current result we find that within quaternionic quantum mechanics the communication complexity of all boolean functions is trivial.

Later Brassard et al. \cite{Brassard:2006:Limit-on-Nonloc} turn van Dam's result around, saying that if there is a non-trivial bound on the communication complexity of boolean functions, then non-local boxes do not exist.  They also made this result robust by introducing a notion of probabilistic communication complexity and showing that if a non-local box can be approximated with probability better than $\approx .906$ then every boolean function has trivial probabilistic communication complexity.  Linden et al \cite{Linden:2007:Quantum-Nonloca} finally showed that for a particular boolean function (AND of two 2-bit strings) if a non-local box can be approximated with probability better than $\cos^{2} \pi/8$ (the quantum upper bound) then no communication is required and the function can be approximated better than the classical (and quantum) bound of $0.75$.  Turning this result around, if the communication complexity of AND of 2-bit strings is non-trivial, then the non-local box cannot be approximated any better than what is achievable by quantum mechanics.  In particular, if there is a non-trivial bound on communication complexity then quaternionic quantum mechanics is not a viable theory.

Following this work, Pawlowski et al. \cite{Pawlowski:2009:A-new-physical-} developed the notion of \emph{information causality} which can be seen as a generalization of no-signalling.  Both classical and quantum theories obey information causality.  Pawlowski et al. were able to show that any physical theory which obeys information causality also obeys the Cirel'son bound \cite{Cirelson:1980:Quantum-general}.  Thus this is another way of excluding quaternionic quantum mechanics as a viable physical theory.

\chapter{Self-testing}\label{chap:selftesting}
\minitoc
\section{Introduction}
Self-testing is a solution to the follow problem:  perform a quantum calculation using only untrusted quantum devices and be sure that the result is correct.  For many tasks, such as factoring, this may be trivial since the result can be checked quickly using only classical computation.  However, it is not known whether all problems efficiently solvable by a quantum computer can be checked in this fashion (i.e. it is not known whether $BQP \subseteq NP$).  Thus there are potentially problems for which a classical check is inefficient.  In this chapter we present and extend the main self-testing concepts and constructions, as developed in  \cite{Mayers:2004:Self-testing-qu} and \cite{Magniez:2006:Self-testing-of}.

\section{Literature review}

Self-testing was introduced by Mayers and Yao in \cite{Mayers:1998:Quantum-Cryptog}.  The initial application was to \index{quantum key distribution}quantum key distribution.  In prepare and measure QKD protocols, it is important that the photon source used does not contain side channels that leak information on the basis choices, although this may be difficult to establish for particular physical implementations.  Mayers and Yao proposed to solve this problem by using a self-checking photon source.  The idea is that the manufacturer who provides the photon source would also provide several measurement devices.  The devices would measure the signal from the photon source and generate some statistics.  Checking the classical statistics, the participants in the QKD protocol could then verify that the photon source is operating correctly before continuing the protocol.

An important consideration in this work is that the photon source is implemented as a source of EPR pairs.  One qubit from the EPR pair is measured, and the other is sent out as a signal.  The self-check consists of establishing that the initial state is indeed an EPR pair, in which case the signal does not contain any information about the basis choice since this is only made on the other photon in the pair.  This reliance on EPR pairs is a feature that is present throughout the remaining self-testing literature.  Another important feature is that the dimension of the Hilbert space of the source is not known beforehand, allowing the result to be applicable in situations where almost nothing is known about the physical implementation of the devices.

Mayers and Yao improved their result in \cite{Mayers:2004:Self-testing-qu}.  In particular, they broadened the scope of their self-test to a generic setting of a source of EPR pairs and two sets of measurements (one for each half of the state.)  Again the result shows that the source produces EPR pairs, but in addition the measurements are also characterized.  This is remarkable since they begin without any trusted apparatus, and end up with both a characterized source and measurement devices.  

In a parallel development, van Dam et al. \cite{van-Dam:1999:Self-Testing-of} investigated self-testing in the context of quantum circuits.  In particular, they developed a series of tests in which a verifier interacts classically with a quantum apparatus by preparing states and measuring in the computational basis and establishes that a gate is operating correctly.  (The gate comes with an attached specification.)  This work extends self-testing in to the realm of general computing, but relies on some important assumptions.  In particular, the Hilbert space of the state is assumed to be known and the devices are used several times with the assumption that it operates identically each time.  Also, the computational state preparation and measurement is trusted, and finally the self-tests are only applicable to a particular set of gates (which is sufficient for universal computation.)

The two lines of research were merged in \cite{Magniez:2006:Self-testing-of} by Magniez et al..  There, the reliance on a particular Hilbert space and trusted computational bases preparation and measurement is removed.  This is done by preparing states in a variety of bases; EPR pairs are prepared and one half is measured in three different bases.  This initial step is self-tested using the Mayers-Yao result.  As well, the Mayers-Yao result is invoked to characterize the measurement devices.  Once these two steps are accomplished, the gate is the only remaining untrusted element, and it may be tested using the now trusted state preparation and measurement.  Besides introducing the circuit test, Magniez et al. improved the Mayers-Yao result by making it robust, allowing their circuit test to be robust as well.

\section{Definitions and main Theorems}

\subsection{Self-testing concepts}
The context for self-testing includes a \emph{\index{verifier}verifier} and several black-box quantum devices.  Ideally, no assumptions are made about the devices that cannot be verified in some way.  For example, we wish not to make any assumptions about the Hilbert space that states live in and gates operate on.  

The verifier interacts with the quantum devices in three ways.  First, the verifier arranges the quantum devices into a circuit by connecting quantum inputs and outputs of various devices.  Second, the verifier provides classical inputs (\emph{\index{measurement settings}measurement settings}) to the devices.  Finally, the verifier obtains classical outputs from the devices (\emph{outcomes}).  Importantly, we assume that no signalling occurs between devices that are not explicitly connected in the topology, and if there is some physical channel then this cannot be used to communicate backwards.
\begin{assumption}
No communication occurs between quantum devices that are not explicitly connected.  Further, quantum communication from the output of one device to the input of another device is one-way.
\end{assumption}
There are two pieces of information about an experiment available to the verifier.  First, the topology is known, since the verifier is responsible for it.  Second, the verifier may estimate the statistics generated by the experiment, provided the experiment behaves the same each time it is used.
\begin{assumption}\label{assumption:devicesbehavesame}
The quantum devices behave the same each time they are used.
\end{assumption}
The statistics will allow the verifier to determine if the physical experiment simulates the reference experiment.  In the case that it does not, we offer no conclusions.  If it does, however, then we can make some non-trivial conclusions about the structure of the physical experiment, which will be the subject of this chapter.

\subsection{Equivalence}
For general reference experiments, if the verifier is able to calculate the statistics generated in order to compare with a physical experiment, then there is little value in performing the physical experiment; the verifier may do the calculation themselves.  The strength of self-testing is in choosing experiments for which the statistics are easy to calculate - many such experiments are performed, each one adding more gates - and combining the conclusions gained from each experiment to conclude that the final experiment simulates the desired reference experiment.

In order to combine results about individual experiments and make the final conclusion, we need a notion stronger than simulation.  Consider the case of factoring.  A physical experiment that outputs the factors of a large number may perform the calculation in any number of ways, from Shor's algorithm to brute force search.  All of these methods simulate each other, but clearly they are different in important ways.  For this reason we introduce the concept of \emph{\index{equivalence}equivalence}.  We intend to capture the idea that, to the greatest extent we can possibly conclude, the physical experiment is \emph{the same as} the reference experiment.

When defining a notion of equivalence in this setting we must first consider how me might change the reference experiment in a way that preserves the statistics of the outcomes.  Any such change is invisible from the perspective of the verifier and hence we cannot rule them out.  Here is a list of such changes

\begin{enumerate}
	\item Local changes of basis
	\item Adding ancillae to physical systems, prepared in any joint state
	\item Changing the action of the observables outside the support of the state
	\item Locally embedding the state and operators in a larger (or smaller) Hilbert space
\end{enumerate}

In order to accommodate these various changes we define equivalence as follows:

\begin{definition}\label{def:equivalence}
 A {\it reference experiment} is described by a $n$-partite state $\ket{\psi}$ on Hilbert space $\mathcal{X} = \mathcal{X}_{1} \otimes \dots \mathcal{X}_{n}$ and local measurement observables (acting on a single part) $M_{m}$ for various $m$.  Further, consider a physical experiment described by a $n$-partite state $\ket{\psi^{\prime}}$ on Hilbert space $\mathcal{Y} = \mathcal{Y}_{1}\otimes \dots \otimes \mathcal{Y}_{n}$ and local measurement observables $M^{\prime}_{m}$ for various $m$.  We say that the physical experiment is \emph{equivalent} to the reference experiment (and the physical state and measurement observables are equivalent to the reference state and measurement observables) if there exists a local isometry
\begin{equation}
\Phi = \Phi_{1} \otimes \dots \Phi_{n}, \, \, \, \,
\Phi_{j} : \mathcal{Y}_{j} \mapsto \mathcal{Y}_{j} \otimes \mathcal{X}_{j}
\end{equation}
 such that
\begin{eqnarray}
\Phi(\ket{\psi^{\prime}}) & = & \ket{junk}_{\mathcal{Y}} \otimes \ket{\psi}_{\mathcal{X}}\\
\Phi(M^{\prime}_{m}\ket{\psi^{\prime}}) & = & \ket{junk}_{\mathcal{Y}} \otimes M_{m}\ket{\psi}_{\mathcal{X}}.
\end{eqnarray}
\end{definition}

When we are performing gate testing, we will be concerned with several experiments, and several equivalences at once.  The particular isometries and junk state $\ket{junk}_{\mathcal{Y}}$ used will be important.  In this case we will specify them.

The isometry $\Phi$ may be constructed by attaching ancillae in some product state $\ket{00\dots 0}_{\mathcal{X}}$ and applying local unitaries to the subsystems.  Note that if we make any finite number of changes from the list above then we may construct a suitable local isometry and show that the experiment is equivalent to the reference experiment.  Also, any experiment that is equivalent to the reference experiment may be constructed by applying changes from the list above:  one simply attaches ancillae in the state $\ket{junk}$ and performs a suitable change of basis.  Equivalence is thus exactly the notion we need to take these changes into account. 

\subsection{Results and contributions}

The central idea behind self-testing is that, for certain carefully chosen experiments, simulation implies equivalence.  Furthermore, experiments may be grouped together to strengthen the conclusions.  This allows us to construct self-tests for
\begin{itemize}
	\item EPR pairs along with a particular set of measurements
	\item Real unitaries on single qubits and $\text{CTRL-}Z$ gates
	\item Arbitrary circuits composed of the above gates
\end{itemize}
These results are already present in \cite{Mayers:2004:Self-testing-qu} and \cite{Magniez:2006:Self-testing-of}, however we make several contributions:
\begin{itemize}
	\item Streamlined definitions and notation
	\item New (and often simplified) proofs for all results
	\item Corrected proofs for certain technical lemmas
	\item Explanation of the restriction to real unitaries
\end{itemize}
We also extend the definition of equivalence, taking into account the results of chapter~\ref{chap:simulations}, and extend the Mayers and Yao test for EPR pairs to include complex measurements.  This lays the foundation for self-testing of complex gates.

The extended Mayers and Yao test is to be published in \cite{McKague:2010:Generalized-sel}.

\section{State and measurement testing}\label{sec:mayersyao}
The backbone of self-testing is in testing states and measurements.  Since we cannot test these individually without relying on some trusted apparatus, we test them simultaneously.  The test that we use here was developed by Mayers and Yao \cite{Mayers:2004:Self-testing-qu}.

\begin{theorem}[Mayers and Yao \cite{Mayers:2004:Self-testing-qu}]\label{theorem:mayers-yao}
Suppose a physical experiment has the same topology and generates the same statistics as the reference experiment described in section~\ref{sec:myreftest}.  Then the physical experiment is equivalent to the reference experiment.
\end{theorem}

\subsection{State and measurement self-test reference experiment}\label{sec:myreftest}
A general schematic for the Mayers-Yao reference experiment is shown in figure~\ref{eprtest}.  A bipartite state $\ket{\psi}$ is distributed to a pair of measurement devices.  The two measurement devices take classical inputs $a$ and $b$, which each take one of three values.  The devices then output classical bits, $x$ and $y$.

 \begin{figure}
\[ \Qcircuit {
 & & a \\
 & & \cwx[1] \\
  & & \meter & \cw & x \\
\lstick{\ket{\psi}}   \ar@{-}[ur] \ar@{-}[dr]\\
  & & \meter & \cw & y \\
  & & \cwx[-1] \\
  & & b \\
 }
\]
\caption{Mayers-Yao self-test circuit}
\label{eprtest}
\end{figure}
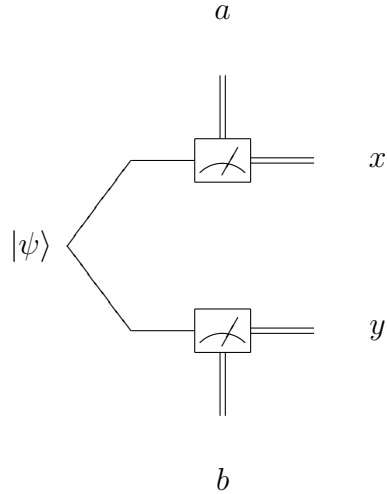

The reference state is an EPR pair $\ket{\phi_{+}} = \frac{1}{\sqrt{2}}\left(\ket{00} + \ket{11} \right)$ and the reference measurement observables are $X, Z, \frac{X + Z}{\sqrt{2}}$ for each side of the EPR pair.  For brevity we label $\frac{X + Z}{\sqrt{2}} = D$.  For physical devices we will have to derive this relationship, so there the separate label $D$ is required.

\subsection{Proof of Theorem~\ref{theorem:mayers-yao}}

\subsubsection{Proof Overview}
The main advantages of the following new proof for the Mayers-Yao self-test is that it is shorter, clearer, and more naturally extends to the more general test given in this paper.

The proof has two distinct parts.  The first part establishes some equations on the state and observables based on the observed statistics.  These are straightforward and are a direct result of the statistics observed.  Next we use these equations to show that the $X$ and $Z$ observables on each side anti-commute on the support of the state.  The second part uses the anti-commuting observables to construct local isometries that take the state and observables to the reference state and observables.

\subsubsection{Statistics}

In the reference experiment the marginals for each observable are all 0.  That is,
\[
	\bra{\phi_{+}}M \otimes I \ket{\phi_{+}} = 0
\]
for $M \in \{X, Z, D \}$.  (Swapping the systems in this and the following equations gives the same result since $\ket{\phi_{+}}$ is symmetric.)  Measuring the same observable on both sides always give identical outcomes.  Thus
\[
\bra{\phi_{+}}M \otimes M\ket{\phi_{+}} = 1.
\]
Additionally, $X$ and $Z$ measurements are uncorrelated.
 \[
\bra{\phi_{+}} X \otimes Z \ket{\phi_{+}}  = 0.
\]
The interesting part comes when we measure $X$ or $Z$ on one side and $D$ on the other.
\[
\bra{\phi_{+}} X \otimes D \ket{\phi_{+}} = \bra{\phi_{+}} Z \otimes D \ket{\phi_{+}} = \frac{1}{\sqrt{2}}
\]

\subsubsection{State equalities}

Using the equations on the measurement outcomes from above and the fact that $\ket{\psi}$ is normalized gives us the following equations.

\begin{eqnarray}
\ket{\psi} & = & X_{A} \otimes X_{B} \ket{\psi} \\
& = & Z_{A} \otimes Z_{B} \ket{\psi} \\
& = & D_{A} \otimes D_{B} \ket{\psi} \\
X_{A} \otimes I \ket{\psi} & = & I \otimes X_{B} \ket{\psi} \label{eq:xAxB} \\
Z_{A} \otimes I \ket{\psi} & = & I \otimes Z_{B} \ket{\psi} \label{eq:zAzB} \\
D_{A} \otimes I \ket{\psi} & = & I \otimes D_{B} \ket{\psi} \\
X_{A}Z_{A} \otimes I \ket{\psi} & = & I \otimes Z_{B}X_{B} \ket{\psi} \\
Z_{A}X_{A} \otimes I \ket{\psi} & = & I \otimes X_{B}Z_{B} \ket{\psi} \\
X_{A}Z_{A} \otimes I \ket{\psi} & = & X_{A} \otimes Z_{B} \ket{\psi} \\
Z_{A}X_{A} \otimes I \ket{\psi} & = & Z_{A} \otimes X_{B} \ket{\psi}
\end{eqnarray}

We can also establish some orthogonality relationships between various vectors.  In particular the vectors $\ket{\psi}, X_{A} \otimes I \ket{\psi}, Z_{A} \otimes I \ket{\psi}, X_{A}Z_{A} \otimes I \ket{\psi}$ are pairwise orthogonal.

Our goal for the remainder of the proof is to show that any state for which these equations hold must be equivalent to $\ket{\phi_{+}}$.

\subsubsection{Anti-commuting observables}\label{sec:my_anticommuting}

We now move to more salient matters.  First, we note that  $D_{A} \otimes I \ket{\psi}$ must be in the space spanned by $X_{A} \otimes I \ket{\psi}$ and $Z_{A} \otimes I \ket{\psi}$ because it has overlap $\frac{1}{\sqrt{2}}$ with each of these orthogonal vectors, and it has norm 1.  Thus
 \[
 D_{A} \otimes I \ket{\psi} = \frac{X_{A} + Z_{A}}{\sqrt{2}} \otimes I \ket{\psi}
\]
and analogously for $I \otimes D_{B} \ket{\psi}$.  This allows us to make the following deductions.

\begin{eqnarray*}
\ket{\psi} & = & D_{A} \otimes D_{B} \ket{\psi} \\
& = & \frac{1}{2} (X_{A} + Z_{A}) \otimes (X_{B} + Z_{B}) \ket{\psi} \\
& = & \ket{\psi} + (X_{A}\otimes Z_{B} + Z_{A} \otimes X_{B}) \ket{\psi}
\end{eqnarray*}

\noindent
Applying equations~\ref{eq:xAxB} and~\ref{eq:zAzB} we obtain

\begin{equation}
(X_{A}Z_{A} + Z_{A} X_{A}) \otimes I \ket{\psi} = 0.
\end{equation}
By Lemma~\ref{lemma:anticommute_support}, below,  it follows that $X_{A}$ and $Z_{A}$ anti-commute on the support of $\ket{\psi}$ on $A$.  Similarly, the observables $X_{B}$ and $Z_{B}$ anti-commute on support of $\ket{\psi}$ on $B$.

\begin{lemma}\label{lemma:anticommute_support}
Let $X_{A}$ and $Z_{A}$ be operators and $\ket{\psi}_{AB}$ a bipartite state such that
\begin{equation}
X_{A}Z_{A}\otimes I_{B} \ket{\psi}_{AB} = - Z_{A}X_{A} \otimes I_{B} \ket{\psi}_{AB}.
\end{equation}
then $X_{A}Z_{A} \ket{\phi} = - Z_{A}X_{A} \ket{\phi}$ for any $\ket{\phi}$ in the support of $\ket{\psi}_{AB}$ on $A$.
\end{lemma}

\begin{proof}
Let
\begin{equation}
\ket{\psi} = \sum_{j} \lambda_{j} \ket{j}_{A} \ket{j}_{B}
\end{equation}
be the singular value decomposition of $\ket{\psi}$.  We then have
\begin{equation}
X_{A}Z_{A}\otimes I_{B} \sum_{j} \lambda_{j} \ket{j}_{A} \ket{j}_{B} = - Z_{A}X_{A} \otimes I_{B} \sum_{j} \lambda_{j} \ket{j}_{A} \ket{j}_{B}.
\end{equation}
We now take the inner product with $\ket{k}_{A}\ket{k}_{B}$ for some $k$, to obtain
\begin{equation}
\lambda_{k}\bra{k}_{A}X_{A}Z_{A}\ket{k}_{A} = -\lambda_{k}\bra{k}_{A}X_{A}Z_{A}\ket{k}_{A}.
\end{equation}
When we restrict to the subspace to the subspace spanned by the $\ket{k}_{A}$ for which $\lambda_{k} \neq 0$ (i.e. on the support of $\ket{\psi}$ on $A$) we find that $X_{A}Z_{A} = - Z_{A}X_{A}$.

\end{proof}

\subsubsection{Construction of the local isometry}

Now we can easily build the local unitaries required to extract the EPR pair.  We use the circuit shown in figure~\ref{fig:epr_local_unitary_circuit}.  There the outer $\ket{0}$ states are added while the two inner wires carry the two halves of the bipartite state $\ket{\psi}$.  This circuit essentially builds a SWAP gate out of two CNOT gates (the usual third gate is not necessary since we initialize with $\ket{0}$.)  The SWAP gate extracts the entanglement out of $\ket{\psi}$ and swaps in a product state.

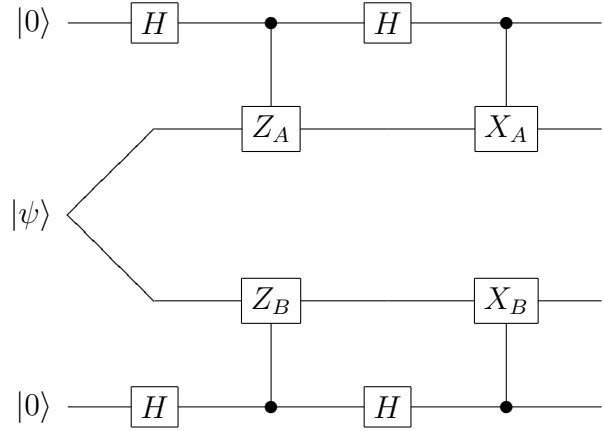
\begin{figure}
\[
\Qcircuit {
\lstick{\ket{0}}  & \gate{H} & \ctrl{1}  & \gate{H} & \ctrl{1} &\qw \\
 &  & \gate{Z_{A}} & \qw & \gate{X_{A}} & \qw\\
 \lstick{\ket{\psi}} \ar@{-}[ur] \ar@{-}[dr] & & & & & \\
              &    & \gate{Z_{B}} & \qw & \gate{X_{B}} & \qw\\
\lstick{\ket{0}}     & \gate{H} & \ctrl{-1} & \gate{H} & \ctrl{-1} & \qw
}
\]
\caption{Circuit for $\Phi$ showing equivalence of physical circuit to reference circuit in Mayers-Yao test}
\label{fig:epr_local_unitary_circuit}
\end{figure}

The circuit gives two isometries, one for each wire in EPR test circuit, which we denote $\Phi_{A}$ and $\Phi_{B}$.

\subsubsection{Isometry applied to state}

After applying this circuit the resulting state is
 \begin{eqnarray*}
 \Phi_{A} \otimes \Phi_{B} ( \ket{\psi}) & = &\frac{1}{4} (I + Z_{A}) \otimes (I + Z_{B}) \ket{\psi}\ket{00} \\
 & + & \frac{1}{4}(I + Z_{A}) \otimes X_{B}(I - Z_{B}) \ket{\psi}\ket{01} \\
 & + & \frac{1}{4}X_{A}(I- Z_{A}) \otimes (I + Z_{B}) \ket{\psi}\ket{10}\\
 & + & \frac{1}{4} X_{A} (I - Z_{A}) \otimes X_{B} (I-Z_{B}) \ket{\psi}\ket{11}. \\
\end{eqnarray*}
Applying some equations and the anti-commuting result from the previous section we find that this is equal to
 \[
  \Phi_{A} \otimes \Phi_{B} ( \ket{\psi}) = \frac{1}{4}(I+Z_{A}) \otimes (I + Z_{B}) \ket{\psi}\left(\ket{00} + \ket{11} \right) +
 \]
 \[
 (I + Z_{A})(I - Z_{A})\otimes X_{B}\ket{\psi}\ket{01} +
  X_{A}\otimes (I + Z_{B})(I - Z_{B})\ket{\psi}\ket{10}
\]
 \begin{equation}
  =  \frac{1}{\sqrt{2}} (I \otimes I + I \otimes Z_{B}) \ket{\psi} \ket{\phi_{+}}.
  \end{equation}
This may look curious since $I + Z_{A}$ and $I + Z_{B}$ are not unitary.  In fact it is easy to show that the final state still has the correct norm.  To give some intuition, note that in the reference case we want to extract $\ket{\phi_{+}}$ and swap in $\ket{00} = \frac{1}{2\sqrt{2}}(I+Z) \otimes (I+Z)\ket{\phi_{+}}$.

\subsubsection{Isometry and measurement operators}

We start with $X_{A}$ (the result for $X_{B}$ follows analogously).  Applying $X_{A}$ to $\ket{\psi}$ before applying the circuit is the same as applying it at the end, with a $-1$ phase introduced by anti-commuting past the controlled $Z_{A}$ operation (recall from section~\ref{sec:my_anticommuting} that $X_{A}$ and $Z_{A}$ anti-commute on the relevant subspace).  The resulting state is
 \begin{eqnarray*}
  \Phi_{A} \otimes \Phi_{B} ( X_{A} \otimes I_{B}\ket{\psi})& =  &\frac{1}{4} X_{A}(I - Z_{A}) \otimes (I + Z_{B}) \ket{\psi}\ket{00} \\
 & + & \frac{1}{4}X_{A}(I - Z_{A}) \otimes X_{B}(I - Z_{B}) \ket{\psi}\ket{01} \\
 & + & \frac{1}{4}(I+ Z_{A}) \otimes (I + Z_{B}) \ket{\psi}\ket{10}\\
 & + & \frac{1}{4}  (I + Z_{A}) \otimes X_{B} (I-Z_{B}) \ket{\psi}\ket{11}. \\
\end{eqnarray*}
Following the same logic as used in the state equivalence, we find that the final state is
 \[
 \Phi_{A} \otimes \Phi_{B} ( X_{A} \otimes I \ket{\psi})  = \frac{1}{\sqrt{2}} (I \otimes I + I \otimes Z_{B}) \ket{\psi}(X \otimes I) \ket{\phi_{+}}.
 \]

For the $Z_{A}$ operation, we see that the effect is a $-1$ phase kicked back through the final controlled $X_{A}$ operation.  This phase appears on the terms with $\ket{1}$ in the qubit, exactly as if a $Z$ operation had been applied to the qubit.  The equivalence for the $D$ operators results from the fact that $D = \frac{X + Z}{\sqrt{2}}$ on the relevant subspace, and linearity.

This concludes the proof of Theorem~\ref{theorem:mayers-yao} .

\section{Gate self-test}\label{sec:selftestgatetest}

\subsection{Main result}

Unlike state and measurement testing, gate testing will require three different experiments.  The first two are state and measurement tests, and the last verifies the operation of the gate.  We wish to test a gate $T \in U(\mathcal{X})$ with the restriction that $T$ have all real matrix entries.  For our purposes we will consider $\mathcal{X} = \mathcal{H}_{2}$ with $T$ any real unitary, and $\mathcal{X} = \mathcal{H}_{2}^{2}$ with $T = \text{CTRL-}Z$.  Note that this set of gates is sufficient to simulate universal quantum computation using the real simulation given in section~\ref{sec:real_simulation}.  The reference state $\ket{\phi}$ will either be an EPR pair, $\ket{\phi_{+}}$, or a pair of EPR pairs.  Measurements $M_{a}$ and $N_{b}$ will be the tensor products of the real operators $\{I, X, Z, \frac{X+Z}{\sqrt{2}}\}$.

The experiments are as follows:  First we test the state $\ket{\phi}$ and the measurements, as in figure~\ref{fig:epr_test} (experiment 1).  This establishes that the state is correct, and that the measurements applied directly to the state are also correct.  Next we apply $T \otimes T$ to $\ket{\phi}$ and then measure, as in figure~\ref{fig:epr_test_with_gates}, which tests whether the measurements are still operating correctly \emph{after} the gate is applied (experiment 2). (Note that for $T$ real and our chosen state, $T \otimes T\ket{\phi} = \ket{\phi}$.)  Finally, we apply $T \otimes I$ to $\ket{\phi}$ and measure, as in figure~\ref{fig:gate_test} (experiment 3).

We can think of the final test in the following way:  Measuring the $B$ side of the maximally entangled state is equivalent to preparing a state on the $A$ side.  Critically, there is no information about the basis $A$ side.  We then apply the gate to this prepared state and measure afterwards in a variety of bases to establish that the gate is working.  Also, since there is no way to distinguish the second and third tests from the $A$ side alone, either the measurement is working correctly or will be detected as faulty.

 \begin{figure}
\[ \Qcircuit {
  & & \meter \\
\lstick{\ket{\psi}}   \ar@{-}[ur] \ar@{-}[dr]\\
  & & \meter \\
 }
\]
\caption{EPR test}
\label{fig:epr_test}
\end{figure}

\begin{figure}
 \[
 \Qcircuit {
  & & \gate{T} & \meter \\
\lstick{\ket{\psi}}   \ar@{-}[ur] \ar@{-}[dr]\\
  & & \gate{T} & \meter \\
 }
\]
\caption{EPR test after gates applied}
\label{fig:epr_test_with_gates}
\end{figure}

\begin{figure}
 \[
 \Qcircuit {
  & & \gate{T} &  \meter \\
\lstick{\ket{\psi}}   \ar@{-}[ur] \ar@{-}[dr]\\
  & & \meter \\
 }
 \]
\caption{Testing a gate}
\label{fig:gate_test}
\end{figure}
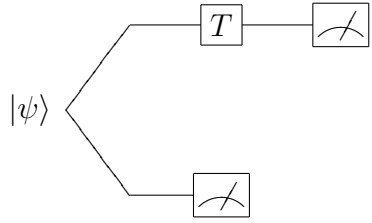

Recall Assumption~\ref{assumption:devicesbehavesame}, that the devices always operate the same.  Also, we assume that there is one way communication along the quantum channels and no side channels.  Thus the devices cannot determine which experiment is being performed.  The $B$ side measurements devices cannot distinguish between experiments 1 and 3, and so the $B$ side measurements on experiment 3 are implicitly tested by experiment 1.  Likewise, the $A$ side measurement devices cannot distinguish between experiments 2 and 3, and so are implicitly verified by experiment 2 to be working correctly in experiment 3.

One potential concern is the fact that a physical device may hold the state in memory rather than receiving it from the source device.  In fact this is not a concern.  If this were the case, then we simply take $\ket{\psi}$ to be the state held in memory across all devices.  There remains the possibility that the maximally entangled pairs found in experiment 1 and 2 are not the same pairs.  However, they must in fact be the same pair since the statistics generated by experiment 3 are only possible for a maximally entangled state, with one measurement the same as in experiment 1 and the other the same as in experiment 2.  The net result is that we do not have to concern ourselves with this possibility since the Theorem takes care of this case as well.

\begin{theorem}\label{theorem:gatetesting}
Suppose that a collection of physical devices are arranged into three physical experiments:
\begin{itemize}
	\item (experiment 1) $\ket{\psi}$ with measurements $M^{\prime}_{a} \otimes N^{\prime}_{b}$ 
	\item (experiment 2) $G^{\prime} \otimes H^{\prime} \ket{\psi}$ with measurements $M^{\prime\prime}_{a} \otimes N^{\prime\prime}_{b}$
	\item (experiment 3) $G^{\prime} \otimes I_{\mathcal{Y}} \ket{\psi}$ with measurements $M^{\prime\prime}_{a} \otimes N^{\prime}_{b}$.
\end{itemize}
Further suppose that each physical experiment simulates the corresponding reference experiment described above.  Then there exist unitaries $U_{A}, U_{B}, V_{A}, V_{B}$ and state $\ket{junk}_{\mathcal{Y}}$ such that
\begin{itemize}
	\item physical experiment 1 is equivalent to reference experiment 1 with unitaries $U_{A} \otimes U_{B}$
	\item physical experiment 2 is equivalent to reference experiment 2 with unitaries $V_{A} \otimes V_{B}$
	\item physical experiment 3 is equivalent to reference experiment 3 with unitaries $V_{A} \otimes U_{B}$
\end{itemize}
with junk state $\ket{junk}_{\mathcal{Y}}$ in each case.  Moreover,
\begin{equation}
V_{A} G^{\prime} \otimes I_{\mathcal{X}} U_{A}^{\dagger} = I_{\mathcal{Y}} \otimes T
\end{equation}
on the support of $\ket{junk}_{\mathcal{Y}^{2}}$ on $\mathcal{Y}_{A}$.
\end{theorem}

The theorem compares the reference gate $T$ with the physical gate $G$, and establishes, roughly, that if $G$ simulates $T$ then it is equivalent to $T$.

\subsection{Technical background for proof}\label{sec:selftestingtechnicalbackground}

\subsubsection{Choi-Jamiolkowski representation}
We first define the Choi-Jamiolkowski representation of a linear map on operators \cite{Jamiolkowski1972275} \cite{Choi:1975:Completely-posi} and discuss its properties.  We define it as follows

\begin{definition}
Let $\Phi: L(X) \rightarrow L(Y)$ be a linear operator.  Then the \emph{Choi-Jamiolkowski representation} of $\Phi$ is given by the operator
 \[
 J(\Phi) = \sum_{x,y} \Phi(\proj{x}{y}) \otimes \proj{x}{y}
\]
operating on the space $L(Y) \otimes L(X)$
\end{definition}

An equivalent definition is
 \[
 J(\Phi)  = \Phi \otimes I_{L(X)} \left( \sum_{x, x^{\prime}} \proj{xx}{x^{\prime}x^{\prime}}\right).
\]
Thus $J(\Phi)$ can be found by applying $\Phi$ to one half of a maximally entangled state on $X \otimes X$, finding the density matrix, and multiplying by a suitable scaling factor (the dimension of $X$).

We will use the following properties of $J(\Phi)$.  

\begin{theorem}
Let $\Phi: L(X) \rightarrow L(Y)$ and $J(\Phi)$ be given.  Then
\begin{itemize}
	\item $\Phi$ is completely positive if and only if $J(\Phi)$ is positive semi-definite.
	\item $\Phi$ is trace preserving if and only if $\text{Tr}_{Y}J(\Phi) = I_{X}$.
	\item $\Phi$ is unitary if and only if $J(\Phi)$ is rank 1 and the above two conditions hold.
\end{itemize}
\end{theorem}

For more details and rigorous proofs, see Watrous' lecture notes \cite{Watrous:2008:Characterizatio}.  We do not give a formal proof here, but instead offer some intuition.  If we interpret $J(\Phi)$ as the output from the completely positive map $\Phi \otimes I$, then it should be PSD whenever the input is.  For the second property, we start with a maximally entangled state, so tracing out one side should leave us with the completely mixed state.  For the last property, we are inputing a pure state, so the output should be pure if the map is unitary.  The converses of these statements are, of course, more difficult and we leave out the proofs.

\subsubsection{Technical Lemmas}
The following Lemma is a restriction of Lemma 5 appearing in \cite{Magniez:2006:Self-testing-of}, arXiv version.  The proof given by Magniez et al. contains an error, since real density matrices in general are \emph{not} in the span of the tensor products of $\{I, X, Z\}$.  As a concrete example, $\proj{\phi_{+}}{\phi_{+}}$ contains $Y \otimes Y$, a real matrix, in its decomposition into Paulis.  Here we give correct proofs for the case of two qubits, the case of four qubits with $T = \text{CTRL-}Z$, and the case of a tensor product of such gates on any number of qubits.  The proof for the second case could be adapted to work with other two qubit Clifford gates as well. 

\begin{lemma}\label{lemma:qubitrealunitarychoi}

Let $\sigma = \proj{\psi}{\psi}$ be a two qubit state with $\ket{\psi} = T \otimes I \ket{\Phi+}$ and $T$ a unitary having real coefficients.  Further, suppose that
\begin{equation}
\tr{\rho M \otimes N} = \tr{\sigma M \otimes N}
\end{equation}
for $M,N \in \{I, X, Z\}$.  Then $\rho = \sigma$.
\end{lemma}
\begin{proof}
We first fix some notation.  Let U and V be 2-qubit Pauli operators.  Let $U \cdot V$ be defined by
\begin{equation}
U \cdot V = 
\begin{cases}
	1 & U\, \text{and}\, V\, \text{commute} \\
	-1 & U\, \text{and}\, V\, \text{anti-commute.}\\
\end{cases}
\end{equation}
Also, for density operator $\rho$ (and analogously for $\sigma$ and other density operators) define $\rho_{U}$ by
\begin{equation}
\rho_{U} = \tr{\rho U}
\end{equation}
and analogously for other density operators.  Note that $|\rho_{U}| \leq 1$.

We will make use of the following observations about $\sigma$, which can be easily verified from the fact that $T$ is unitary and real.  For $M$ and $N$ Pauli operators we find
\begin{itemize}
	\item $\displaystyle\sum_{M, N \in{I,X,Y,Z}} \sigma_{M \otimes N}^{2} = 4$
	\item $\sigma_{Y \otimes M} = \sigma_{M \otimes Y}= 0$ when $M \neq Y$
	\item $|\sigma_{Y\otimes Y}| = 1$
	\item $\sigma_{M \otimes I} = \sigma_{I \otimes M} = 0$ for all $M \neq I$.
\end{itemize}

Since $\sigma$ and $\rho$ are positive semi-definite, we have
 \[
 \tr{\rho U \sigma U^{\dagger}} \geq 0
\]
for any unitary $U$.  We may write this, using the notation above, as
 \[
 \sum_{P \in \{I,X,Y,Z\}^{\otimes 2}}\rho_{P} \sigma_{P} (U \cdot P) \geq 0.
\]
With the choice $U = Y \otimes I\, $Êwe find
 \[
 1 - \sum_{M,N \in \{X,Z\}} \rho_{M \otimes N} \sigma_{M \otimes N}  + \rho_{Y\otimes Y} \sigma_{Y \otimes Y} = -1 + \rho_{Y \otimes Y} \sigma_{Y \otimes Y}\geq 0.
\]
This is obtained by removing 0 terms and noting that $\sum_{M,N \in \{X,Z\}} \sigma^{2}_{M \otimes N} = 2$.  The implication is that $\sigma_{Y \otimes Y} = \rho_{Y \otimes Y}$ and hence $\sigma_{M \otimes N} = \rho_{M \otimes N}$ whenever $\sigma_{M \otimes N}$ is not $0$.

Let $S$ be the set of Pauli operators $M \otimes N$ for which $\sigma_{M \otimes N} \neq 0$.  Then, since $\sigma_{M \otimes N} = \rho_{M \otimes N}$ for all $M \otimes N \in S$ we have
 \[
 \tr{\rho^{2}} = 1 + \frac{1}{4}\sum_{M \otimes N \notin S} \rho^{2}_{M \otimes N} \leq 1
\]
since $\rho$ has trace 1.  This immediately implies that $\rho_{M \otimes N} = 0$ for $M \otimes N \notin S$ and hence $\rho_{M \otimes N} = \sigma_{M \otimes N}$ for all $M \otimes N$.  Recall that the 2-fold tensor products of the Pauli operators form a basis for the space of 2-qubit states.  Thus $\rho = \sigma$.
\end{proof}
\begin{lemma}\label{lemma:qubitrealunitarychoi2}

Let $\sigma = \proj{\psi}{\psi}$ be a four qubit state with 
\begin{equation}
\ket{\psi} = \text{CTRL-}Z \otimes I_{B} \left(\frac{1}{2}\sum_{x \in \{0,1\}^{2}}\ket{x}_{A}\ket{x}_{B}\right)
\end{equation}
Further, suppose that
\begin{equation}
\tr{\rho M \otimes N} = \tr{\sigma M \otimes N}
\end{equation}
for $M,N \in \{I, X, Z\}$.  Then $\rho = \sigma$.
\end{lemma}
\begin{proof}
The proof follows the same plan as for the two qubit version.  To begin, we examine the effect of $\text{CTRL-}Z$ on the two qubit Paulis.  Define $C(P)$ to be $\text{CTRL-}Z( P) \text{CTRL-}Z$ for two qubit Pauli $P$.  This is summarized in the following table:
\begin{equation}
\begin{array}{c|cccc}
 C(\cdot)& I & X & Y & Z \\
\hline \\
I & II & ZX & ZY & IZ \\
X & XZ & YY & YX & XI \\
Y & YZ & XY & XX & YI \\
Z & ZI & IX & IY & ZZÊ\\
\end{array}
\end{equation}
(It may benefit the reader to note that the table is symmetric and $C(\cdot)$ is its own inverse, since $\text{CTRL-}Z$ has these properties.)  The state $\frac{1}{2}\sum_{x \in \{0,1\}^{2}}\ket{x}_{A}\ket{x}_{B}$, written as a density operator, is
\begin{equation}
\frac{1}{16}\sum_{M, N \in \{I,X,Y,Z\}} (-1)^{\delta_{M,Y} + \delta_{N,Y}}(M \otimes N)_{A} \otimes (M \otimes N)_{B}.
\end{equation}
Note the coefficient, which is $-1$ when exactly one of $M$ and $N$ is $Y$.  We find that $\sigma$ is given by
\begin{equation}
\sigma = \frac{1}{16}\sum_{M, N \in \{I,X,Y,Z\}} (-1)^{\delta_{M,Y} + \delta_{N,Y}}C(M \otimes N)_{A} \otimes (M \otimes N)_{B}
\end{equation}
and hence
\begin{equation}
|\sigma_{(M^{\prime} \otimes N^{\prime})_{A} \otimes (M \otimes N)_{B}}| = \delta_{ M^{\prime} \otimes N^{\prime},C(M \otimes N)}.
\end{equation}
As for the two qubit case, we have
\[
 \tr{\rho U \sigma U^{\dagger}} \geq 0
\]
for any unitary $U$.  Dropping the zero terms and subbing in, we may write 
 \[
 \sum_{M \otimes N} (-1)^{\delta_{M,Y} + \delta_{N,Y}}\rho_{C(M \otimes N)_{A} \otimes (M \otimes N)_{B}} \geq 0.
\]
Define
\begin{equation}
R(M \otimes N) = (-1)^{\delta_{M,Y} + \delta_{N,Y}}\rho_{C(M \otimes N)_{A} \otimes (M \otimes N)_{B}}
\end{equation}
in which case we find
\begin{equation}
\sum_{M,N} R(M \otimes N) (U \cdot C(M \otimes N) \otimes (M \otimes N)) \geq 0.
\end{equation}
We will prove that $R(M \otimes N) = 1$, in which case we have $\rho_{(M \otimes N)_{A} \otimes C(M \otimes N)_{B}} = \sigma_{C(M \otimes N)_{A} \otimes (M \otimes N)_{B}}$.  Note that this is given by the conditions of the Lemma in the cases where $C(M \otimes N) \otimes (M \otimes N)$ is a tensor product of $\{I, X, Z\}$.  This occurs for $M \otimes N \in \{ I\otimes I, I\otimes X, X\otimes I, Z\otimes I, I\otimes Z, Z\otimes Z, X\otimes Z, Z\otimes X\}$.  

We find the inequalities with $U = (Y \otimes I)_{A}$, $U = (I \otimes Y)_{A}$, $U = (Y \otimes I)_{B}$, and $U = (I \otimes Y)_{B}$ and sum them, dividing by 4.  We obtain
\begin{equation}
R(I \otimes Y) + R(Y \otimes I) + R(Y \otimes Z) + R(Z \otimes Y) \geq 4.
\end{equation}
Since all four values $R$ values appearing cannot exceed $1$, we find that they are all equal to 1.  Next we find the inequality for $U = (Z \otimes Z)_{A}$.  For the four remaining unknown $R$ values we obtain
\begin{equation}
R(X\otimes X) + R(Y \otimes Y) + R(X \otimes Y) + R(Y \otimes X) \geq 4
\end{equation}
and again all four $R$ values must be 1.

We now have $\sigma_{C(M \otimes N)_{A} \otimes (M \otimes N)_{B}} =\rho_{C(M \otimes N)_{A} \otimes (M \otimes N)_{B}}$.  There are 16 such terms.  By examining the trace of the squares of $\sigma$ and $\rho$, and following the same reasoning as in the two qubit case, we conclude that the remaining terms must all be 0.  Thus $\sigma = \rho$.

\end{proof}
\begin{lemma}\label{lemma:qubitrealunitarychoi3}

Let $\sigma = \proj{\psi}{\psi}$ be a $2n$ qubit state with $\ket{\psi} = T_{A} \otimes I_{B} \ket{\Phi+}^{\otimes n}$ and $T$ a tensor product of gates that are either single qubit real unitaries or $\text{CTRL-}Z$.  Further, suppose that
\begin{equation}
\tr{\rho M \otimes N} = \tr{\sigma M \otimes N}
\end{equation}
where $M,N \in \{I, X, Z\}^{\otimes n}$ are one of the following
\begin{itemize}
	\item For a single qubit gate in the tensor product, $M$ and $N$ measure that qubit and its pair
	\item For a two qubit gate in the tensor product, $M$ and $N$ measure those qubits and their pairs.
\end{itemize}
Note that $M$ or $N$ may be the identity (measuring marginals).  Then $\rho = \sigma$.
\end{lemma}
Note that the EPR pairs are each divided into $A$ and $B$ sides and $T$ is applied to the $n$ $A$ side qubits.
\begin{proof}
Divide up $\rho$ and $\sigma$ into pairs or fours of qubits, depending on whether the gate is a single qubit gate or a $\text{CTRL-}Z$.  For each piece trace out the remaining qubits and apply the appropriate two or four qubit version of the Lemma.  Then, since the reduced qubit on each piece is pure, $\rho$ and $\sigma$ are both the tensor product of the reduced qubits.
\end{proof}
\begin{futurework}
Extend the previous Lemma to two-qubit Clifford gates, and ultimately all real gates.
\end{futurework}

Like the previous Lemma, a similar Lemma to the following appears in \cite{Magniez:2006:Self-testing-of} (arXiv version) as Lemma 6.  Also like the previous Lemma, the proof given by Magniez et al. relies on the claim that the real density matrices are in the span of the tensor products of $\{I, X, Z\}$.  We give a correct proof that uses the Choi-Jamiolkowski representation.  The technique used will be applied several times in the remainder of this chapter.

\begin{lemma}\label{lemma:equivmeasidentonjunk}
Suppose that a physical experiment with bipartite state $\ket{\psi} \in \mathcal{Y}^{2}$ and with measurements $M^{\prime}_{a} \otimes N^{\prime}_{b}$ ($M_{0} = N_{0} = I_{\mathcal{Y}}$) is equivalent to the reference experiment with a maximally entangled state $\ket{\phi} \in \mathcal{X}^{2}$ and measurements $M_{a}Ê\otimes N_{b}$ ($M_{0} = N_{0} = I_{\mathcal{X}}$).  Then there exist unitaries $U_{A}, U_{B} \in U(\mathcal{Y} \otimes \mathcal{X})$ such that
\begin{equation}\label{eq:lemmaequivmeasidentonjunk1}
U_{A} \otimes U_{B} \left( M^{\prime}_{a} \otimes N^{\prime}_{b}\right) \ket{\psi}_{\mathcal{Y}} \ket{00}_{\mathcal{X}}
= \ket{junk}_{\mathcal{Y}^{2}} M_{a} \otimes N_{b} \ket{\phi}_{\mathcal{X}}
\end{equation}
and
\begin{equation}\label{eq:lemmaequivmeasidentonjunk2}
U_{A} \left(M^{\prime}_{a} \otimes I_{\mathcal{X}}\right) U^{\dagger}_{A} = I_{\mathcal{Y}} \otimes M_{a}
\end{equation}
when confined to the support of $\ket{junk}_{\mathcal{Y}^{2}}$ on $\mathcal{Y}_{A}$, and analogously for $N^{\prime}_{b}$.
\end{lemma}
\begin{proof}
Equation~\ref{eq:lemmaequivmeasidentonjunk1} follows directly from the definition of equivalence and a straightforward extension of the local isometry.  From this fact we make the following observations:
\begin{equation}
\ket{\psi}\ket{00} = (M^{\prime}_{a} \otimes I_{\mathcal{X}^2})^{\dagger} U_{A}^{\dagger}Ê\otimes U_{B}^{\dagger} \ket{junk}_{\mathcal{Y}^{2}} M_{a} \otimes I_{\mathcal{X}_B} \ket{\phi}_{\mathcal{X}}.
\end{equation}
Applying equation~\ref{eq:lemmaequivmeasidentonjunk1} with $M_{0}$ and $N_{0}$, we find
\begin{equation}
\ket{junk}_{\mathcal{Y}^{2}}\ket{\phi}_{\mathcal{X}}
= U_{A} \otimes U_{B} (M^{\prime}_{a} \otimes I_{\mathcal{Y}_{B}} \otimes I_{\mathcal{X}^2} )^{\dagger} U_{A}^{\dagger}Ê\otimes U_{B}^{\dagger} \ket{junk}_{\mathcal{Y}^{2}} T_{a} \otimes I_{\mathcal{X}_B} \ket{\phi}_{\mathcal{X}}
\end{equation}
and hence
\begin{equation}
U_{A}\left(M^{\prime}_{a} \otimes I_{\mathcal{X}}\right) U_{A}^{\dagger} \ket{junk}_{\mathcal{Y}^{2}}\ket{\phi}_{\mathcal{X}}
= \ket{junk}_{\mathcal{Y}^{2}} M_{a} \otimes I_{\mathcal{X}} \ket{\phi}_{\mathcal{X}}.
\end{equation}
We now introduce a technique which will be quite useful.  Let $\Phi$ be the quantum operation on $\mathcal{X}$ defined by adding ancilla in the state $\ket{junk}_{\mathcal{Y}^{2}}$, applying $U_{A}\left(M^{\prime}_{a} \otimes I_{\mathcal{X}}\right) U_{A}^{\dagger}$ and finally tracing out the $\mathcal{Y}^{2}$ register.  Note that the above equation has (after tracing out the $\mathcal{Y}^{2}$ register) left side equal to $J(\Phi)/ \dim(\mathcal{X})$ and the right side (after tracing out the $\mathcal{Y}^{2}$ register) is $J(M_{a}) / \dim (\mathcal{X})$ (abusing notation a little).  From this we may conclude that $\Phi$ equates to simply applying $T_{a}$ and is hence unitary.  The operator $U_{A}\left(M^{\prime}_{a} \otimes I_{\mathcal{X}}\right) U_{A}^{\dagger}$ must then have the form $W_{\mathcal{Y}} \otimes M_{a}$ (when restricted to the support of $\ket{junk}_{\mathcal{Y}^{2}}$) and since the $\mathcal{Y}$ register remains in the same state, we conclude that $W$ is the identity.   

We apply the same reasoning and obtain the analogous result for $N_{b}$.
\end{proof}
\subsection{Proof of main Theorem}

We begin by applying Lemma~\ref{lemma:equivmeasidentonjunk} twice, to obtain 
\begin{equation}
U_{A} \otimes U_{B} M^{\prime}_{a} \otimes N_{b}^{\prime} \ket{\psi} \ket{00} = \ket{junk}_{\mathcal{Y}^{2}}M_{a} \otimes N_{b} \ket{\phi}
\end{equation}
\begin{equation}
U_{A} (M^{\prime}_{a}\otimes I_{\mathcal{Y}}) U_{A}^{\dagger} = I_{\mathcal{Y}} \otimes M_{a},\,\,  
U_{B} (N^{\prime}_{b}\otimes I_{\mathcal{Y}}) U_{B}^{\dagger} = I_{\mathcal{Y}} \otimes N_{b}
\end{equation}
when confined to the support of $\ket{junk}_{\mathcal{Y}^{2}}$ on the appropriate side, and
\begin{equation}
V_{A} \otimes V_{B} M^{\prime\prime}_{a} \otimes N_{b}^{\prime\prime} G^{\prime} \otimes H^{\prime}\ket{\psi} \ket{00} = \ket{junk_{2}}_{\mathcal{Y}^{2}}M_{a} \otimes N_{b} \ket{\phi}
\end{equation}
\begin{equation}\label{eq:selftestingmeasaftergate3}
V_{A} (M^{\prime\prime}_{a}\otimes I_{\mathcal{Y}}) V_{A}^{\dagger} = I_{\mathcal{Y}} \otimes M_{a},\,\,  
V_{B} (N^{\prime\prime}_{b}\otimes I_{\mathcal{Y}}) V_{B}^{\dagger} = I_{\mathcal{Y}} \otimes N_{b}
\end{equation}
when confined to the support of $\ket{junk_{2}}_{\mathcal{Y}^{2}}$ on the appropriate side.  By virtue of these equations, we find that there must be local unitaries operating on $\mathcal{Y}_{A}$ and $\mathcal{Y}_{B}$ that take $\ket{junk_{2}}_{\mathcal{Y}^{2}}$ to $\ket{junk}_{\mathcal{Y}^{2}}$.  We may absorb these operations into $V_{A} \otimes V_{B}$ and take $\ket{junk_{2}}_{\mathcal{Y}^{2}} = \ket{junk}_{\mathcal{Y}^{2}}$.  Note that this does not disturb equation~\ref{eq:selftestingmeasaftergate3} since the redefinition of $V_{A} \otimes V_{B}$ amounts to conjugating the $\mathcal{Y}^{2}$ register on the right sides.  Thus we find
\begin{equation}
V_{A} \otimes V_{B} M^{\prime\prime}_{a} \otimes N_{b}^{\prime\prime} G^{\prime} \otimes H^{\prime}\ket{\psi} \ket{00} = \ket{junk}_{\mathcal{Y}^{2}}M_{a} \otimes N_{b} \ket{\phi}.
\end{equation}

Now we consider the case when $G^{\prime}$ is applied and $H^{\prime}$ is not applied.  We wish to see what this means if we convert to the reference experiment, so we consider the following state:
\begin{equation}
\ket{\theta} = \left(V_{A} \otimes U_{B}\right)\left( G^{\prime} \otimes I_{\mathcal{Y}} \otimes I_{\mathcal{X}^{2}} \right)\ket{\psi}\ket{00}.
\end{equation}
Applying our equations from above we find that this is equal to 
\begin{equation}
\left(V_{A} \otimes U_{B} \right)\left(G^{\prime} \otimes I_{\mathcal{Y}} \otimes I_{\mathcal{X}^{2}}\right)\left( U_{A}^{\dagger} \otimesÊU_{B}^{\dagger}\right) \ket{junk}_{\mathcal{Y}^{2}} \ket{\phi}
\end{equation}
\begin{equation}
= \left(V_{A} (G^{\prime} \otimes I_{X}) U_{A}^{\dagger}\right) \otimes I_{B} \ket{junk}_{\mathcal{Y}^{2}} \ket{\phi}.
\end{equation}
Define the quantum operation $\Phi$ on $\mathcal{X}$ by attaching an ancilla in the state $\ket{junk}_{\mathcal{Y}^{2}}$, applying $ \left(V_{A} (G^{\prime} \otimes I_{X}) U_{A}^{\dagger}\right)$ and tracing out the $\mathcal{Y}^{2}$ register.  Then $J(\Phi) = \dim(\mathcal{X})\tr[\mathcal{Y}^{2}]{\proj{\theta}{\theta}}$.

We wish to characterize $J(\Phi)$.  To this end we note that $\frac{1}{\dim(\mathcal{X})}\tr{M_{a} \otimes N_{b} J(\Phi)} = \bra{\theta} I_{\mathcal{Y}^{2}} \otimes M_{a} \otimes N_{b} \ket{\theta}$.  Note that
\begin{equation}
V_{B}(H^{\prime}\otimes I_{\mathcal{X}_{B}}) U_{B}^{\dagger}\ket{\theta} = V_{A} \otimes V_{B} G^{\prime} \otimes H^{\prime} \otimes I_{\mathcal{X}^{2}} \ket{\psi}\ket{00} = \ket{junk}_{\mathcal{Y}^{2}}\ket{\phi}_{\mathcal{X}}
\end{equation}
and since this differs from $\ket{\theta}$ only on the $B$ side, the support of $\ket{\theta}$ on $\mathcal{Y}_{B}$ is the same as the support of $\ket{junk}_{\mathcal{Y}^{2}}$ on $\mathcal{Y}_{B}$.  We obtain 

\begin{equation}
\frac{1}{\dim(\mathcal{X})}\tr{M_{a} \otimes N_{b} J(\Phi)} = \bra{\psi} \left((G^{\prime})^{\dagger} \otimes I_{\mathcal{Y}}\right) \left(M^{\prime\prime}_{a}Ê\otimes N^{\prime}_{b}\right)\left( G^{\prime} \otimes I_{\mathcal{Y}} \right)\ket{\psi}.
\end{equation}
Since $G^{\prime} \otimes I_{\mathcal{Y}} \ket{\psi}$ with measurements $M^{\prime\prime}_{a} \otimes N^{\prime}_{b}$ simulates $T \otimes I \ket{\phi}$ with measurements $M_{a} \otimes N_{b}$ we have
 \[
 \tr{J(\Phi) M_{a} \otimes N_{b}} = \tr{ J(T) M_{a} \otimes N_{b}}.
\]
Lemma~\ref{lemma:qubitrealunitarychoi} then implies $J(\Phi) = J(T)$.  From this we conclude that
\begin{equation}
V_{a}Ê\left(G^{\prime} \otimes I_{\mathcal{X}}\right) U_{a}^{\dagger} = W \otimes T
\end{equation}
for some $W \in U(\mathcal{Y})$.  Since $W \otimes T$ preserves $\ket{junk}_{\mathcal{Y}^{2}}$ we must have $W = I_{\mathcal{Y}}$ on the support of $\ket{junk}_{\mathcal{Y}^{2}}$ on $\mathcal{Y}_{A}$.

\subsection{A conspiracy against self-testing complex gates}
The restriction to real gates in the reference experiments may seem a bit curious, but has an easy explanation.  If the gate had complex entries, then there would be no way of distinguishing the reference experiment from an otherwise identical experiment that had the complex conjugate applied.  That is to say, the complex conjugate of an experiment simulates it.  However, the definition of equivalence requires the existence of a unitary operation that takes the physical experiment to the reference experiment, but complex conjugation is anti-unitary.  Thus there can be no extension of the gate testing Theorem to complex gates with the current definition of equivalence.  Real gates are acceptable because complex conjugation does not affect them.

Complex conjugation does not pose very great a problem, since applying the complex conjugate can be seen as a change in convention.  Nothing changes in the structure of the experiment.  However, the simulations described in chapter~\ref{chap:simulations} introduce many new physical experiments that simulate the reference experiment, and none of them are equivalent to the reference experiment.  Further, these simulations are defined on different Hilbert spaces from the reference experiment and hence something more complex than a change in convention is happening.  In the case of real gates, however, the simulations all reduce to the reference experiment, allowing the gate testing Theorem to go through.

Later, in section~\ref{sec:extmayersyao}, we introduce a new definition of equivalence that takes the simulations into account.  We then extend the Mayers and Yao result to include complex measurements.  This paves the way for a gate testing Theorem allowing complex gates.

\section{Circuit testing}

\subsection{Overview}

We now turn our attention to testing an entire circuit.  We begin with a given circuit composed of single qubit real gates and $\text{CTRL-}Z$ gates.  We may divide this up into a sequence of unitaries operating on all qubits.  Each of these unitaries is a tensor product of single qubit and $\text{CTRL-}Z$ gates and thus is efficiently testable, as discussed in section~\ref{sec:selftestgatetest}.  Each of these gates is tested against the reference gate determined by the given circuit.  As explained below, this allows us to conclude that the entire physical circuit is equivalent to the reference circuit when both are considered as one large unitary.  Finally, we perform the circuit by measuring on one side in the computational basis and performing the physical circuit on the other side, as shown in figure~\ref{fig:selftestedcircuit}.

\subsection{Composabiltiy}
In order to test a sequence of gates we perform gate tests after each gate is added.  This is illustrated in figures~\ref{fig:epr_test_with_2gates} and~\ref{fig:2gate_test}.  After each gate is tested, the gate is applied to the other half of the state, taking the state back to (a state equivalent to) the maximally entangled state $\ket{\phi}$.  This allows us to test the next gate in the sequence as though it were applied by itself.  The following Lemma, applied inductively, allows us to conclude that the sequence of gates, applied together, is equivalent to the sequence of reference gates.

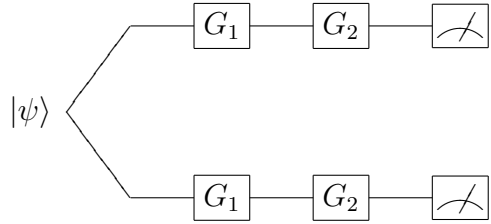
\begin{figure}
 \[
 \Qcircuit {
  & & \gate{G_{1}} & \gate{G_{2}}& \meter \\
\lstick{\ket{\psi}}   \ar@{-}[ur] \ar@{-}[dr]\\
  & & \gate{G_{1}} & \gate{G_{2}} & \meter \\
 }
\]
\caption{EPR test after two gates applied}
\label{fig:epr_test_with_2gates}
\end{figure}

\begin{figure}
 \[
 \Qcircuit {
  & & \gate{G_{1}} &\gate{G_{2}}  &\meter \\
\lstick{\ket{\psi}}   \ar@{-}[ur] \ar@{-}[dr]\\
  & & \gate{G_{1}} & \meter \\
 }
 \]
\caption{Testing a second gate}
\label{fig:2gate_test}
\end{figure}
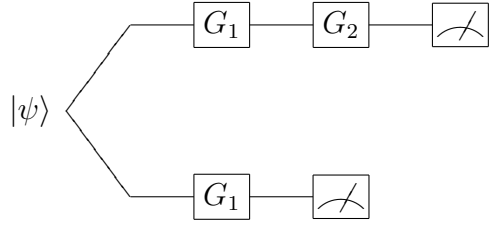

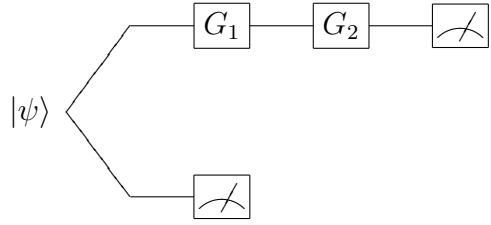
\begin{figure}
 \[
 \Qcircuit {
  & & \gate{G_{1}} &\gate{G_{2}}  &\meter \\
\lstick{\ket{\psi}}   \ar@{-}[ur] \ar@{-}[dr]\\
  & &  \meter \\
 }
 \]
\caption{Performing a self-testing circuit}
\label{fig:selftestedcircuit}
\end{figure}

\begin{lemma}
Suppose we have two sets of experiments satisfying the conditions of Theorem~\ref{theorem:gatetesting}, the first testing $G^{\prime}_{1}$ against unitaries $T_{1}$ and the second testing $G_{2}^{\prime}$ against $T_{2}$ such that experiment 1 of the first set coincides with experiment 2 of the second set.  The conclusions of the Theorem hold with $G^{\prime}_{2}G^{\prime}_{1}$ tested against $T_{2}T_{1}$.
\end{lemma}
\begin{proof}
From Lemma~\ref{lemma:equivmeasidentonjunk} and Theorem~\ref{theorem:gatetesting} we find unitaries $U_{A}, U_{B}, V_{A}, V_{B}, W_{A}, W_{B}$ such that
\begin{equation}
U_{A} (M^{\prime}_{a}\otimes I_{\mathcal{Y}}) U_{A}^{\dagger} = I_{\mathcal{Y}} \otimes M_{a},\,\,  
U_{B} (N^{\prime}_{b}\otimes I_{\mathcal{Y}}) U_{B}^{\dagger} = I_{\mathcal{Y}} \otimes N_{b}
\end{equation}
\begin{equation}\label{eq:selftestingmeasaftergate}
V_{A} (M^{\prime\prime}_{a}\otimes I_{\mathcal{Y}}) V_{A}^{\dagger} = I_{\mathcal{Y}} \otimes M_{a},\,\,  
V_{B} (N^{\prime\prime}_{b}\otimes I_{\mathcal{Y}}) V_{B}^{\dagger} = I_{\mathcal{Y}} \otimes N_{b}
\end{equation}
\begin{equation}\label{eq:selftestingmeasaftergate2}
V_{A} (M^{\prime\prime\prime}_{a}\otimes I_{\mathcal{Y}}) V_{A}^{\dagger} = I_{\mathcal{Y}} \otimes M_{a},\,\,  
V_{B} (N^{\prime\prime\prime}_{b}\otimes I_{\mathcal{Y}}) V_{B}^{\dagger} = I_{\mathcal{Y}} \otimes N_{b}
\end{equation}
\begin{equation}
V_{a}Ê\left(G_{1}^{\prime} \otimes I_{\mathcal{X}}\right) U_{a}^{\dagger} = I \otimes T_{1}
\end{equation}
\begin{equation}
W_{a}Ê\left(G_{2}^{\prime} \otimes I_{\mathcal{X}}\right) V_{a}^{\dagger} = I \otimes T_{2}
\end{equation}
all for the support of $\ket{junk}_{\mathcal{Y}^{2}}$ on $\mathcal{Y}_{A}$.  Note that although we have used the gate testing Theorem twice, we can use the same junk state since experiment 2 from the first use is the same as experiment 1 from the second usage of the Theorem.  We also find
\begin{equation}
U_{A} \otimes U_{B} M^{\prime}_{a} \otimes N_{b}^{\prime} \ket{\psi} \ket{00} = \ket{junk}_{\mathcal{Y}}M_{a} \otimes N_{b} \ket{\phi}
\end{equation}
Applying these results we see that
\begin{equation}
(W_{A}\otimes U_{B})\left(M^{\prime\prime\prime}_{a} \otimes N^{\prime}_{b}\right)\left( G^{\prime}_{2}G^{\prime}_{1} \otimes I_{\mathcal{Y}}\right) \ket{\psi}\ket{00} = \ket{junk}_{\mathcal{Y}^{2}} \left(M_{A} \otimes N_{B}\right) T_{2}T_{1} \otimes I_{\mathcal{X}} \ket{\phi}
\end{equation}
and 
\begin{equation}
W_{a}Ê\left(G_{2}^{\prime}G^{\prime}_{1} \otimes I_{\mathcal{X}}\right) U_{A}^{\dagger} = I \otimes T_{2}T_{1}.
\end{equation}
Thus the conclusions of Theorem~\ref{theorem:gatetesting} hold for the gate $G^{\prime}_{2}G^{\prime}_{1}$.
\end{proof}
%

\subsection{Testing on a particular input}
Since we usually do not want to perform a circuit on a random input, we need some control over what the input will be.  The solution proposed in \cite{Magniez:2006:Self-testing-of} is to measure one side of the EPR pairs in the computational basis and then use the results to place $X$ gates as necessary on the other half of each pair to correct the result to the desired input.  The $X$ gates placed this way are then incorporated into the definition of the circuit to be tested.  Note that for this solution to work the calculation must come first.  With Assumption~\ref{assumption:devicesbehavesame} in place this is not a problem, but if we hope to adopt less restrictive assumptions then this method could easily be defeated since the devices could easily subvert the calculation and  perform correctly on the subsequent tests.

Another solution to this problem assumes that the EPR pairs may be manipulated individually.  Suppose we have $n$ EPR pairs.  Measure half of each EPR in the computational basis.  If the result was the desired result, then keep the other half.  If not, then discard the other half and prepare a new EPR pair.  Repeat until $n$ EPR pairs have been prepared.  Since the EPR pairs are all prepared individually there is no problem with exponential blowup as $n$ increases.  A variation on this solution is to prepare $N > n$ EPR pairs and use only those pairs for which the measurement gives the correct result.  If $N$ is made sufficiently large then success is expected with high probability.

\subsubsection{Blind state preparation}
Although conceived with the notion of circuit testing in mind, self-testing could find application in QKD or other security related areas.  In this context the state preparation method used has some advantages.  For example, by measuring half of an EPR pair to prepare a state on the other half there can be no side channel information about the basis used.  One may wish to choose a particular state, however, so corrections might need to be made.  If one is not careful, the correction may reveal the basis.  For example if the bases are chosen from the eigenbase of $X$ and $Z$, then a $Z$ or $X$ correction, respectively, may be necessary and the correction depends on the basis.  There are several possibilities for defusing the situation.  The first is to use $Y$ as a correction, but this is not currently self-testable.  A second possibility would be to always apply one of $Z$ or $X$.  If the $X$ eigenbasis were chosen and a correction is needed, then $Z$ is applied, otherwise $X$ is applied which does not affect the state.  Other possibilities exist, but the main idea is always to choose an operation that may be a correction for one basis, or not affect the other basis.

\section{Extending the Mayers and Yao self-test}\label{sec:extmayersyao}

The original Mayers and Yao EPR test provided only a small set of measurements.  Conspicuously missing is anything with complex coefficients.  An important consequence of this is that the circuit test is not able to test gates with complex coefficients; only gates with real coefficients can be tested.

In fact the Mayers and Yao test cannot be directly extended to include any measurements with complex coefficients.  This is a result of the notion of equivalence that they use.  Suppose that we wish to include the $Y$ measurement in the set of reference measurements.  The devices could just as easily implement $-Y$, the complex conjugate.  So long as all complex measurements were complex conjugated it would be impossible to tell.  Although this does not present an immediate problem - such a transformation is internally consistent and produces the correct outcome statistics - we cannot transform such a circuit back into the reference circuit using unitary transformations.  Anti-unitary transformations are required.

If this were the whole story we could simply require that the physical circuit be transformable into either the reference circuit or its complex conjugate.  However, the real simulation, and now the family of simulations, defined in section~\ref{sec:generalsimulations}, are also indistinguishable from the reference circuit and not unitarily transformable into the reference circuit.

We have one encouraging fact:  all of the known simulations are equivalent to a simulation from the family of simulations (or a classical mixture of them).  We now prove that we can extend the Mayers-Yao experiments such that these are the only simulations.  Hence we may extend our notion of equivalence to include these simulations and obtain a new self-testing Theorem.

\begin{theorem}\label{theorem:extmayersyao}
Suppose a physical experiment duplicates the statistics generated by the reference experiment described in section~\ref{sec:extmayersyaoref}.  Then the physical experiment is equivalent to one of the simulations of the reference experiment described in section~\ref{sec:generalsimulations}.
\end{theorem}

With the extended state and measurement testing in place there exists the possibility of testing complex gates as well.

\begin{futurework}
Extend gate testing to complex gates using extended definition of equivalence.
\end{futurework}
%

\subsection{Extended Mayers-Yao self-test reference experiment}\label{sec:extmayersyaoref}
The extended Mayers-Yao test will consist of three regular Mayers-Yao tests, performed together.  Alice and Bob will perform the Mayers-Yao test with measurement settings (labelled with subscript $A$ when used by Alice, and subscript $B$ when used by Bob):
\begin{enumerate}
                \item $X$, $Z$, and $D$
                \item $X$, $Y$, and $E$
                \item $Y$, $Z$, and $F$
\end{enumerate}

In the reference experiment the measurement settings $X$, $Y$ and $Z$ are realized by the Pauli operators, with $Y_B=-Y$ and otherwise
$X_A = X_B = X$, $Y_A = Y$, $Z_A  = Z_B = Z$.  The other settings are realized by $D_A = \frac{X+Z}{\sqrt{2}}$, $E_A = \frac{X+Y}{\sqrt{2}}$, $F_A=\frac{Y+Z}{\sqrt{2}}$ on Alice's side and $D_B = \frac{X+Z}{\sqrt{2}}$, $E_B = \frac{X-Y}{\sqrt{2}}$, $F_B=\frac{Z-Y}{\sqrt{2}}$ on Bob's side.  Bob's $Y_B$ measurements all carry the $-1$ phase since measuring the state $\ket{\phi_{+}}$ with the operator $Y\otimes Y$ produces $-1$ instead of $1$ as in the Mayers-Yao reference experiment.  The reference state is again $\ket{\phi_{+}}$.

\subsection{Proof of Theorem~\ref{theorem:extmayersyao}}
We start by assuming that the states are all pure as in the Mayers-Yao test.  Again we may incorporate the purification of a mixed state into either Alice or Bob's state by adding an ancilla.

First we apply the Mayers-Yao result with the measurements $X$, $Z$ and $D$.  We find that we may apply a suitable local isometry $\Phi$ to take the $X_{A}$, $Z_{A}$, $X_{B}$ and $Z_{B}$ measurements to $X_{Q_{A}} \otimes I_{R_{A}}$, $Z_{Q_{B}} \otimes I_{R_{A}}$, $X_{Q_{B}} \otimes I_{R_{B}}$ and $Z_{Q_{B}} \otimes I_{R_{B}}$ where $R_{A}$ and $R_{B}$ are the junk registers.  Meanwhile the state has the form $\ket{\phi_{+}}_{Q_{A}Q_{B}} \otimes \ket{junk}_{R_{A} R_{B}}$, where $Q_{A}$ and $Q_{B}$ are the qubit registers that the measurements act on.

We now consider the remaining measurements.  The reference circuits for these measurements can be transformed using local unitaries into the usual Mayers-Yao reference circuit.  Thus we may apply the result.  However, we stop short of using the full result.  In section~\ref{sec:my_anticommuting} we note that the measurement observables $X_{A}$ and $Z_{A}$ anti-commute on the support of the state, as do $X_{B}$ and $Z_{B}$.  When we apply this result to the remaining measurements in the extended test, we find that $X_{A}$ and $Y_{A}$ anti-commute on the support of the state, as do $X_{B}$ and $Y_{B}$, $Z_{A}$ and $Y_{A}$ and $Z_{B}$ and $Y_{B}$.  For the remaining discussion we will limit ourselves to the support of the state.

Consider the $A$ side measurements first.  We may express $Y_{A}$ as
 \[
Y_{A} = \sum_{P,k} y_{P,k} P_{Q_{A}} \otimes E_{k, R_{A}}
\]
where the $P$s are Pauli operators and the $E_{k}$s are other operators (i.e. pick a basis for Hermitian matrices consisting of Pauli matrices tensor product with something else).  Since $Y_{A}$ anti-commutes with $X_{Q_{A}} \otimes I_{R_{A}}$ all the terms with $P=X$ must be 0.  Indeed, since $-Y_{A} = (X_{Q_{A}} \otimes I_{R_{A}}) Y_{A} (X_{Q_{A}} \otimes I_{R_{A}})$ we have
 \[
-\sum_{P,k} y_{P,k} P_{Q_{A}} \otimes E_{k, R_{A}}
 = \sum_{P\in\{I,X\}, k}y_{P,k} P_{Q_{A}} \otimes E_{k, R_{A}} -  \sum_{P \in\{Y,Z\},k} y_{P,k} P_{Q_{A}} \otimes E_{k, R_{A}}
\]
where on the right hand side we have separated out the terms that commute with $X_{Q_{A}} \otimes I_{R_{A}}$ and those that anti-commute.  We see that we must have $y_{X,k} = -y_{X,k} = 0$ and $y_{I,k} = -y_{I,k} = 0$ for all $k$.

Applying similar reasoning and the test with $Y$ and $Z$ we find that $y_{Z,k} = 0$ for all $k$.  Thus $Y_{A} = Y_{Q_{A}} \otimes M_{R_{A}}$ for some Hermitian and unitary $M_{R_{A}}$.  We consider the two eigenspaces of $M_{R_{A}}$.  If they are not the same dimension (or if there is only one eigenvalue), we may construct an isomorphism that adds extra dimensions to $R_{A}$ and extend $M_{R_{A}}$ onto the new dimensions so that both eigenspaces have the same dimension.  Next we construct an isomorphism that maps the space to a tensor product between a qubit and a space with half the dimension of $R_{A}$.   We construct it so that the +1 eigenspace gets mapped to the subspace spanned by states of the form $\ket{0} \ket{\chi}$ and the -1 eigenspace gets mapped to the subspace spanned by $\ket{1} \ket{\chi}$.  Then $M_{R_{A}}$ gets mapped to $Z \otimes I$.  Let $P_{A}$ be the qubit register, and $R^{\prime}_{A}$ the remaining register.  We obtain
\begin{equation}
Y_{A} \mapsto Y_{Q_{A}} \otimes Z_{P_{A}} \otimes I_{R^{\prime}_{B}}
\end{equation}
under the isomorphism described above.  Importantly this isomorphism does not disturb $X_{A}$ or $Z_{A}$ since it only operates on the junk register.  Thus we may modify $\Phi$ obtained from Theorem~\ref{theorem:mayers-yao} to additionally perform the isomorphism just described.

The above process can be repeated for Bob's side, with analogous conclusions.  In order to be consistent with the reference experiment, we may construct our isomorphism so that $Y_{B} \mapsto -Y_{Q_{B}} \otimes Z_{P_{B}} \otimes I_{R^{\prime}_{B}}$.

Now we turn our attention to the state.  From the Mayers-Yao test on $X$ and $Z$ we know that the state on $Q_{A} \otimes Q_{B}$ (after applying $\Phi$) is $\ket{\phi_{+}}$.  We next consider the state on the remaining registers, which we denote $\ket{\theta}$.  We may express this in the singular value decomposition, split between $P_{AB}$ and $R^{\prime}_{AB}$:
\begin{equation}
\ket{\theta} = \sum_{j} \lambda_{j} \ket{j}_{P_{AB}} \ket{j}_{R^{\prime}_{AB}}
\end{equation}
with $\lambda_{j} > 0$.
Since the $Y$ measurement setting gives correlated results (recall we introduced a -1 factor on the $B$ side measurement observable) and the form of $Y_{A}$ and $Y_{B}$, the states $\ket{j}_{P_{AB}}$ must all be $+1$ eigenvectors of $Z_{P_{A}}\otimes Z_{P_{B}}$.  If this were not the case then a $-1$ phase would be introduced and the measurement results would be incorrect at least some of the time.  Thus the only possible states for $\ket{j}_{P_{AB}}$ are superpositions of $\ket{00}$ and $\ket{11}$.  We do some relabelling and arrive at
 \begin{equation}
\ket{\psi} = \ket{\phi_{+}}_{Q_{AB}} \otimes \left(\alpha \ket{00}_{P_{AB}} \ket{\theta_{00}}_{R^{\prime}_{AB}} + \beta \ket{11}_{P_{AB}} \ket{\theta_{11}}_{R^{\prime}_{AB}}\right)
\end{equation}
with $\ket{\theta_{00}}$ and $\ket{\theta_{11}}$ not necessarily orthogonal.  Note that tracing out the $R^{\prime}_{AB}$ ancillae results in a state of the form in equation~\ref{eq:contsimstate}.  Thus we have demonstrated that the physical experiment is equivalent to the reference experiment, and completed the proof of Theorem~\ref{theorem:extmayersyao}.

\section{Robustness and assumptions for implementations}

The results in this chapter are concerned with probability distributions which exactly match.  However, if there is any hope for a physical implementation then this requirement must be relaxed and the Theorems made robust.  Robustness was established for the Mayers and Yao test, as well as the gate test in \cite{Magniez:2006:Self-testing-of}.  These robustness results show that error in the statistics translates to a polynomial sized error in the equivalence.  This is measured in terms of the $2$-norm on states and the operator norm on gates.  These are not the preferred measures of error since they do not have a straightforward operational interpretation.

\begin{futurework}
Determine robustness of the tests using operationally meaningful measures such as the trace norm on states and diamond norm on operations.
\end{futurework}

Another consideration is that the robustness results relied on technical Lemmas with incorrect proofs.  Also, the robustness of the extended Mayers and Yao test has yet to be determined.

\begin{futurework}
Determine robustness of Lemmas~\ref{lemma:qubitrealunitarychoi}, \ref{lemma:qubitrealunitarychoi2}, \ref{lemma:qubitrealunitarychoi3}, \ref{lemma:equivmeasidentonjunk} and Theorem~\ref{theorem:extmayersyao}.
\end{futurework}

Another important consideration for potential physical implementations is the assumptions necessary to gather statistics.  Here the physical experiment defines a probability distribution on outcomes, but in order to estimate this probability distribution many trials must be made.  For this to make sense, we must make some repeatability assumptions:

\begin{assumption}
The physical devices have no memory and always operate identically and the state for each trial is unentangled with other trials.
\end{assumption}

This allows us to take many samples from a single physical experiment and estimate the probability distribution.  In some situations this assumption may not be reasonable.  One potential means of relaxing this assumption would be to use techniques from QKD proofs, such as the quantum de Finetti Theorem, to allow arbitrary states. 

\begin{futurework}
Relax assumptions on state and devices.
\end{futurework}
%

\section{Authenticated quantum computing}
A recent development with goals very similar to self-testing is blind quantum computing, introduced by Broadbent et al. in~\cite{Broadbent:2008:Universal-blind}.  We are specifically interested in authenticated quantum computing, which is an extension of blind quantum computing.  Authenticated quantum computing involves a semi-quantum verifier (able to only prepare qubits in a finite number of states) and a quantum prover.  Using measurement based quantum computing \cite{Raussendorf:2001:A-One-Way-Quant} and fault-tolerant quantum computing techniques the verifier sends the prover several qubits prepared in random states known only to the verifier and then interacts classically with the prover.  The goal is for the verifier to have the prover perform a quantum circuit and be able to certify, through classical interaction only, that the correct circuit has been performed.  An important side effect of the process is that the prover does not know what the circuit is, even at the end of the protocol (hence \emph{blind} computing.)  

The goal is similar to that of self-testing, but the verifier requires some quantum capacity.  In an extension of their result, the authors claim that the verifier can interact with \emph{two} isolated provers (who are entangled) and eliminate the requirement for state preparation by the verifier.  The idea is to begin with the two provers sharing a number of EPR pairs.  The verifier first interacts with one prover, using the authenticated blind quantum computing protocol (without first sending qubits) to implement a circuit that simply measures half of each EPR pair in a randomly chosen bases, emulating the verifier's state preparation.  Next, the verifier interacts with the second prover, again using the authenticated blind quantum computing protocol, performing the desired circuit with the other half of each EPR pair.  A similar result was shown by Aharonov et al. in \cite{Aharonov:2008:Interactive-Pro}.

The claim is that the verifier cannot distinguish between errors occurring in each prover, and so we can assume that all errors happen in the second prover.  Then the authenticated blind quantum computing protocol used with the second prover will catch all errors.  Hence a pair of cheating provers will be caught (with high probability).

From the perspective of black-box quantum computing we may identify two important problems with this argument.  The first is that the argument is not sound.  The simulations in Chapter~\ref{chap:simulations} indicated that it is possible for the two provers to perform a conspiracy that produces the correct outcome statistics, but does not implement the reference circuit.  In particular, the measurement based quantum computing model used in the protocol requires operations that are complex.  Thus the general simulations are not unitarily equivalent to the reference circuit.  We may view the conspiracy as ``errors'', in which case the claim that all errors will be caught during the interaction with the second prover does not hold:  there are conspiracies in which \emph{both} provers perform ``errors'' that cannot be caught at all.

One counterargument to the above criticism is that all the general simulations would still provide the correct classical outcome, regardless, and hence the protocol is still sound.  To be clear, we do not claim that the protocol is not sound, only that the proof is not sufficient.  As well, this simply highlights the second problem, which is a lack of a rigorous claim.  The final claim is that any language in the complexity class BQP has an interactive proof with a BPP verifier and two non-interacting BQP provers.  However, this is not immediate since the authenticated blind quantum computing protocol is not about recognizing languages in BQP, but performing quantum circuits.  The implied mediating claim is that the two prover protocol, for any desired reference circuit, certifies that the reference circuit was performed, hence any language in BQP may be recognized by the protocol.  

Here it is unclear how to precisely say that the reference circuit was performed.  The obvious interpretation, that the initial state was correct and each gate and measurement in the circuit was performed, is clearly not sufficient, since the general simulations (and even the unitarily equivalent simulations, if we are to be precise) defeat this claim.  From the perspective of self-testing we may offer a more suitable claim:  that the physical experiment was unitarily equivalent to either the reference experiment or one of the general simulations of it.  As our counterargument above shows, this is the most that can be established.  Fortunately, it is sufficient to imply the desired final claim:  that all languages in BQP have classical interactive proofs with two non-interacting BQP provers.  Unfortunately, there is currently no proof of such a claim.

\begin{futurework}
Show that the authenticated blind quantum computing protocol, with two entangled provers, certifies that a general simulation of the desired reference circuit was implemented by the provers.
\end{futurework}
%

\chapter{Device independent quantum key distribution}\label{chapter:diqkd}
\minitoc

\section{Introduction}
Traditional \index{quantum key distribution}quantum key distribution protocols, such as BB84 \cite{Bennett:1984:Quantum-cryptog} and Ekert91 \cite{Ekert:1991:Quantum-cryptog} rely on a model of the physical devices being used in order to determine a secure key rate.  In prepare and measure protocols, for example, a model of the source is used to determine to what extent Eve may differentiate between the various states sent and in all protocols a model of the measurements performed is incorporated into the parameter estimation portion of the protocols, deriving estimates of the states received.

In contrast, \index{device independent quantum key distribution (DIQKD)}device independent quantum key distribution (DIQKD) aims to provide security without relying on a particular physical device model.  The intent is to provide a higher level of security.  Physical devices used for implementing QKD protocols are vastly more complicated than the simple physical models used in security proofs, allowing for a mismatch between theory and reality.  If the security models are not conservative enough this may lead to an insecure physical implementation of a theoretically secure protocol.

In this chapter we describe the device model used in DIQKD and discuss the existing literature on the subject.  In particular, we will be interested in the line of inquiry leading to the AMP06 protocol \cite{Acin:2006:Efficient-quant} and subsequent reinterpretation in the DIQKD framework.  We then consider the previous security models and provide some partial results extending security to a more general model.

The original material in this chapter is published in \cite{McKague:2009:Device-independ}.

\section{Literature review}
For the current work we are interested in four different lines of research.  The first line of research is that of non-signalling based key distribution, from which the AMP06 protocol is drawn.  The second is a pair of articles that introduce the notion of DIQKD and give a partial proof of security of the AMP06 protocol within the DIQKD framework.  For our expanded proof we require some techniques from the literature on QKD security proofs, particularly those of Renner.  Finally, it will be useful to review some concepts from the literature on Bell inequalities.

\subsection{\index{non-signalling key distribution}Non-signalling key distribution}

The AMP06 protocol was introduced in \cite{Acin:2006:Efficient-quant}, which belongs to the literature on non-signalling based key distribution protocols.  These protocols do not rely on quantum mechanics being correct for their security.  Rather, they consider a wider context of probability distributions which are limited by being non-signalling.  

\subsubsection{\index{non-signalling distribution}Non-signalling distributions}
Consider a probability distribution $P(x,y | a,b)$ which assigns a probability to outcomes $x$ and $y$ for each inputs $a$ and $b$.  The inputs $a$ and $b$ are analogous to measurement settings in the usual QKD framework, with $x$ and $y$ the measurement outcomes.  We associate the variables $x,a$ with one location, controlled by Alice, and $y,b$ with another location, controlled by Bob.

We define a probability distribution of this form to be non-signalling according to the following definition

\begin{definition}\label{def:nonsignallingdistribution}
A probability distribution $P(x,y|a,b)$ is \emph{\index{non-signalling distribution}non-signalling} if for every $a, b, x, y$
\begin{enumerate}
	\item $P(x|a,b) = P(x|a)$
	\item $P(y|a,b) = P(y|b)$.
\end{enumerate}

\end{definition}

The \index{marginal distributions}marginal distributions above are found by summing over the possible values of the variables not mentioned.  For example
\begin{equation}
P(x|a) = \sum_{y^{\prime},b^{\prime}} P(x,y=y^{\prime}|a,b=b^{\prime}).
\end{equation}
The definition of a \index{non-signalling distribution}non-signalling distribution says that the distribution of one outcome is not dependent on the measurement setting in the other location.  Such distributions are consistent with general relativity and cannot be used to transmit information.

\index{non-signalling key distribution}Non-signalling key distribution makes use of different principles for security than does QKD.  While QKD uses measurements to perform some quantum state estimation (typically measuring deviation from a maximally entangled state), non-signalling key distribution uses correlation functions on probability distributions (measuring deviation from a non-local probability distribution).  

Often, \index{Bell inequality}Bell inequalities are central in the discussion of non-signalling key distribution, but the security of such schemes usually depends, not on the inequality itself, but on an analysis of the correlation functions that Bell inequalities bound.  The Bell inequality provides a lower bound on the strength of correlations necessary for secure key generation, since local distributions cannot generate secure key.  However, the bound may not be tight, as in the CHSH based protocol considered by Masanes \cite{Masanes:2008:Universally-com} where even quantum correlations give a zero secure key rate.  

\subsubsection{Non-signalling literature}
Non-signalling key distribution was anticipated by Barrett et al. \cite{Barrett:2005:Nonlocal-correl} with a study on the value of \index{non-signalling distribution}non-signalling distributions as information theoretic resources that may be converted between one another much like how different entangled quantum states may be  converted to one another using LOCC operations.  Later, Barrett et al. \cite{Barrett:2005:No-Signaling-an} introduced a proof-of-concept protocol which uses many trials to estimate the expected value of a correlation function and produces a single bit of secure key.  Their protocol is not robust against noise and is inefficient in the use of the channel.  However, this early work opened up the area to further research.

Theoretical work in this area continued with results by Barrett et al. on monogamy of maximally entangled quantum states \cite{Barrett:2006:Maximally-Nonlo}.  They introduce a correlation function similar to that used in the chaining inequality \cite{Braunstein:1990:Wringing-out-be} and show that if two parties in a 3-partite non-signalling distribution produces the same correlations  achieved by a $d \times d$-dimensional maximally entangled quantum state (measured by the new correlation function) then the third party can have no information about the measurement outcomes of the first two parties.  This opens up the potential for a higher dimensional non-signalling key distribution protocol implementable by quantum apparatus.

New protocols were introduced in \cite{Acin:2006:From-Bells-Theo} and \cite{Acin:2006:Efficient-quant}.  Both of these protocols are based on the correlation function in the CHSH inequality \cite{Clauser:1969:Proposed-Experi} (see section~\ref{sec:chshinequality}). The latter is the AMP06 protocol which we will consider in detail in this chapter, and is a refinement of the protocol in  \cite{Acin:2006:From-Bells-Theo}.   

All security proofs up to this point were concerned with individual or collective attacks only.  Proofs of security against general attacks were developed in \cite{Masanes:2006:Security-of-key}, with universal composability achieved in \cite{Masanes:2008:Universally-com}.  However, the proof makes some impractical assumptions.  In particular, the probability distribution is assumed to be $n+1$-fold non-signalling, with one party controlling $n$ parts (corresponding to $n$ trials) and the other controlling $1$ part (the $n$ trials need not be non-signalling for the second party.)

\subsection{DIQKD}
Non-signalling protocols and security proofs perform poorly in a quantum context.  As an example, we consider the two protocols analyzed by Masanes in \cite{Masanes:2008:Universally-com}.  The first protocol analyzed is that presented in \cite{Acin:2006:Efficient-quant}, which is also the first (and so far only) DIQKD protocol.  This protocol has secure key rate 0 in the non-signalling framework when implemented by quantum devices.  The non-local correlations required for key generation cannot be generated by quantum devices.  The second protocol relies on a larger number of measurements.  For a particular \index{Bell inequality}Bell inequality (the Braunstein-Caves inequality \cite{Braunstein:1990:Wringing-out-be}) quantum devices may (asymptotically in the number of measurements) achieve the same correlations as arbitrary non-signalling distributions.  This allows a quantum implementation that achieves the full non-signalling secure key rate, but practical implementation is problematic due to the large number of measurements required.

Non-signalling based QKD is likely limited to a theoretical context because of its impracticality in a quantum setting, but its reliance on a smaller set of assumptions than traditional QKD is appealing.  For this reason Acin et al. \cite{Acin:2007:Device-Independ} \cite{Pironio:2009:Device-independ} reinterpreted the protocol of \cite{Acin:2006:Efficient-quant} in a quantum setting.  The result is DIQKD.  The initial work in this area consists of a proof of security in the quantum setting against a limited class of attacks.  This class of attacks is analogous to the collective attack model in QKD and bears the same name.  In this model the state is a tensor product $\rho^{\otimes n}$ of $n$ identical states.  This is measured by measurement devices with a fixed (but unknown) operation.  That is, the measurement operators for each measurement setting are fixed.

\subsection{Security proofs in QKD}
The results in this chapter rely on Renner's PhD thesis on the security of traditional QKD \cite{Renner:2005:Security-of-Qua}.  Renner's work provides a robust framework for security proofs of QKD protocols against general attacks.  Renner's work is notable for several reasons.  First, he adopts a \index{composable security}composable security definition, which means that the final key is secure for any application \cite{Konig:2007:Small-Accessibl}.  Second, he develops and uses a finite version of the quantum de Finetti Theorem \cite{Renner:2007:Symmetry-of-lar}, which allows the security proofs to be applicable to general attacks in which the combined state across all measurements is arbitrary.  Another important contribution is Renner's development of smooth min- and max-entropies, which play a major role in his security proofs and allow for finite, rather than asymptotic, analysis of security.

In order to use Renner's framework with a particular QKD protocol one must consider the states and measurements used in the protocol and determine two things.  The first is the set of states that pass the parameter estimation phase.  These are found by considering the measurements used during parameter estimation along with various security parameters.  Once this set of states is found one must determine the minimum secure key rate (found by calculating conditional entropies on the state) over all states that pass the parameter estimation phase.

\subsection{\index{CHSH inequality}CHSH inequality}\label{sec:chshinequality}
The main idea for the AMP06 protocol is foreshadowed in Ekert's work on entanglement based QKD protocols \cite{Ekert:1991:Quantum-cryptog}.  Ekert's protocol, Ekert91, used the CHSH inequality \cite{Clauser:1969:Proposed-Experi} (or rather the correlation function that the CHSH inequality bounds) to estimate how close the measured state is to a pair of maximally entangled qubits.  This estimate was then used to prove that a secure key may be extracted.  However, the state estimate is determined assuming that the measurement devices exactly implement the Pauli $X$ and $Z$ basis measurements.  The AMP06 protocol retains the use of the CHSH inequality but uses the black box device model.

The \index{CHSH inequality}CHSH inequality is a \index{Bell inequality}Bell inequality utilizing two measurement settings and two measurement outcomes for two parties.  The two parties, Alice and Bob, each randomly apply one of the two measurement operators to a bipartite state $\rho$ and compare outcomes.  The measurement operators are $A_{a}$ and $B_{b}$, where $a,b \in \{0,1\}$ are the measurement settings for Alice and Bob, respectively.  The operators $A_{a}$ and $B_{b}$ are Hermitian with eigenvalues 1 and -1.  The CHSH operator is a non-local measurement defined by
\begin{equation}\label{eq:chshoperator}
 CHSH = \sum_{a,b = 0,1}A_{a} \otimes B_{b} (-1)^{ab}.
\end{equation}
The CHSH inequality may be expressed as 
\begin{equation}
  S = \tr{CHSH \rho} = \sum_{a,b = 0,1}\tr{A_{a} \otimes B_{b} \rho} (-1)^{ab} \leq 2
\end{equation}
for local classical strategies, with an upper bound of $2 \sqrt{2}$ for quantum strategies.  Equivalently, we may use uniformly distributed random variables $a,b \in \{0,1\}$ for the measurement settings and random variables $x, y \in \{0,1\}$ for measurement outcomes, and derive the inequality
\begin{equation}
p =  P\left(x \oplus y = ab \right) \leq 0.75
\end{equation}
for local classical strategies, with an upper bound of $\cos^{2} \frac{\pi}{8} \approx 0.85$ for quantum strategies.  We say that a trial is successful if $x \oplus y = ab$.  In this notation the scenario may be described as a binary XOR game in which a referee supplies uniformly distributed queries $a$ and $b$ and receives replies $x$ and $y$.  Alice and Bob win the game if $x \oplus y = ab$.

The values $p$ and $S$ are related by
\begin{equation}
 S = 8p  - 4.
\end{equation}
Both of these quantities will be useful in this paper.  We will be interested in the maximum value of $S$ or $p$ achievable by a state $\rho$, maximized over all possible measurements.  We denote these values by $S_{max}(\rho)$ and $p_{max}(\rho)$.

Later, we need to determine $S_{max}(\rho)$ for a pair of qubits.  For this, we turn to Horodecki et al. \cite{Horodecki:1995:Violating-Bell-} who did exactly the required calculation.  Later, Verstraete and Wolf \cite{Verstraete:2002:Entanglement-ve} gave a different presentation of this calculation, which we will use here.

We begin by writing
\begin{equation}
\rho = \sum_{U,V \in \{I, X, Y, Z\}} \frac{R_{U,V}}{4} U \otimes V
\end{equation}
with
\begin{equation}
R_{U, V} = \tr{\rho UÊ\otimes V}.
\end{equation}
We define a matrix $R^{\prime}$, with rows and columns indexed by $X,Y,Z$ and entries $R_{U,V}$.  Meanwhile, we may write the measurement operators as
\begin{equation}
A_{a} = \sum_{U \in \{X,Y,Z\}} s_{a,U} U
\end{equation}
\begin{equation}
B_{b} = \sum_{V \in \{X,Y,Z\}} t_{b,V} V.
\end{equation}
We further define a matrix $M$ by
\begin{equation}
M = \left(
	\begin{matrix}
	s_{0} & s_{1}
	\end{matrix}
\right)
\left(
	\begin{matrix}
	1 & 1 \\
	1 & -1 \\
	\end{matrix}
\right)
\left(
	\begin{matrix}
	t^{T}_{0} \\ t^{T}_{1} \\
	\end{matrix}
\right),
\end{equation}
which is constrained by having $\tr{M^{T}M} = 4$ and $\text{Rank}(M) = 2$.  Then the value of the CHSH operator may be written as $\tr{R^{\prime}M}$.  Using standard optimization techniques (least squares approximation) we find the maximum to be $2 \sqrt{u^{2} + v^{2}}$ where $u$ and $v$ are the largest singular values of $R^{\prime}$ (or square roots of the eigenvalues of $R^{\prime}(R^{\prime})^{T}$.  In the case where $R^{\prime}$ is diagonal, $u$ and $v$ are the two largest (in absolute value) of the diagonal entries.

\subsection{The AMP06 protocol}
The AMP06 protocol was originally described in \cite{Acin:2006:Efficient-quant} and shown to be secure against collective quantum attacks in \cite{Acin:2007:Device-Independ} and \cite{Pironio:2009:Device-independ}.  Two parties, Alice and Bob, share a small amount of secret key and wish to expand this into a larger key.  They have access to an uncharacterized device which emits bipartite states, connected by quantum channels to a pair of uncharacterized measurement devices.  Alice's measurement device has three settings, while Bob's has two.  Finally, they have access to an insecure classical channel.  They use some secret key to authenticate data sent on the classical channel.

\begin{enumerate}
	\item Before beginning, Alice randomly chooses a list of $m$ trials to be used for parameter estimation which she sends to Bob encrypted, using some private key bits.
	\item For each trial, Alice and Bob request a state from the source.  If the trial is to be used for parameter estimation, Alice and Bob choose their measurement settings uniformly at random from $\{0,1\}$.  Otherwise Alice chooses setting 2 and Bob chooses setting 0.
	\item After all trials are completed, Alice and Bob announce their measurement settings.
	\item Alice and Bob publicly announce a permutation and reorder their trials according to this permutation.
	\item Alice and Bob estimate $S$, the CHSH value, from the parameter estimation trials.
	\item Alice and Bob perform error correction on the remaining trials, correcting Alice's outcomes to correspond with Bob's, resulting in the raw key.
	\item Alice and Bob perform privacy amplification on the raw key according to the secure key rate predicted by $S$.
\end{enumerate}

The above protocol could be efficiently implemented using quantum apparatus by a source of qubit pairs in the state $\ket{\phi_{+}} = \frac{1}{\sqrt{2}}\ket{00} + \frac{1}{\sqrt{2}}\ket{11}$, with Alice's measurements given by the operators $X$, $Y$, and $\frac{X + Y}{\sqrt{2}}$.  Bob's measurement operators are $\frac{X + Y}{\sqrt{2}}$ and $\frac{X - Y}{\sqrt{2}}$.  The security comes from the fact that in order to achieve a high value of $S$, the state that Alice and Bob measure must be close to $\ket{\phi_{+}}$ and hence Bob's measurements are uncorrelated with Eve.  The efficiency of the protocol comes from the fact that Alice can align her measurement with Bob's a significant amount of the time and obtain strongly correlated results, so long as she chooses the other measurements often enough to detect any deviation in the state from $\ket{\phi_{+}}$.

Instead of choosing which trials to use for parameter estimation in advance, Alice and Bob may choose their settings independently, saving some key.  This introduces trials which are unusable (when Alice chooses 2 and Bob chooses 1) and unless Bob chooses 0 and 1 uniformly, there will be some parameter estimation settings that occur more than others.  Conceptually it is easier to suppose that the parameter estimation trials are first chosen and then the settings are chosen uniformly.

In \cite{Acin:2007:Device-Independ} and \cite{Pironio:2009:Device-independ} the protocol requires that Alice and Bob symmetrize their data by flipping their outcomes according to a random string which is publicly broadcast.  This simplifies the analysis by introducing symmetries in the quantum state.  However, the symmetrization procedure need not be done in practice since it does not change the amount of information leaked to an adversary;  Eve may account for the symmetrization in her own analysis after observing the public random string.  Here we omit the symmetrization.

\section{Security models}

\subsection{\index{DIQKD black box model}Black box model}
DIQKD uses a black box model of quantum devices as in self-testing.  The devices are considered to be adversarial, always operating in such a way as to maximize the information leakage to Eve.  Of course we must place some restrictions on the devices, otherwise they may simply transmit all their information to Eve.  We require:

\begin{assumption}\label{assumption:diqkdblackbox}
	The measurement devices do not leak any information to Eve.
\end{assumption}

The black box measurement devices have a quantum input, a classical input (measurement setting) and a classical output (measurement outcomes).  Alice (or Bob for his device) has exclusive control of the classical input and output, and there are no other side channels.  The quantum input is strictly input only.  There can be no quantum or classical states ``leaking'' from the quantum input.  We may model each device as a quantum channel that has two input registers, one the quantum input and the other the measurement setting, and one output register.  Alice possesses both classical registers.  

The device model for DIQKD is sometimes described as ``Eve provides the measurement devices.''  However, this depiction is only applicable if we can for some reason trust Eve not to build the device in such a way that it leaks information.  Instead, the model should be understood as a theoretical tool which eliminates the dependance on a particular physical model.  In a secure physical implementation there must still be a physical model which makes assumption~\ref{assumption:diqkdblackbox} reasonable.

\subsection{\index{collective attack model}Collective attack model and security}

As described above, the protocol could be performed using the same devices over and over.  Pironio et al. (\cite{Pironio:2009:Device-independ}) originally considered security against collective attacks, which relies on the assumption that the devices operate identically each time, and have no memory of the previous trials.  For the source this means that state emitted over $n$ trials has the form $\rho^{\otimes n}$.  A physical implementation using devices that are used repeatedly must meet the following assumptions

\begin{assumption}[DIQKD Collective attacks]

\noindent
\begin{itemize}
	\item On each trial the source emits $\rho$.
	\item The combined state that the source emits is $\rho^{\otimes n}$.
	\item The measurement devices have no memory.
\end{itemize}
\end{assumption}

 Pironio et al. proved security in this model with a secure key rate depending on $S$ and the error rate between the measurement outcomes for setting $2$ on Alice's side and setting $0$ on Bob's side.
 
\begin{theorem}[Pironio et al. \cite{Pironio:2009:Device-independ}]\label{theorem:diqkdcollectivesecurity}
The AMP06 protocol is secure against collective quantum attacks with secure key rate
\begin{equation}
1 - h \left(\frac{1 + \sqrt{(S/2)^{2} - 1}}{2} \right) - h(q).
\end{equation}
where $S$ is the CHSH value and $q$ is the bit error rate.
\end{theorem}

\subsection{\index{memoryless device attack model}Memoryless device attack model}\label{sec:memorylessattackmodel}
In this chapter we will give a partial result on security in a more relaxed model than the collective attack model.  We may describe it in two different ways:  in a serial or parallel fashion, which are equivalent given certain assumptions.  In order to illustrate the difference and how the current model differs from the collective attack model we give two descriptions of the latter.  In the collective attack model there are several trials, and each trial is identical; the states are identical and independent, and the measurement applied for each trial (for a given measurement setting) is the same.  In a physical implementation we may consider two means of achieving this scenario.  The first is a serial model where a source emits identical and independent copies of a particular state, which is measured by a measurement device which operates identically (for a given setting) each time it is used.  In a parallel model, by contrast, each trial is implemented by a separate physical system (prepared in identical states) which is measured by separate devices (each operating identically).  Clearly for the collective attack model these two physical depictions carry no theoretical difference.  Physical implementation is of course easier with a serial model, while security proofs for QKD protocols typically rely on the parallel model.

We now move to the memoryless device attack model.  The goal is to achieve as general a security proof as possible.  We consider first the parallel model.  The most general attack model would be to allow any state and any measurement.  Taken to the extreme, one may consider a single large POVM, dependent on the measurement settings for all trials, outputting the results for all trials simultaneously.  A more restrictive model would have the measurement for each trial arbitrary, but operating on a separate physical system.  The separation could be enforced by some type of shielding, which is already necessary to obtain Assumption~\ref{assumption:diqkdblackbox}.

Clearly a large number of space-like separated measurements is not practical.  A practical implementation could be made with single devices used serially, with only the following assumption:

\begin{assumption}[DIQKD Global attacks with memoryless devices]

\noindent
\begin{itemize}
	\item The measurement devices have no memory.
\end{itemize}
\end{assumption}

\noindent  Suppose we operate the memoryless measurement devices in a lockstep fashion with the measurement settings so that the next measurement setting is only given to the device once the result of the previous trial has been given.  In this case, since the devices have no memory, the various trials are completely independent and the measurements on each trial commute with all other measurements.  Thus this model is equivalent to the parallel model with measurements operating on separate physical systems.

Thus we arrive at the memoryless device attack model.  The source may emit any type of state, which may include a complete specification on how the measurement devices are to operate on a particular trial, and the state may be entangled between trials.  There is no restriction on the dimension of the state or on the form of the measurement operators.  However, the measurement devices have no memory.

\section{Security of AMP06}

The main result in this chapter is to give a partial result showing that the AMP06 protocol described in \cite{Acin:2007:Device-Independ} is secure in the memoryless device attack model.  Unfortunately we are not able to give a full proof, and instead give a partial result which depends on a conjecture.  To be precise, we give a proof of the following Theorem:

\begin{theorem}
The AMP06 protocol is secure against qubit strategies with a symmetric state and memoryless measurement devices with secure key rate  
\begin{equation}
1 - h \left(\frac{1 + \sqrt{(S/2)^{2} - 1}}{2} \right) - h(q).
\end{equation}
where $S$ is the CHSH value and $q$ is the bit error rate.
\end{theorem}

We also outline how this can might be extended to qubit strategies with an arbitrary state and memoryless measurement devices.  Unfortunately, the proof is incomplete:
\begin{conjecture}\label{conjecture:permutesym}
Permuting trials of a qubit strategy is equivalent to a qubit strategy with a symmetric state and a lower $S_{max}$.
\end{conjecture}

Applying techniques from \cite{Pironio:2009:Device-independ} the result can then be extended to strategies with a state of any dimension.

\begin{theorem}\label{theorem:diqkdmemorylesssecurity}
If conjecture~\ref{conjecture:permutesym} is true, then the AMP06 protocol is secure against quantum attacks with memoryless measurement devices. 
\end{theorem}

\subsection{Proof overview}

We will make extensive use of Renner's framework for QKD security proofs (\cite{Renner:2005:Security-of-Qua} chapter 6), but it will require adaptation in order to be applicable within the black-box model.  In particular the parameter estimation in Renner's framework assumes that the same measurement (for a given setting) is applied for every trial, which cannot be assumed within the black-box model.  Also, the finite de Finetti Theorem is sensitive to the dimension of the Hilbert space.  Since in the black-box model the Hilbert space is unknown, we cannot use the finite de Finetti Theorem directly to obtain any bounds.

In addition to solving the above problems, we must also characterize the set of states that pass the parameter estimation phase (in this case, a Bell inequality) and determine the minimum key rate for these states.  For this stage we will make use of state parameterization and entropy bounds from \cite{Pironio:2009:Device-independ}.

\subsubsection{Overview of Renner's security proof}
Since our proof of Theorem~\ref{theorem:diqkdmemorylesssecurity} is an adaptation of Renner's security proof, we sketch the steps in that proof here:

\begin{enumerate}
	\item Permute trials to obtain a symmetric state
	\item Apply the finite de Finetti Theorem
	\item Measure $m$ trials and apply parameter estimation Lemma
	\item Measure remaining trials to obtain the raw key
	\item Estimate the min-entropy of the raw key
	\item Apply classical post-processing to obtain final key
\end{enumerate} 

The final security claim consists of an estimate of the trace distance (induced by the $1$-norm) between the processed measurement outcomes and a uniformly random key independent of Eve.  In order to obtain this estimate we make use of two tools: the triangle inequality and the fact that trace distance is non-increasing under quantum operations.  This produces a chain of inequalities finally ending with the security claim.

\subsubsection{Outline of proof of Theorem~\ref{theorem:diqkdmemorylesssecurity}}
Our proof follows the same sketch as that of Renner's security proof, but we must make adaptations at all but the last step:

\begin{enumerate}
	\item Reduce arbitrary strategies to qubit strategies
	\item Show permuting trial outcomes is sufficient in the black box model
	\item Prove new parameter estimation Lemma for black box model and CHSH value
	\item Minimize min-entropy over possible measurements
	\item Estimate min-entropy from CHSH value estimate
	\item Apply classical post-processing to obtain final key
\end{enumerate}

The physical model we use in the proof is that of many parallel trials where the state for each trial is contained in a pair of subsystems (one for Alice and one for Bob), as introduced in section~\ref{sec:memorylessattackmodel}.  Some careful thought will show that all the procedures used can be either serialized or performed once all the quantum systems have been measured solely using the classical data.

The remainder of this section is divided into subsections devoted to each of the above steps. 

\subsection{Reduction to qubit strategies}

Before we can use Renner's QKD proof framework we must first fix the dimension of the subsystems.  This is because the finite de Finetti Theorem, described below in section~\ref{sec:symmetricstates}, is sensitive to the dimension.  In particular, if the dimension is unbounded then no conclusion may be drawn.  Since we have no \emph{a priori} bound on the dimension, we must make some form of reduction.  Our main tool will be the following Lemma, which is originally due to Jordan \cite{Jordan:1875:Essai-sur-la-ge}, but has been rediscovered many times.  Modern proofs appear in \cite{Masanes:2006:Asymptotic-Viol} and \cite{Pironio:2009:Device-independ}.  We will use the formulation appearing in \cite{Pironio:2009:Device-independ}.

\begin{lemma}[Pironio et al. \cite{Pironio:2009:Device-independ} Lemma 2]\label{lemma:blockdiagonalization}
Let $A^{0}$ and $A^{1}$ be two operators on $\mathcal{H}$ with two eigenvalues.  Then $A^{0}$ and $A^{1}$ can be simultaneously block diagonalized with block sizes $2 \times 2$ and $1 \times 1$.
\end{lemma}

\begin{corollary}\label{lemma:qubit_observables_direct_product}
Let $A^{0}$ and $A^{1}$ be two Hermitian operators on $\mathcal{H}$ with dimension $2n$ or $2n - 1$ and eigenvalues 1 and -1, then there exists an isometry $F$ from $\mathcal{H}$ to $\mathcal{H}_{n} \otimes \mathcal{H}_{2}$ and Hermition operators $A^{a,z}$ on $\mathcal{H}_{2}$ with eigenvalues 1 and -1, such that
\begin{equation}
F(A^{a}) = \sum_{z}\proj{z}{z} \otimes A^{a,z}
\end{equation}
with $A^{a,z}$ $2 \times 2$ operators with eigenvalues in $\{1, -1\}$.
\end{corollary}

This corollary says that we can think of applying one of these two observables as first applying a projection to learn $z$.  The value of $z$ then simultaneously determines a measurement strategy for either measurement setting.  Importantly, the projection onto $z$ can be applied before learning the measurement setting.  This will allow us to consider an arbitrary strategy as a probabilistic combination of qubit strategies.

Let $A_{j}^{a}$ be the observable for Alice	's mesaurements on the $j$th trial with setting $a$, and analogously for Bob.  We apply corollary \ref{lemma:qubit_observables_direct_product} to pairs of observables  $A_{j}^{0}$ and $A_{j}^{1}$  to obtain isometry $F_{j}$, from the Hilbert space of the original state to $\mathcal{Z}^{A}_{j} \otimes \mathcal{H}_{2}$.  The result is that we can map $A_{j}^{a_{j}}$ to 
\begin{equation}
 \sum_{z_{j}} \Pi_{z_{j}}^{j} \otimes A_{j}^{a_{j},z_{j}}
\end{equation}
with the $\Pi_{z_{j}}^{j}$ commuting for different $j$.  We do the same with observables $B_{j}^{0}$ and $B_{j}^{1}$ and map $B_{j}^{b_{j}}$ to
\begin{equation}
 \sum_{w_{j}} \Pi_{w_{j}}^{j} \otimes B_{j}^{b_{j},w_{j}}
\end{equation}

At this point we may decompose Eve's strategy into qubit strategies, indexed by $z = (z_{1}, \dots, z_{n})$ and $w = (w_{1} \dots w_{n})$.  However, we make a further simplification.  It may be the case that some $A_{j}^{a,z_{j}}$ is either $I$ or $-I$, and the measurement outcome is fixed.  We may replace the state for Alice's $j$th qubit in each strategy $(z,w)$ for this value of $z_{j}$ with $\ket{0}$ or $\ket{1}$ and replace $A_{j}^{a,z_{j}}$ with $Z$.  This new strategy is identical in terms of Eve's information and the outcomes as the previous strategy.  Applying this reduction many times we obtain a strategy in which each measurement operator is 2-outcome.  This is important since our proof of security for qubit strategies will make this assumption.

As a final reduction, we note that different trials may have different sized Hilbert spaces.  We may map the Hilbert space for each trial to the maximum sized Hilbert space, leaving the state intact and extending the measurement operator onto the extra dimensions in any arbitrary fashion.  This makes no change in Eve's information or the measurement outcomes.

We have mapped a strategy of Eve to a strategy with state $\rho$ on Hilbert space $\mathcal{Z} \otimes (\mathcal{H}_{2}^{\otimes n})_{A} \otimes (\mathcal{H}_{2}^{\otimes n})_{B}$ with measurement operators of the form above.  Note that we may perform a projective measurement with projectors $\Pi_{z_{j}}^{j}$ for each $j$ to determine all the $z_{j}$ and analogously for Bob's side to determine the $w_{j}$s before determining the measurement setting without changing anything, since these projectors commute with the measurements $A_{j}^{a,}$ and $B_{j}^{b}$.  Eve loses nothing by performing this measurement herself, so we may assume that she does so and learns $(z,w)$.  We may also suppose that we first project the state down to a block diagonal state since this operation commutes with measuring $(z,w)$.  The result is that any strategy is equivalent to one in which Eve prepares a mixture of qubit strategies.  We may further suppose that Eve holds the purification for each possible qubit strategy and only increase her power.

We have reduced all possible strategies to a mixture of strategies on qubits.

\subsection{Reduction to symmetric qubit strategies}\label{sec:symmetricstates}

\subsubsection{Symmetric states and the de Finetti Theorem}
A \index{symmetric state}symmetric state on $n$ subsystems is a state that is invariant under permutation of the subsystems.  For our purposes the subsystems will correspond to different trials in the DIQKD protocol.  One set of particularly useful symmetric states is the symmetric subspace along a state.
\begin{definition}\label{def:symmetricsubspacealongstate}
The \emph{\index{symmetric subspace along a state}symmetric subspace of $\mathcal{H}^{\otimes n}$ along $\ket{\phi}^{\otimes n-r}$} is the subspace spanned by states 
\begin{equation}
\Pi(\ket{\phi}^{\otimes n-r} \otimes \ket{\phi^{\prime}})
\end{equation}
for any $\ket{\phi^{\prime}}$ on $\mathcal{H}^{\otimes r}$  and operation $\Pi$ which permutes the subsystems.  This subspace is denoted by $Sym(\mathcal{H}, \ket{\phi}^{\otimes n-r})$.
\end{definition} 
This subspace is important because the states in it are very close to symmetric product states, and are hence easy to work with.  The \index{de Finetti Theorem}finite quantum de Finetti Theorem allows us to break symmetric states into a mixture of these near-product states.

\begin{theorem}[Renner \cite{Renner:2007:Symmetry-of-lar} Theorem 4.3.2]\label{theorem:definetti}
Let $\rho \in \mathcal{H}^{\otimes n+k}$ be a pure, permutationally invariant state and let $0 \leq r \leq n$.  There exists a measure $\nu$ on the normalized pure states of $\mathcal{H}$, and for each normalized pure state $\ket{\phi}$ in $\mathcal{H}$ a pure density operator $\rho_{\phi}$ on $Sym(\mathcal{H}, \ket{\phi}^{\otimes n-r})$ such that
\begin{equation}
\left|\left| \tr[k]{\rho} - \int \rho_{\phi}\nu(\phi)  \right|\right|_{1} \leq 2 \exp\left(-\frac{k(r+1)}{2(n+k)} + \frac{1}{2}\dim(\mathcal{H}) \ln k \right)
\end{equation}
\end{theorem}

Here $\tr[k]{\cdot}$ means tracing out any $k$ subsystems.

\subsubsection{Does permuting trials imply a symmetric state?}
In the Renner's QKD security proof a symmetric state is implied by the fact that Alice and Bob randomly permute their measurement outcomes.  Although this operation does not operate directly on the state, it commutes with the measurements since the measurements are identical for each trial (for a given measurement setting) and the measurement settings are chosen uniformly at random.  In the black box scenario the former is not true and the measurements may differ for each trial.

For general strategies we may repair this problem by mapping a strategy to a larger Hilbert space by attaching to each trial a variable indicating its position.  Then the measurement operators may all be replaced with a single measurement that reads this new variable and implements the appropriate measurement strategy.  The measurements are then all identical, but the dimension has been multiplied by the number of trials.  We then permute outcomes, which is equivalent to permuting the state.

Unfortunately, this does not help us.  A symmetric strategy does not necessarily reduce to a mixture of symmetric qubit strategies.  Since we aim to apply the quantum de Finetti Theorem at the qubit level (since the Theorem is sensitive to dimension) we must have a symmetric qubit state.

From this point we offer a possible means of providing a symmetric qubit strategy from an arbitrary qubit strategy.  Unfortunately, this method will result in a decrease in the secure key rate as we detail below.  The main idea is similar to the usual argument:  permute the results and show that this is equivalent to a symmetric state, but the symmetric state we obtain will have a different $S_{max}$ than the original state.

To begin with, we identify the critical security constraint:  measuring the original state $\rho$ and permuting the outcomes must be the same as measuring the symmetric state $\rho^{\prime}$.  We may satisfy this constraint as follows.  When permuting $\rho$ to obtain $\rho^{\prime}$ we must identify the basis in which we are working.  Since $\rho$ is provided within a black box device, we have no natural basis to use.  Instead, we will specify one.  The ideal candidate is for the $B$ side basis of each trial to be the eigenbasis of $B_{0}$.  That is to say, we choose the $B$ side basis for each trial so that $B_{0}= Z$.  Then we partially recover the QKD mechanism:  measuring $\rho$ according $B_{0}$ on each trial and permuting the outcomes \emph{is} the same as measuring $\rho^{\prime}$ according to $B_{0}$.

In order to provide a full solution we must also be able to estimate $S_{max}$ for the trials in $\rho^{\prime}$ (it is the same for each trial), but $\rho^{\prime}$ is merely a convenient fiction.  However, $\rho^{\prime}$ is derived from $\rho$ and we may hope to provide a lower bound using the average $S_{max}$ for $\rho$ (averaged over all trials).  In order to do this we must specify the basis for the $A$ side of each trial.

Recall from section~\ref{sec:chshinequality} that we may express the CHSH operator as a matrix $M$ indexed by the Pauli operators $X,Y,Z$.  We may choose our bases so that $A_{0}$ and $A_{1}$ lie on the $X,Z$ plane of the Bloch sphere.  For the $B$ side we have already fixed $B_{0} = Z$.  By choosing an appropriate phase reference we may also have $B_{1}$ on the $X,Z$ plane.  Then we may consider only the entries of $M$ indexed $X$ and $Z$ since the others are all 0.

Ideally, we would choose the bases so that $M$ is diagonal with positive entries, but since $B_{0}$ is already fixed this will not be possible in general.  We still have some freedom in the choice of the $A$ side basis, though, and we may arrange it so that $M$ is symmetric with positive diagonal.  Recall that $S = \tr{R^{\prime}M}$, where $R^{\prime}$ is determined by the state.  Thus the only parameters of $R^{\prime}$ that matter are the entries $R^{\prime}_{XX}$, $R^{\prime}_{ZZ}$, and $R^{\prime}_{XZ}= R^{\prime}_{ZX}$.

The goal at this point is to show that averaging $S_{max}$ over trials puts a lower bound on $S_{max}$ of the symmetrized state.  There are only three parameters in this analysis, so the hope is that it can be accomplished.  Unfortunately, we have not been able to derive the bound.  Thus the result remains a conjecture.

\subsection{Parameter estimation in symmetric qubit strategies}

At this point we need to develop techniques for estimating the CHSH value of states in $Sym(\mathcal{H}, \ket{\phi}^{\otimes n-r})$.  This is analogous to Theorem 4.5.2 in \cite{Renner:2005:Security-of-Qua}.  However, in that case the measurement operations on each subsystem are all known and identical.  In our case the measurements are not in our control, and we may have no description of them.  Fortunately this is not a very important issue.  The CHSH value that can be achieved by a particular state is a property of the state itself.  If the measurements used are not optimal, then the observed CHSH value can only be lower than if the measurements are optimal.  Since we are only interested in lower bounding the CHSH value, this is sufficient.  Any CHSH value that we observe will (leaving statistical fluctuations aside) be a lower bound on the maximum CHSH value achievable by the state.

\begin{lemma}[Parameter estimation]\label{lemma:diqkdparameterestimation}
Let $\ket{\psi} \in Sym(\mathcal{H}_{2} \otimes \mathcal{H}_{2}, \ket{\phi}^{\otimes n + m-r})$ and let $p = p_{max}(\ket{\phi})$ be the maximum expected value for success on the CHSH test on $\ket{\phi}$, optimized over all measurements.  Let $Y$ be the number of successes after conducting the CHSH test on the first $m$ subsystems of $\ket{\psi}$ according to any measurement strategy.  Then for $\mu > 0$
\begin{equation}
P\left(Y/m > p + \mu \right) \leq  \exp{\frac{-2(m\mu - r(1- p))^{2}}{(n-r)cos^{4}\pi/8 } + (n+m)h(\frac{r}{n+m})\ln 2}.
\end{equation}
\end{lemma}

The proof has two main steps and parallels the proof of Renner's Theorem 4.5.2 \cite{Renner:2005:Security-of-Qua}.  The step first is to bound the given probability for states of the form $\Pi(\ket{\phi}^{\otimes m-r} \otimes \ket{\phi^{\prime}})$ for some permutation $\Pi$.  Next we use Lemma 4.1.6 of Renner which that says $\ket{\psi}$ can be expressed as a superposition of a small number of such states and use Lemma 4.5.1 of Renner which bounds how much the probability can change for such superpositions.

\begin{proof}

We now suppose our system is in the state $\ket{\psi^{\prime}}=\ket{\phi}^{\otimes m-r} \otimes \ket{\phi^{\prime}}$ for some $\ket{\phi^{\prime}}$ on $r$ subsystems.  (We may also permute the subsystems without changing the argument.)  Let $X_{j}$ be the random variable corresponding to the success or failure of the CHSH test on the $j$th subsystem for the measurement strategy actually used (which may vary with $j$).  Since the measurement strategy cannot do better than the optimal strategy, we have $ E(X_{j}) < p$ for $1 \leq j \leq m-r$ and $E(X_{j}) < \cos^{2}\frac{\pi}{8} $ for $j > m-r$.  Applying Hoeffding's inequality (\cite{Hoeffding:1963:Probability-Ine}) to the first $m-r$ subsystems, we obtain for $t > 1$
\begin{equation}
 Pr \left( \sum_{j=1}^{m-r}X_{j} > (m-r)(p + t)\right) \leq e^{\frac{-2(m-r)t^{2}}{\cos^{4}\frac{\pi}{8}}}.
\end{equation}
The remaining $r$ subsystems cannot add very much if $r$ is small.  Thus
\begin{equation}
 Pr \left( \sum_{j=1}^{m}X_{j} > m(p + t) + r\left(1 - p - t\right)\right) \leq e^{\frac{-2(m-r)t^{2}}{\cos^{4}\frac{\pi}{8}}}.
\end{equation}
where $m(p+t) + r (1 - p - t) = (m-r)(p+t) + r$ and the additional $r$ upper bounds the value of $\sum_{j=m-r + 1}^{m} X_{j}$.

We now turn our attention back to $\ket{\psi}$.  Let $z$ be an $m$-tuple with $z_{j} = 1$ if the $j$th trial is successful and $z_{j} = 0$ if it is a failure.  We may write the measurement operator for the CHSH tests together as one large projective measurement $\{M_{z}\}$ with $M_{z}$ the projector corresponding to the outcomes of success and failure given according to $z$.  Then the probability of getting the success/failure outcomes according to $z$ is $\bra{\psi}M_{z}\ket{\psi}$.  Note that $M_{z}$ is positive semi-definite.

We are only interested in the number of successful outcomes, which is given by $w(z)$, the Hamming weight of $z$.  We can restate the above result as 
\begin{equation}
\sum_{w(z) > m(p + t) + r\left(1 - p - t\right)}\bra{\psi^{\prime}}M_{z}\ket{\psi^{\prime}}  \leq e^{\frac{-2(m-r)t^{2}}{\cos^{4}\frac{\pi}{8}}}.
\end{equation}

Now suppose that $\ket{\psi}$ is in $Sym(\mathcal{H}, \ket{\phi}^{\otimes n +m- r})$.  We can express $\ket{\psi}$ as a superposition of states of the form $\ket{\phi}^{n+m-r} \otimes \ket{\phi^{\prime}}$ up to permutations of subsystems.  We can apply the above argument to each of these terms in the superposition.  We are only measuring $m$ of the subsystems, so depending on the permutation anywhere between $m-r$ and $m$ of the subsystems may be in the state $\ket{\phi}$.  Note that our bound still applies since the last $r$ subsystems are arbitrary.  The following two Lemmas from \cite{Renner:2005:Security-of-Qua} bound how much error may be introduced by this procedure.

\begin{lemma}[Renner \cite{Renner:2005:Security-of-Qua} Lemma 4.5.1]\label{lemma:renner_psd_superposition}
Let $\ket{\psi} = \sum_{x \in X} \ket{x}$ and let $P$ be a positive semi-definite operator, then
\begin{equation}
\bra{\psi} P \ket{\psi}  \leq |X| \sum_{x \in X} \bra{x} P \ket{x}.
\end{equation}
\end{lemma}

\begin{lemma}[Renner \cite{Renner:2005:Security-of-Qua} Lemma 4.1.6] \label{lemma:renner_sym_state_basis}
Let $\ket{\psi}$ be a state in $Sym(\mathcal{H}, \ket{\phi}^{\otimes n - r})$.  Then there exist orthogonal vectors $\ket{x}$, which are permutations of $\ket{\phi}^{\otimes n -r} \otimes \ket{\phi_{x}}$ for $x \in X$ such that $\ket{\psi}$ is in the span of the $\ket{x}$ for various $x$, and $|X| \leq 2^{nh(r/n)}$ where $h(\cdot)$ is the binary Shannon entropy. 
\end{lemma}
 
Applying these results we obtain
\begin{equation}
\sum_{w(z) > m(p + t) + r\left(1 - p - t\right)}\bra{\psi}M_{z}\ket{\psi}  \leq e^{\frac{-2(m-r)t^{2}}{\cos^{4}\frac{\pi}{8}}}2^{(n+m)h(\frac{r}{n+m})}.
\end{equation}
Rewriting as a probability, we get
\begin{equation}
 P\left(Y > m(p + t) + r\left(1 - p - t\right) \right) \leq e^{\frac{-2(m-r)t^{2}}{\cos^{4}\frac{\pi}{8}}}2^{(n+m)h(\frac{r}{n+m})}
\end{equation}
or, equivalently
\begin{equation}
 P\left(Y/m > p + \mu \right) \leq \exp\left({\frac{-2(m\mu - r(1- p))^{2}}{(m-r)cos^{4}\pi/8 } + (n+m)h(\frac{r}{n+m})\ln 2}\right).
\end{equation}

\end{proof}

\subsection{Estimating conditional entropies}\label{sec:diqkdestent}

In this section we put some bounds on conditional entropies which will later provide us with the asymptotic key rate.  We begin with a 2 qubit state $\rho_{ABE}$, to which Eve holds the purification, with the property $S_{max}(\rho_{AB}) = S$.  We then measure Bob's system  with an adversarial measurement and estimate $H(X|E)$, the entropy of Bob's outcome conditioned on Eve's system.

There are several tasks.  First, we show that we may take $\rho$ to be a Bell diagonal state on Alice and Bob's qubits.  Second, we estimate $S_{max}(\rho_{AB})$ from the eigenvalues of a Bell diagonal state.  We then use this estimate to bound certain entropies on the eigenvalues of the state.  Finally, we estimate $H(X|E)$, minimized over possible measurements.  Once these tasks are complete we may bound $H(X|E)$ using a function of $S_{max}(\rho_{AB})$.  

The material in this section largely follows Pironio et al.'s argument in \cite{Pironio:2009:Device-independ}.  We indicate where we deviate from their work.

\subsubsection{Bell diagonalization}

We now prove the following Lemma:

\begin{lemma}\label{lemma:diqkdbelldiagsuff}
Let $\rho_{AB}$ be a given two qubit state with measurements $A^{a}$ and $B^{b}$ and consider $H(B|E)$, where $B$ is Bob's system after being measured by $B^{0}$, and $E$ indicates Eve's system, which is a purification of $\rho_{AB}$.  If there exists $f$ such that
\begin{equation}
H(B|E) \leq f(S_{max}(\rho_{AB}))
\end{equation}
for all Bell diagonal $\rho_{AB}$, then the bound also holds for arbitrary $\rho_{AB}$.
\end{lemma}

We begin by supposing that all marginals of $\rho_{AB}$ after being measured by Alice or Bob are uniform.  Pironio et al. originally considered a protocol in which Alice and Bob actively symmetrize the marginals of their measurement outcomes by flipping each outcome randomly and announcing whether or not they did so.  This is fine for qubit strategies, but is slightly problematic for arbitrary strategies since a strategy on a larger Hilbert space may have symmetric marginals but decompose into qubit strategies which do not.  Here we offer a different argument that arrives at the same conclusion:  we may take the marginals to be symmetric without compromising security.

Consider any qubit strategy.  We may fix any basis for our discourse, so we choose one so that $A^{0}$ and $A^{1}$ are in the $X,Z$ plane of the Bloch sphere, and analogously for Bob's measurements.  We may produce a new qubit strategy by applying $Y \otimes Y$ to the state.  This simply flips all outcomes, so it has the same $S_{max}$ and error rate, and gives Eve the same information.  Now consider a strategy which is formed by combining the two states with equal probability, with Eve recording which one is performed.  This strategy again has the same $S_{max}$, error rate, and gives Eve the same information.  Thus we may consider this final strategy alone, and if it is secure then so must be the original strategy.

We now return to Pironio et al.'s argument with a presentation of the proof of Lemma 3 from \cite{Pironio:2009:Device-independ}.  We retain the basis above, with the measurements on the $X,Z$ plane and assume that $Y \otimes Y$ has been applied with probability $1/2$.  In the Bell basis we obtain
\begin{equation}
\rho_{AB} = \left(
	\begin{matrix}
	\lambda_{\phi_{+}} & r_{1} e^{i \theta_{1}} & 0 & 0 \\
	r_{1}e^{-i\theta_{1}} & \lambda_{\psi_{-}} & 0 & 0 \\
	0 & 0 & \lambda_{\phi_{-}} & r_{2} e^{i \theta_{2}} \\
	0 & 0 & r_{2} e^{-i \theta_{2}} & \lambda_{\psi_{+}} \\
	\end{matrix}
\right)
\end{equation}
since $\ket{\phi_{+}}$ and $\ket{\psi_{-}}$ are eigenvectors of $Y \otimes Y$ with eigenvalue 1 while the other two Bell basis vectors have eigenvalue -1.

There is still freedom in the choice of basis, and we may make a rotation about $Y$ on both Alice and Bob's side while keeping the measurements on the $X,Z$ plane.  The rotation angles may be chosen (see \cite{Pironio:2009:Device-independ} Lemma 3 for details) to obtain
\begin{equation}
\rho_{AB} = \left(
	\begin{matrix}
	\lambda_{\phi_{+}} & i r_{1} & 0 & 0 \\
	-ir_{1}& \lambda_{\psi_{-}} & 0 & 0 \\
	0 & 0 & \lambda_{\phi_{-}} & i r_{2}  \\
	0 & 0 & -ir_{2}  & \lambda_{\psi_{+}} \\
	\end{matrix}
\right).
\end{equation}

Finally, we apply an argument analogous to the marginal symmetrization and find that an equal mixture of the above state and its complex conjugate gives the same security.  This mixture has no off-diagonal entries in the Bell basis, and thus is Bell diagonal.  This allows us to consider only Bell diagonal states.

\subsubsection{Estimating $S_{max}$}
We prove the following:

\begin{lemma}\label{lemma:diqkdbelldiagchsh}
Let $\rho_{AB}$ be a given two qubit Bell diagonal state.  Then
\begin{equation}
S_{max}(\rho_{AB}) \geq 2 \sqrt{2} \sqrt{\left(\lambda_{\phi_{+}}  - \lambda_{\psi_{-}}\right)^{2} + \left(\lambda_{\phi_{-}} - \lambda_{\psi_{+}}\right)^{2}}.
\end{equation}
\end{lemma}

\begin{proof}
First, we recall the definition of $R^{\prime}$ from section~\ref{sec:chshinequality}.  For Bell diagonal operators we find
\begin{equation}
R^{\prime} = \left(
	\begin{matrix}
	\lambda_{\phi_{+}} - \lambda_{\phi_{-}} + \lambda_{\psi_{+}} - \lambda_{\psi_{-}} & 0  & 0 \\
	0 & 	-\lambda_{\phi_{+}} + \lambda_{\phi_{-}} + \lambda_{\psi_{+}} - \lambda_{\psi_{-}} & 0 \\
	0 & 0 & \lambda_{\phi_{+}} + \lambda_{\phi_{-}} - \lambda_{\psi_{+}} - \lambda_{\psi_{-}} \\
	\end{matrix}
\right).
\end{equation}
Note that $R^{\prime}$ is diagonal.  Using $R_{XX}$ and $R_{ZZ}$ as the largest in absolute value (if this is not then case, then we still get a lower bound) we find exactly the desired lower bound on $S_{max}$.\footnote{Note that in \cite{Pironio:2009:Device-independ}, Lemma 7 $S_{max}$ is found as a maximum over two values.  The second of these values is obtained for measurements in the $(Z,Y)$ plane.}
\end{proof}
\subsubsection{Entropy inequalities} 

We prove:

\begin{lemma}[Pironio et al. \cite{Pironio:2009:Device-independ} Lemma 6]
Let $\rho_{AB}$ be a two qubit state.  Then
\begin{equation}
h(\overline{\lambda}) - h(\lambda_{\Phi_{+}} + \lambda_{\Phi_{-}}) \leq h \left(\frac{1 + \sqrt{(S_{max}(\sigma_{AB})/2)^{2} - 1}}{2} \right).
\end{equation}
where $h$ is the Shannon entropy.
\end{lemma}

The proof is long and technical, so we refer the reader to \cite{Pironio:2009:Device-independ} for details.  Briefly, the proof introduces a parameterization of the $\lambda$s in terms of $S$ and two other parameters.  Then the entropies on the left hand side of the bound are written in terms of this parameterization.  Finally, optimization techniques are used to find the maximum, giving the required bound.

\subsubsection{Bounding $H(B|E)$ for Bell diagonal states}

\begin{lemma}
Let $\rho_{AB}$ be a two qubit Bell-diagonal state with eigenvalues $\overline{\lambda}$, purified on system $E$.  Suppose the $B$ system is measured with an observable in the $X,Z$ plane to obtain random variable $Y$.  Then
\begin{equation}
 H(Y|E) \geq 1 + h(\lambda_{\Phi_{+}} + \lambda_{\Phi_{-}}) - h(\overline{\lambda}).
\end{equation}.
\end{lemma}

Although this Lemma does not appear in \cite{Pironio:2009:Device-independ}, the proof roughly follows that of Lemma 5 from \cite{Pironio:2009:Device-independ}.  We fill in several details, including the derivation of $\rho_{YE}$ and the calculation of its eigenvalues\footnote{Note that in \cite{Pironio:2009:Device-independ} there is a typo in equation 31.  The final $\sigma$ should appear with $\ket{e_{3}}$ rather than $\ket{e_{4}}$.  Also, the final eigenvalues contain a $\cos 2\phi$ rather than $\cos 4 \theta$ here.  This is due the fact that $\phi$ is an angle in the Bloch sphere and so corresponds to $2 \theta$.}  

We begin with a slight change of notation which will allow more compact formulas below.  We may write the Bell states as
\begin{equation}
\ket{\psi_{st}} = \frac{1}{\sqrt{2}}\sum_{r = 0,1}  (-1)^{rt} \ket{r}\ket{r \oplus s}
\end{equation}
with $\ket{\phi_{+}} = \ket{\phi_{00}}$, $\ket{\phi_{-}} = \ket{\phi_{01}}$, $\ket{\psi_{+}} = \ket{\phi_{10}}$ and $\ket{\psi_{-}} = \ket{\phi_{11}}$.  

We may write a purification of $\rho_{AB}$ as
\begin{equation}
\ket{\Psi_{ABE}} = \frac{1}{\sqrt{2}} \sum_{r,s,t = 0,1} \sqrt{\lambda_{st}} (-1)^{rt} \ket{r}\ket{r\oplus s} \ket{e_{st}}.
\end{equation}
Tracing out the $A$ system may be accomplished by measuring in the $Z$ eigenbasis and ignoring the outcome, obtaining a mixture of the two non-normalized states:
\begin{equation}
\sum_{s,t} \sqrt{\frac{\lambda_{st}}{2}} \ket{s \oplus x} (-1)^{xt}\ket{e_{st}}
\end{equation}
for $x = 0,1$, corresponding to the two reduced states on obtaining outcome $x$.

Next we measure the $B$ system in the basis $q_{0}\ket{0} + q_{1}\ket{1}$ / $q_{1} \ket{0} - q_{0} \ket{1}$ with $q_{0}^{2 } + q_{1}^{2} = 1$.  For outcome $y$ the $E$ system is left in a mixture of the two non-normalized states
\begin{equation}
\sum_{s,t} \sqrt{\frac{\lambda_{st}}{2}} q_{s \oplus x \oplus y}\ket{s \oplus x} (-1)^{xt \oplus ys} \ket{e_{st}}
\end{equation}
for $x = 0,1$.  The state $\rho_{YE}$ is then a direct sum of two blocks.  The blocks may be written as
\begin{equation}
M_{y} = \sum_{s,t,u,v,x} \frac{1}{2}\sqrt{\lambda_{st} \lambda_{uv}} q_{(x \oplus y \oplus s)} q_{(x \oplus y \oplus u)} 
	(-1)^{x(t \oplus v) \oplus y(s \oplus u)} \proj{e_{st}}{e_{uv}}
\end{equation}

Finally we determine the eigenvalues of $\rho_{YE}$, which is the direct sum of the two mixed states above.  For each block (corresponding to a particular outcome $y$) we may obtain the eigenvalues using the following procedure:  determine the trace of the block (which in this case is $1/2$ for each block since either outcome is equally likely for $\rho_{AB}$ Bell diagonal).  Next we determine the trace of the square of the block.  Since each block is a mixture of two states, there will be two eigenvalues, $\Lambda_{w}$ for $w = 0,1$.  Then for block $M$ we have $m = \tr{M_{y}} = \Lambda_{0} + \Lambda_{1}$ and $n = \tr{M_{y}^{2}} = \Lambda_{0}^{2} + \Lambda_{1}^{2}$.  We solve for $\Lambda_{w}$ in terms of $m$ and $n$ by noting that $m^{2} = n + 2\Lambda_{0}\Lambda_{1}$, subbing in $\Lambda_{0} = m - \Lambda_{1}$, and solving the resulting quadratic for $\Lambda_{0}$.  We obtain
\begin{equation}
\Lambda_{w} = \frac{1}{2}\left(\frac{1}{2} + (-1)^{w}  \sqrt{2nÊ- \frac{1}{4}}\right).
\end{equation}
Squaring $M_{y}$ and tracing we obtain
\begin{equation}
\tr{M_{y}} = \sum_{s,t,u,v,x,z} \frac{1}{4}{\lambda_{st} \lambda_{uv}} q_{(x \oplus y \oplus s)} q_{(x \oplus y \oplus u)} q_{z \oplus y \oplus s} q_{z \oplus y \oplus u} (-1)^{(x\oplus z)(t \oplus v)}.
\end{equation}
Fixing $s,t,u,v$ and summing over $x$ and $z$, we find that the $q$s sum to
\begin{equation}
\begin{cases}
 (q_{0}^{2}+q_{1}^{2})^{2} = 1 & t = v, u = s \\
 (q_{0}^{2} - q_{1}^{2})^{2} & t \neq v, u = s \\
 4(q_{0}^{2}q_{1}^{2}) & t = v,  u \neq s \\
 0 & t \neq v, uÊ\neq s.
\end{cases}
\end{equation}
We may write $q_{0} = \cos \theta$, $q_{1} = \sin \theta$, in which case $(q_{0}^{2} - q_{1}^{2})^{2} = \cos^{2} 2\theta = \frac{1 + \cos 4\theta}{2}$, $q_{0}^{2}q_{1}^{2} = \sin^{2} 2\theta = \frac{1 - \cos 4 \theta}{2}$, and  
\begin{equation}
\tr{M_{y}} = \frac{1}{4}\left(\lambda_{00}^{2} + \lambda_{01}^2 + \lambda_{10}^2 + \lambda_{11}^2 + \lambda_{00}\lambda_{01} + \lambda_{00}\lambda_{10} + \lambda_{01}\lambda_{11} + \lambda_{10}\lambda_{11} \right)
\end{equation}
\[
+\frac{\cos 4 \theta}{4}\left(\lambda_{00}\lambda_{01} - \lambda_{00}\lambda_{10} - \lambda_{01}\lambda_{11} + \lambda_{10}\lambda_{11}\right).
\]
Collecting some terms, recalling that $\sum_{st} \lambda_{st} = 1$, and factoring we obtain
\begin{equation}
\frac{1}{2}\left(1 + (\lambda_{00} - \lambda_{11})^{2} + (\lambda_{01} - \lambda_{10})^{2} + 2 \cos 4\theta(\lambda_{00} - \lambda_{11})(\lambda_{01} - \lambda_{10})\right).
\end{equation}
Substituting in, we find
\begin{equation}
\Lambda_{w} = \frac{1}{4} \left(1 + (-1)^{w} \sqrt{(\lambda_{\phi_{+}} - \lambda_{\psi_{-}})^{2} + (\lambda_{\phi_{-}} - \lambda_{\psi_{+}})^{2} + 2 \cos 4 \phi (\lambda_{\phi_{+}} - \lambda_{\psi_{-}}) (\lambda_{\phi_{-}} - \lambda_{\psi_{+}})} \right).
\end{equation}
Each eigenvalue occurs with multiplicity 2, for the two values of $y$.   Recall that for state $\rho_{YE}$, $H(Y|E) = H(\rho_{YE}) - H(\rho_{Y})$.  We find $H(\sigma_{YE})$ to be $1 + h(\Lambda_{+})$, which is minimized for $\theta = 0$ if the term with $\cos$ is positive, or $\theta = \pi/4$ when the term with $\cos$ is positive.  In these cases we obtain $\Lambda_{+} = \lambda_{\phi_{+}} + \lambda_{\phi_{-}}$.  Meanwhile, the state $\rho_{E}$ has the same eigenvalues as $\rho_{AB}$ since $\rho_{E}$ is the purification.  The eigenvalues are thus given by $\overline{\lambda}$, so $H(E) = h(\overline{\lambda})$.  We obtain
\begin{equation}
 H(Y|E) \geq 1 + h(\lambda_{\phi_{+}} + \lambda_{\phi_{-}}) - h(\overline{\lambda}).
\end{equation}

\subsubsection{Bounding $H(Y|E)$ in terms of $S$}

We finally obtain the bound we seek:

\begin{lemma}\label{lemma:diqkdhye}
Let $\rho_{AB}$ be a 2-qubit state, measured on the $B$ system in the $X,Z$ plane to obtain a system $X$.  Further suppose that $\rho_{AB}$ is purified by system $E$.  Then
\begin{equation}
H(Y|E) \geq 1- h \left(\frac{1 + \sqrt{(S/2)^{2} - 1}}{2} \right)
\end{equation}
\end{lemma}
%

\subsection{Security for qubit strategies on symmetric states}
In this section we restrict our attention to the case where the state source emits a pair of qubits and the devices each measure one of these qubits.  Our proof of security is derived from the one given by Renner in \cite{Renner:2005:Security-of-Qua}.  The main difference is in the parameter estimation.  Central to the argument is the finite quantum de Finetti Theorem published in \cite{Renner:2007:Symmetry-of-lar}.
Security for qubit strategies follows from the same proof as Theorem 6.5.1 in \cite{Renner:2005:Security-of-Qua}, with different parameters.  Since the proof is laid out in great detail in \cite{Renner:2005:Security-of-Qua} we will only sketch the proof and indicate the necessary changes.

We begin with a symmetric state $n + m + k$ pairs of qubits, which we purify (according to Lemma 4.2.2 of \cite{Renner:2005:Security-of-Qua}) on Eve's system to a pure symmetric state $\rho$.  Importantly, the purification is symmetric even when considering Eve's systems.  According to the finite quantum de Finetti Theorem (\ref{theorem:definetti}), we may drop $k$ subsystems and obtain
\begin{equation}
\left|\left| \tr[k]{\rho} - \int \rho_{\phi}\nu(\phi)  \right|\right|_{1} \leq \frac{2}{9} \epsilon
\end{equation}
with $\rho_{\phi} \in Sym(H_{2}^{\otimes 4}, \ket{\phi}^{\otimes n+m-r})$ and $r$ depending on $n,m,k,\epsilon$ according to table 6.2 of \cite{Renner:2005:Security-of-Qua}.  We next apply parameter estimation by measuring $m$ systems with measurement settings chosen uniformly at random for Alice and Bob, and determine the number of CHSH successes, $y$.  Then $\frac{y}{m}$ is our estimate of $p$.  If this estimate is below some threshold, $p_{thres} + \mu$ ($p_{thres}$ is used to determine the key rate in the privacy amplification phase) we abort and map the state to 0.  According to Lemma~\ref{lemma:diqkdparameterestimation}, if we choose $\mu$ to be 
\begin{equation}
 \mu = \frac{4r}{m}\sqrt{\left(-\ln \frac{2 \epsilon}{9} - (n+m) h\left(\frac{r}{n+m} \right) \ln 2 \right)(m-r) \cos^{4}\frac{\pi}{8}},
\end{equation}
then the true value of $p$ is lower than $\frac{y}{m} - \mu$, only with probability less than $\frac{2}{9} \epsilon$.  Thus we may apply the parameter estimation to obtain
\begin{equation}
\left|\left| \rho^{PE} - \int_{V} \rho_{\phi}^{PE}\nu(\phi)  \right|\right|_{1} \leq \frac{4}{9} \epsilon
\end{equation}
where we restrict the integral to the set of states $\ket{\phi}$ which have CHSH probability of success $p_{thres}$ or higher (denoted by $V$).  The $PE$ superscripts indicate the application of the parameter estimation protocol.

We now have (if the protocol did not abort) a state $\rho^{PE}$ which is nearly indistinguishable from a mixture of near-product states each with CHSH success probability better than $p_{thres}$.  We may characterize the smooth min entropy of this family of states and apply privacy amplification, deriving a security bound for the finite case.  However, the calculation is essentially the same as it appears in \cite{Renner:2005:Security-of-Qua} and is beyond our scope.  Instead, we will appeal to the final result and calculate the asymptotic key rate.

In \cite{Renner:2005:Security-of-Qua}, Corollary 6.5.2 we find the asymptotic key rate after privacy amplification to be 
\begin{equation}
 \min_{\sigma_{AB}:S_{max}(\sigma_{AB}) \geq S} H(Y|E) - H(Y|X)
\end{equation}
with $H(Y|E)$ and $H(Y|X)$ evaluated for state $\sigma_{AB}$, and $S = 8 p_{thres} - 4$, while $X$ and $Y$ are the classical outcomes for Alice and Bob upon measuring $\sigma_{AB}$.  The system $E$ is Eve's system, which we take to be a purification of $\sigma_{AB}$.  Additionally, we must minimize over measurement strategies of Bob's devices.

By Lemma~\ref{lemma:diqkdhye} the secret key rate is thus bounded below by
\begin{equation}
1- h \left(\frac{1 + \sqrt{(S/2)^{2} - 1}}{2} \right) - h(q)
\end{equation}
where $H(Y|X) = h(q)$ and $q$ is the bit error rate between Alice and Bob's raw keys (Alice error corrects).  This is the same asymptotic rate achieved in \cite{Acin:2007:Device-Independ}.  Note that there is no relationship between $S$ and $q$, since Alice's raw key comes from an unknown measurement.  Her measurement may measure $\rho$ or some other system.  In all cases it is possible for $q$ to range from $0$ to $1$, regardless of the value of $S$.   

\section{Issues with practical implementations and discussion}
Since DIQKD aims to provide higher security in practical implementations than QKD we must carefully examine the extent to which this is true.  In particular we must address issues regarding the selection of measurement settings.  

\subsection{Detector efficiency}

\subsubsection{Detector efficiency loophole}
A common complaint against DIQKD is the \index{detector efficiency loophole}detector efficiency loophole, which was originally studied in the context of Bell inequalities.  In practical optical experiments the detectors in the measurement devices do not always record a photon when it is present.  This means that some trials in a Bell experiment give a ``no outcome'' result.  In a non-adversarial setting this is not problematic, but if we suppose that the measurement devices are adversarial and we discard ``no outcome'' results then the adversary may \emph{selectively} choose a ``no outcome'' result, allowing them to post-select for favourable conditions.  In particular, the devices may post-select for a particular measurement setting, allowing the adversary \emph{de facto} control over the measurement settings.

Consider the following strategy.  Eve determines, for each trial, a measurement setting and outcome for Alice's device and an adaptive strategy for Bob's device.  Alice's device waits for the measurement setting input and if it does not agree with the predetermined setting then the device gives no output, as though the photon was lost.  On the trials for which an output is produced, the adaptive strategy in Bob's device allows for $S=4$ since the measurement and outcome on Alice's side is known.

In the above strategy the detector efficiency of Alice's side is $0.5$ and on Bob's side it is $1$.  By randomly exchange roles, so that Bob's measurement setting is post-selected, we obtain a randomized strategy with efficiency $0.75$ for both Alice and Bob's detectors.  Thus for detector efficiency below $0.75$ there is a local hidden variable model which obtains $S=4$.

Clearly if the observed detector efficiency is too low then DIQKD is not secure; the devices may seem to be lossy, but in fact are implementing a predetermined strategy that allows Eve to know all information.  How efficient must the detectors be to obtain a secure DIQKD implementation?  Currently there has been only minimal study of the robustness of DIQKD against low detector efficiency.  Pironio et al. \cite{Pironio:2009:Device-independ} consider a scenario where ``no outcome'' results are assigned a random outcome.  In this case the inefficiency of the detectors translates directly into effective noise.  Their analysis concludes that a positive key rate is achievable for detector efficiency over $92.4\%$ and they provide a bound for maximum key rate vs. efficiency using this strategy.

A proper analysis of DIQKD accounting for detector efficiency would likely use a modified CHSH inequality.  It is possible to derive Bell type inequalities that account for detector efficiency (see \cite{Eberhard:1993:Background-leve} for one such derivation) which could be used to characterize a set of quantum states which from which secret key may be extracted.  An important roadblock in this program is the fact that Lemma~\ref{lemma:blockdiagonalization} is restricted to observables with two outcomes and no extension is known.

\begin{futurework}
Derive robust security bounds for DIQKD that take into account low detector efficiency.
\end{futurework}

In practical implementations we may make the following assumption:

\begin{assumption}[\index{fair sampling}Fair sampling]
The efficiency on each trial of Alice and Bob's detectors is independent of the measurement setting.
\end{assumption}
\noindent
With this assumption in place there is no post-selection of measurement settings and the problems outlined in this section do not exist.  Of course this assumption is clearly not satisfied in practical settings, as outlined below.  Nevertheless, in security proofs the assumption is often made implicitly.

\subsubsection{Comparison with QKD}
Although the detector efficiency loophole is sometimes given as reason for preferring QKD over DIQKD, this argument does not hold.  In fact, QKD is also susceptible to the detector efficiency loophole.  This point is usually not addressed, since the devices are considered to be under Alice and Bob's control, and hence the efficiency could not vary depending on the measurement setting.  As a concrete counterexample to the validity of this assumption, we consider the class of attacks based on detector efficiency mismatch. 

\index{detector efficiency mismatch}Detector efficiency mismatch attacks were described by Makarov et al. \cite{Makarov:2006:Effects-of-dete} and have been implemented against commercially available QKD systems by Zhao et al. \cite{Zhao:2008:Quantum-hacking}.  Generally speaking, detector efficiency mismatch attacks exploit the fact that different detectors are often used for different measurement settings, and individual detectors differ in their efficiency.  The specific attacks implemented by Zhao et al. are called \index{time shift attacks}time shift attacks.  These attacks are viable against detectors operating in gated mode, where they become sensitive for short periods of time.  The efficiency of the detector thus varies over time in a way that is specific to individual detectors.  The detector for one measurement setting has a relatively high efficiency at a time when the detector for another measurement setting has a low efficiency.  In this case Eve may selectively adjust the arrival time of a photon to match the period where the efficiencies are mismatched, favouring one measurement setting over the other.

\subsection{\index{coincidences}Coincidences and timings}
Since the likelihood of a photon being absorbed within the air or a fibre optic cable is so high, and stray photons and dark counts generate detector clicks when no signal photon is present, it is common to record the timing of detector clicks at both Alice and Bob's labs and look later for coincidences, which are times when a detector clicks for both Alice and Bob.  These coincidences are then used as the raw data for the remainder of the protocol.

Although convenient, finding coincidences may introduce side channels.  For example, if there is also some detector inefficiency mismatch then precise timing data could reveal that a coincidence occurred at a time when the detector is more efficient than another, in which case Eve will learn that one measurement setting or measurement outcome was more likely than another.  Another \index{timing information attack}attack using timing information was described by Lamas-Linares and Kurtsiefer in \cite{Lamas-Linares:07}.  There the specific entangled photon source used has a timing signature that, together with precise coincidence timing data, can leak information to Eve.

The two examples given here show that there may be complex interactions between timing data and other channels which may leak information.  For this reason we make the following assumption (for both DIQKD and typical QKD security models), which is likely to be untrue in many implementations:
\begin{assumption}
There are no correlations between coincidence timing data and the raw key.
\end{assumption}
%

\subsection{The role of DIQKD}
With so many assumptions necessary for security, one may wonder if there is any benefit to DIQKD.  However, these assumptions are also necessary for security in QKD.  Certainly there are fewer assumptions in DIQKD because of the adversarial device model used.  The practical value of DIQKD is limited because of the lower key rates it achieves, but the theoretical value is high because DIQKD places an emphasis on carefully analyzing the assumptions necessary for security and introduces new tools for security analysis.

\chapter{Black box state characterization}
\minitoc

\section{Introduction}
In the previous chapters we have concentrated on using black box devices for specific tasks, which has limited our investigation to particular technical properties that allow us to achieve the goals for each task.  In this chapter we expand the context and consider characterizing a black box state source without reference to its final use.  In particular, we will use robust measures of the quality of the state source with an operational interpretation that will allow us to estimate its behaviour in any context.

\subsection{Reference experiment}
The reference experiment that we will consider is the typical CHSH setup, with a state source emitting a bipartite state, measured by a pair of devices, each with two measurement settings.  The measurements statistics are boiled down into a single number, the CHSH value, as defined in section~\ref{sec:chshinequality}.  Using this value alone we wish to estimate the state emitted by the source, comparing it to an EPR pair, $\ket{\phi_{+}}$.  

The reference experiment consists of an EPR pair, $\ket{\phi_{+}}$ with measurements $A_{0} = X$, $A_{1} = Z$, $B_{0} = \frac{1}{\sqrt{2}}(X + Z)$, and $B_{1} = \frac{1}{\sqrt{2}}(X-Z)$.  The CHSH value for the reference experiment is $2 \sqrt{2}$.

\subsection{Measures of quality}
A major obstacle in this line of research is choosing a suitable measure of quality for the state.  Although our reference state is already chosen, there are many options for how we may compare it to the physical state.  Complicating the picture is the fact that the physical state has an unknown dimension.  Finding a measure that has a reasonable operational interpretation across all possible physical states is challenging.  In this chapter we will develop several measures and compare them.  We will also prove several bounds.  The picture is incomplete, however, as we shall see.  For certain measures proving a bound is problematic, while for less desirable bounds (from an operational perspective) the bounds are reasonably easy to find.

\subsection{Literature review}
There has been limited work in this area to date.  The results in this chapter are all joint work with co-authors Bardyn, Liew, Massar and Scarani and are published in \cite{Bardyn:2009:Device-independ}.  Certain aspects were anticipated by Mayers and Yao \cite{Mayers:2004:Self-testing-qu} in their work on self-testing of EPR pairs discussed in section~\ref{sec:mayersyao}.  Various articles on the CHSH inequality are also relevant, particularly the work by Horodecki et al. \cite{Horodecki:1995:Violating-Bell-}, which may be seen as a non-robust version of the results in this chapter.  Their work shows that any two qubit state that maximally violates the CHSH inequality must be maximally entangled.  Also, the authors developed a method of calculating the maximal CHSH value achievable for any two qubit state, which we used in the proof of Lemma~\ref{lemma:diqkdbelldiagchsh} (although the notation we used was based on \cite{Verstraete:2002:Entanglement-ve}.)

\subsection{Contributions}
In this chapter we develop several measures of quality for an entangled pair source and prove bounds on them.  The definition of these measures presents a significant technical challenge because of the black box nature of the test and the fact that we compare the physical state to a fixed reference state.  These two aspects of the problem mean that the measures must be defined regardless of the dimension of the state and nevertheless have a consistent interpretation that does not depend on the dimension.

We are able to prove bounds for several of the measures of quality that we define.  In particular, we have a complete characterization for qubits:  a lower bound and a set of states that achieve the bound, proving that it is tight.  When we extend the definition to arbitrary dimensions we have some partial results.  For the most restrictive definition ($F_{MY}$, defined below) we have a conjectured lower bound for pure states, and examples that saturate this bound.  For the least restrictive definition ($F_{LOCC}$) we have a lower bound, but no examples that saturate it.  This bound is conjectured to be tight.  Meanwhile, for the final measure ($F_{LO}$), which is bounded above and below by the previous two measures, we have a lower bound but conjecture that it is not tight.

The material in this chapter was published in \cite{Bardyn:2009:Device-independ} and is joint work with the co-authors.

\section{Two qubits}
Although our ultimate goal will be to have no assumptions on the dimension of the physical state, we will first investigate the case of two qubits, both because it is easier to analyze and because the bound we derive will be useful when we move to the case of higher dimensions.  

\subsection{Fidelity and trace distance}
In the case where the physical state is known to be two qubits it is quite easy to measure the quality of an entangled pair.  We simply choose an appropriate distance function on states.  The most meaningful distance, from an operational point of view, is the trace distance given by\footnote{The trace distance is sometimes given without the factor of $\frac{1}{2}$, particularly in Watrous' lecture notes.}
\begin{equation}
||\rho - \sigma||_{Tr} =\frac{1}{2} \tr{|\rho - \sigma|} =  \frac{1}{2}\tr{\sqrt{ (\rho - \sigma)^{2}}}
\end{equation}
which measures how distinguishable two states are by any procedure.  Watrous' lecture notes \cite{Watrous:2008:Operator-decomp} provides a good introduction.  

An easier figure to calculate is the fidelity, given by\footnote{The fidelity is also sometimes given as the square root of this value, particularly in Watrous' lecture notes.}
\begin{equation}
F(\rho, \sigma) = \tr{\sqrt{\rho}\sigma\sqrt{\rho}}.
\end{equation}
Watrous' lecture notes \cite{Watrous:2008:Overview-of-qua} again give a good introduction.  The fidelity is related to the trace distance by
\begin{equation}
1-\sqrt{F(\rho, \sigma)} \leq ||\rho - \sigma||_{Tr} \leq \sqrt{1 - F(\rho, \sigma)}
\end{equation}
(the Fuchs-van de Graaf Inequalities) which is saturated on the right when $\rho$ and $\sigma$ are both pure.  Since the fidelity is easier to calculate, we will use it for the remainder of this chapter.

It is often convenient to calculate the fidelity by means of a purification using the following Lemma (with a proof appearing in Watrous' lecture notes \cite{Watrous:2008:Overview-of-qua}.)

\begin{lemma}[Uhlmann's Theorem]
Let $\rho$ and $\sigma$ be given with $\ket{\psi}$ any purification of $\rho$.  Then
\begin{equation}
F(\rho, \sigma) = \max_{\ket{\phi} \text{ purification of } \sigma}|\braket{\psi}{\phi}|^{2}.
\end{equation}
\end{lemma}

As well, if $\rho = \proj{\psi}{\psi}$ is pure (as frequently will be the case in our discussion) then
\begin{equation}
F(\proj{\psi}{\psi}, \sigma) = \bra{\psi} \sigma \ket{\psi}.
\end{equation}
%
\subsection{Measuring the state}
Suppose we have a bipartite state $\rho$ on a pair of qubits.  Our reference state will always be a pair of maximally entangled qubits:  $\ket{\phi_{+}}$.  Ideally we would like to bound the trace distance between these two states:
\begin{equation}
||\rho - \proj{\phi_{+}}{\phi_{+}}||_{Tr}
\end{equation}
but we will settle for estimating the fidelity, $F(\rho, \proj{\phi_{+}}{\phi_{+}})$, which will also provide a bound on the trace distance.  There is some ambiguity, though, since we have no natural local bases when considering $\rho$, and our description of the reference state assumes that the local bases are known.  To this end we define $F_{MY}$ as
\begin{equation}
F_{MY}(\rho) = \max_{U, V \text{ unitary}} F(U \otimes V\rho U^{\dagger} \otimes V^{\dagger}, \proj{\phi_{+}}{\phi_{+}}).
\end{equation}
The maximization takes into account our lack of preference for particular local bases for $\rho$.  The value of $F_{MY}(\rho)$ will vary between $1/4$ for a maximally mixed state and 1 for a maximally entangled state.  

Another measure that we will use allows for arbitrary local operations, rather than just changes of bases.  For pure states this will not offer any advantage, but for mixed states it allows a slightly higher fidelity.  We define this by

\begin{definition}
Let a two qubit state $\rho$ be given.  Then
\begin{equation}
F_{LO}(\rho) = \max_{\Phi_{A}, \Phi_{B}} F(\Phi_{A}Ê\otimes \Phi_{B}(\rho) , \proj{\phi_{+}}{\phi_{+}})
\end{equation}
where $\Phi_{A}$ and $\Phi_{B}$ range over completely positive trace preserving maps taking qubits to qubits.
\end{definition}

By the simple fact that unitary operations are local, we know that $F_{MY} \leq F_{LO}$.  For pure states the two quantities are equal;  if the state is a product state then unitaries will take it to $\ket{00}$ giving the highest possible fidelity for a product state (mixed or not).  If the state is entangled than no operation can increase the entanglement; only a change of basis is required.  

In the case of mixed states, however, we may find strict inequality.  Any separable state may be replaced by a pure product state.  In the case of the completely mixed state, for example, we find $F_{MY}(I/2) = 1/4$ while $F_{LO}(I/2) = 1/2 = F_{LO}(\ket{00})$.

\subsection{Model for measurement operators}
We will model the measurement operators for the CHSH test as two-outcome observables.  Thus, each measurement operator will be Hermition with eigenvalues $\pm 1$.  For qubits an important consideration is whether we allow all eigenvalues to be $1$ or $-1$, that is, whether we allow $I$ or $-I$ as a measurement operator.  More generally speaking, how do we model a fixed outcome?  We may model a fixed outcome by a $I$ or $-I$ as the measurement, or by using eigenvectors of some non-trivial measurement as the state.  

Generally speaking this will not pose a problem.  We will be only interested in cases where $S > 2$ since if $S \leq 2$ then the experiment may be simulated with a local hidden variable model and we have no interest.  For $S > 2$, a non-trivial measurement must be used.

\subsection{Bound for qubits}
We will prove the following Theorem:

\begin{theorem}\label{theorem:chshqubitmax}
Let a two qubit state $\rho$ be given.  Then
\begin{equation}
F_{MY}(\rho) \geq \frac{ 1 + \sqrt{\left(\frac{S_{max}(\rho)}{2}\right)^{2} - 1}}{2} .
\end{equation}
Furthermore, the bound is saturated for $S \geq 2$ by the states $\cos\theta\ket{00} + \sin\theta\ket{11}$.
\end{theorem}
There are many ways of approaching this proof.  For a different proof, see Bardyn et al. \cite{Bardyn:2009:Device-independ}.

\begin{proof}
\noindent{\bf Lower bound}

We begin much as in section~\ref{sec:diqkdestent} by reducing to the case of Bell diagonal states.  Suppose that $S_{max}(\rho)$ is achieved for measurements $A_{a}$, and $B_{b}$.  We may choose any basis we like for our discourse, so we suppose that the measurements are all in the $X,Z$ plane.  Note that if we apply $Y \otimes Y$ to $\rho$ then all the outcomes on both sides are flipped, but $S$ is not affected since the only thing that matters is whether the outcomes are the same or different.  As in the proof for Lemma~\ref{lemma:diqkdbelldiagsuff} we may thus construct a new state by applying $YÊ\otimes Y$ with probability $1/2$, and the local bases may be chosen so that the off-diagonal entries of the density matrix in the Bell basis are all imaginary.  As well, complex conjugation does not affect $S$ since the measurements are all in the real plane.  We may thus take an equal mixture of the state and its complex conjugate to obtain a state 
\begin{equation}
\sigma = \frac{1}{4}\left(\rho + \rho^{*} + Y \otimes Y (\rho + \rho^{*}) Y \otimes Y\right)
\end{equation}
with $S_{max}(\sigma) \geq S_{max}(\rho)$.

Now consider $F_{MY}(\sigma)$
\begin{equation}
F_{MY}(\sigma) = \bra{\phi_{+}} \sigma \ket{\phi_{+}}
\end{equation}
\begin{equation}
= \bra{\phi_{+}}\rho\ket{\phi_{+}} + \bra{\phi_{+}} \rho^{*} \ket{\phi_{+}}
+
\bra{\phi_{+}}Y \otimes Y\rho Y \otimes Y\ket{\phi_{+}} + \bra{\phi_{+}}Y \otimes Y \rho^{*}Y \otimes Y \ket{\phi_{+}}.
\end{equation}
Since $\ket{\phi_{+}}$ has all real entries, complex conjugation does not matter.  Also, we may rewrite
\begin{equation}
\bra{\phi_{+}}Y \otimes Y\rho Y \otimes Y\ket{\phi_{+}} = \tr{Y \otimes Y \proj{\phi_{+}}{\phi_{+}} Y \otimes Y \rho}.
\end{equation}
The reader may verify that $Y \otimes Y \ket{\phi_{+}}= - \ket{\phi_{+}}$.  Finally, since complex conjugation and multiplying by $Y \otimes Y$ do not change the state $\ket{\phi_{+}}$, we conclude that the optimal basis for $\rho$ is also the optimal basis for the other states in the mixture as well.  Combining these facts, we obtain
\begin{equation}
F_{MY}(\sigma) = F_{MY}(\rho).
\end{equation}

Let $f$ be defined by
\begin{equation}
f(S) = \frac{ 1 + \sqrt{\left(\frac{S}{2}\right)^{2} - 1}}{2}  .
\end{equation}
Note that $f$ is increasing.  If we find that $F_{MY}(\sigma) \geq f(S)$ then we obtain
\begin{equation}
F_{MY}(\rho) = F_{MY}(\sigma) \geq f(S_{max}(\sigma)) \geq f(S_{max}(\rho)).
\end{equation}
Thus we may restrict ourselves to Bell diagonal states. 

We now assume that $\rho$ is Bell diagonal with eigenvalues $\lambda_{\phi_{+}}, \lambda_{2}, \lambda_{3}, \lambda_{4}$ with largest eigenvalue $\lambda_{\phi_{+}}$ (if this is not the case, then a local change of basis will make it so and keep the state Bell diagonal).  The ordering of the remaining eigenvalues is not important.  Following the proof of Lemma~\ref{lemma:diqkdbelldiagchsh} we find
\begin{equation}
S_{max}(\rho) = 2 \sqrt{2} \sqrt{\left(\lambda_{\phi_{+}}  - \lambda_{2}\right)^{2} + \left(\lambda_{3} - \lambda_{4}\right)^{2}}
\end{equation}
for some ordering of the remaining eigenvalues.  For a particular value of $\lambda_{\phi_{+}}$ the largest possible value of $S_{max}$ occurs when $\lambda_{2} = \lambda_{4} = 0$ and $\lambda_{3} = 1 - \lambda_{\phi_{+}}$, hence
\begin{equation}
\left(\frac{S_{max}}{2 \sqrt{2}}\right)^{2} \leq \lambda_{\phi_{+}}^{2} + (1 - \lambda_{\phi_{+}})^{2} = 2 \lambda_{\phi_{+}}^{2} + 1 - 2 \lambda_{\phi_{+}}
\end{equation}
\begin{equation}
\frac{1}{2}\left[\left(\frac{S_{max}}{2}\right)^{2} - 1 \right] \leq 2(\lambda_{\phi_{+}}^{2} - \lambda_{\phi_{+}} + \frac{1}{4}) = 2 \left(\lambda_{\phi_{+}} - \frac{1}{2}\right)^{2}
\end{equation}
\begin{equation}
\frac{1 + \sqrt{\left(\frac{S_{max}}{2}\right)^{2} - 1 }}{2} \leq \lambda_{\phi_{+}} = F_{MY}(\rho).
\end{equation}

\noindent{\bf Tightness of the bound}

We claim that the states $\ket{\psi} = \cos \theta \ket{00} + \sin \theta \ket{11}$ saturate the bound.  We first refer the reader to the proof of Lemma~\ref{lemma:diqkdbelldiagchsh}.  The matrix $R^{\prime}$ for this state is
\begin{equation}
\left(
	\begin{matrix}
	2\cos \theta \sin \theta & 0 & 0 \\
	0 & - 2\cos \theta \sin \theta & 0 \\
	0 & 0 & 1 \
	\end{matrix}
\right).
\end{equation}
Then $S_{max} = 2 \sqrt{1 + 4\cos^{2}\theta\sin^{2} \theta}$.
Meanwhile, $F_{MY}(\ket{\psi})$ is just $|\braket{\psi}{\phi_{+}}|^{2} = \frac{1}{2}\left(\cos \theta + \sin \theta\right)^{2}$ since it is obvious that no unitary operation can improve this.  (This fact is proven in Lemma~\ref{lemma:fmypurestates} below.)  We then find

\begin{equation}
\frac{ 1 + \sqrt{\left(\frac{S_{max}(\rho)}{2}\right)^{2} - 1}}{2}  = \frac{1 + 2\cos \theta \sin \theta}{2} = F_{MY}(\ket{\psi})
\end{equation}
and the state saturates the bound.

\end{proof}

The above bound immediately extends to $F_{LO}$ if we replace $\rho$ with $\ket{00}$ whenever  $\rho$ has $S_{max}(\rho) < 2$.  Hence we obtain the following Lemma.

\begin{lemma}
Let a bipartite state $\rho$ be given.  Then
\begin{equation}
F_{LO}(\rho) \geq  \frac{ 1 + \sqrt{\left(\frac{\max(S_{max}(\rho),2)}{2}\right)^{2} - 1}}{2}.
\end{equation}
Further, this bound is tight.
\end{lemma}
Note that the bounds for $F_{MY}$ and $F_{LO}$ coincide for $S \geq 2$.  As mentioned previously, arbitrary operations do not improve on $F_{MY}$ for pure states, so the pure state also saturate this bound for $S \geq 2$.  For $S \leq 2$, any state can be transformed to a pure product state using $LO$ and a deterministic strategy gives $S = 2$, so the bound is saturated for these states as well. 

Later we will use the qubit bound when considering bounds for higher dimensional systems.  In that application we will need to apply convexity arguments and the following Lemma will be useful.

\begin{lemma}\label{lemma:qubitboundconcavedown}
Define $f$ by
\begin{equation}
f(S) = \frac{1 + \sqrt{\left(\frac{S}{2}\right)^{2} - 1 }}{2}.
\end{equation}
Then $f(x)$ is concave down on the domain $2 \leq x \leq 2\sqrt{2}$.
\end{lemma}
\begin{proof}
We instead consider $g(x) = \sqrt{x^{2}-1}$ for $1 \leq x \leq \sqrt{2}$.  The second derivative of $g(x)$ is
\begin{equation}
-\frac{\sqrt{x^{2} - 1}}{\left(x^{2}-1\right)^{2}}.
\end{equation}
For $1 \leq x \leq \sqrt{2}$ this is always negative.  Hence $g(x)$ and $f(x)$ are concave down.
\end{proof}
%
\section{Measures of quality for higher dimensions}
For higher dimensional physical systems we run into an important question of definition:  how do we compare a physical state on a state of arbitrary dimension with one of a fixed dimension?  The only reasonable reference state to compare the physical state with is a maximally entangled pair of qubits since these are the only states that maximally violate the CHSH inequality, and without further information about the physical state we cannot rule out the two qubit case either.  We present several approaches to this problem.  The most useful measure will likely depend on the intended application.

\subsection{Mayers-Yao type fidelity}
We begin using the approach begun by Mayers and Yao in \cite{Mayers:2004:Self-testing-qu} and continued by Magniez et al. in \cite{Magniez:2006:Self-testing-of}.  Let $\ket{\psi}_{AB}$ be a bipartite pure physical state.  Mayers and Yao evaluate the physical state by asking whether or not there exists a state of the form $\ket{\psi^{\prime}}_{AB} \otimes \ket{\phi_{+}}_{AB}$ that is equivalent to $\ket{\psi}$ under local unitary transformations.  Later, Magniez et al. make the notion robust by considering the quantity
\begin{equation}
\min_{U,V, \ket{\psi^{\prime}}} || U_{A} \otimes V_{B} \ket{\psi} - \ket{\psi^{\prime}}\ket{\phi_{+}} ||_{1}
\end{equation}
where $U$ and $V$ are unitary transformations.  We may easily transform this to match the fidelity measure we use for qubits as follows (abusing notation a little).
\begin{definition}
Let a pure bipartite physical state $\ket{\psi}$ be given.  Then $F_{MY}(\ket{\psi})$ is defined by
\begin{equation}
F_{MY}(\ket{\psi}) = \max_{U,V, \ket{\psi^{\prime}}} F(U_{A} \otimes V_{B} \ket{\psi}, \ket{\psi^{\prime}}\ket{\phi_{+}}).
\end{equation}
where $U$ and $V$ are unitary transformations and $\ket{\psi^{\prime}}$ is any state on the appropriate space.
\end{definition}
We may extend this definition to mixed states by replacing $\ket{\psi^{\prime}}$ with a mixed state $\rho^{\prime}$.  We need to keep the transformations unitary however, otherwise we could always take $\ket{\psi^{\prime}}$ to be a maximally mixed state and have the transformation completely mix the corresponding portion of the physical state. 

With this definition we are able to provide a complete characterization for pure states in the following Lemma.  This relies on the singular value (Schmidt) decomposition for pure bipartite states.
\begin{lemma}\label{lemma:fmypurestates}
Let bipartite $\ket{\psi}$ be given with
\begin{equation}
\ket{\psi} = \sum_{j} \lambda_{j} \ket{a_{j}}\ket{b_{j}}.
\end{equation}
Then
\begin{equation}
F_{MY}(\ket{\psi})  = \sum_{l} \frac{\left(\lambda_{2l} + \lambda_{2l + 1} \right)^{2}}{2}.
\end{equation}

\end{lemma}
\begin{proof}
The closest state of the form $\ket{?} \otimes \ket{\phi_{+}}$ has Schmidt decomposition
\begin{equation}
\ket{\phi} = \sum_{j} \mu_{j} \ket{c_{j}}\ket{d_{j}}
\end{equation}
with $\mu_{2l} = \mu_{2l+1}$.  For concreteness, we may assume that the $\lambda_{j}$s and $\mu_{j}$s are both in decreasing order

We first show that we may take $\ket{c_{j}} = \ket{a_{j}}$ and $\ket{d_{j}} = \ket{b_{j}}$.  Note that
\begin{equation}
|\braket{\psi}{\phi}| \leq \sum_{jk} \lambda_{j} \mu_{k} |\braket{a_{j}}{c_{k}}\braket{b_{j}}{d_{k}}|.
\end{equation}
Let us define the matrix $M$ by 
\begin{equation}
M_{jk} = |\braket{a_{j}}{c_{k}}\braket{b_{j}}{d_{k}}|.
\end{equation}
The values $|\braket{a_{j}}{b_{k}}|$ for various $k$ and fixed $j$ form a vector of norm 1 since $\ket{b_{k}}$ is a basis and $\ket{a_{j}}$ has norm 1.  The same is true for the values $|\braket{b_{j}}{d_{k}}|$ and if we fix $k$ and vary $j$ instead.  Thus columns (and rows) of $M$ are formed by entrywise products of norm 1 vectors and the sum of each row and column of $M$ is at most 1.  This means that we can find a new matrix $N$ with positive entries such that $M + N$ is doubly stochastic.  Note that
\begin{equation}
|\braket{\psi}{\phi}| \leq \sum_{jk} \lambda_{j} \mu_{k} (M+N)_{jk}.
\end{equation}

By the Birkhoff-von Neumann Theorem we may write $M+N$ as a convex combination of permutation matrices, thus
\begin{equation}
M+N = \sum_{m} p_{m} P_{m}
\end{equation}
with $\sum_{m} p_{m} = 1$ and $P_{m}$ permutation matrices.  Since the combination is convex, there exists some $m$ for which
\begin{equation}
|\braket{\psi}{\phi}| \leq \sum_{jk}\lambda_{j} \mu_{k} (P_{m})_{jk}.
\end{equation}
The permutations merely reorder the $\mu_{j}$s and it is easy to prove that the maximum is achieved when the $\lambda_{j}$s and $\mu_{j}$s are both in decreasing order.  Hence $P_{m} = I$ satisfies the above equation.  We may achieve this by choosing the bases $\ket{c_{j}} = \ket{a_{j}}$ and $\ket{d_{j}} = \ket{b_{j}}$, so we need not consider any other bases. 

We now optimize over $\mu_{j}$ subject to the condition $\mu_{2l} = \mu_{2l+1}$.  By the Cauchy-Schwarz inequality we have

\begin{equation*}
|\braket{\psi}{\phi}|^{2} = \left(\sum_{l} (\lambda_{2l} + \lambda_{2l + 1})\mu_{2l}\right)^{2}
\end{equation*}
\begin{equation} \leq 
\left(\sum_{l} (\lambda_{2l} + \lambda_{2l + 1})^{2}\right) \left(\sum_{l} \mu_{2l}^{2} \right)
\end{equation}

with equality when $\mu$ and $\lambda$ are collinear.  Thus we set
\begin{equation}
\mu_{2l} = \mu_{2l+1} = \frac{\lambda_{2l} + \lambda_{2l+1}}{N}
\end{equation}
with $N$ a normalization constant equal to
\begin{equation}
N = \sqrt{2 \sum_{l} (\lambda_{2l} + \lambda_{2l +1})^{2}}.
\end{equation}
With these values, we obtain
\begin{equation}
F_{MY}(\ket{\psi}) = |\braket{\psi}{\phi}|^{2} = \sum_{l} \frac{\left(\lambda_{2l} + \lambda_{2l + 1} \right)^{2}}{2}.
\end{equation}
\end{proof}
%

\subsection{Local operations and LOCC}\label{sec:flofloccdef}
The $F_{LO}$ we defined for qubits could be extended in a number of ways.  In particular, $F_{MY}$ defined in the previous section reduces to $F_{MY}$ in the case where the physical state is a pair of qubits, which is equal to $F_{LO}$ for states with $S > 2$.  Here we will extend it in another way.  We begin with a physical bipartite state $\rho$ and ask the question ``What is the most entangled pair of qubits we can obtain through local operations?''  This gives rise to the following definition:

\begin{definition}
\begin{equation}
F_{LO}(\rho) = \max_{\Phi_{A}, \Phi_{B}} F(\Phi_{A} \otimes \Phi_{B}(\rho), \proj{\phi_{+}}{\phi_{+}})
\end{equation}
with $\Phi_{A}$ and $\Phi_{B}$ ranging over completely positive trace preserving maps from the spaces of the physical state a qubits.
\end{definition}
For qubits this definition is not the same as $F_{MY}$ since arbitrary qubit channels are allowed instead of just unitaries.  In the Lemma below we establish a relationship between $F_{MY}$ and $F_{LO}$.

\begin{lemma}
Let bipartite $\rho$ be given.  Then
\begin{equation}
F_{MY}(\rho) \leq F_{LO}(\rho)
\end{equation}
with equality if $\rho$ is pure.
\end{lemma}
\begin{proof}
Let $\rho$ be given.  Then
\begin{equation}
F_{LO}(\rho) = \max_{\Phi \in LO} F(\Phi(\rho), \proj{\phi_{+}}{\phi_{+}})
\end{equation}
with $LO$ the set of local operations that take the space $AB$ to a pair of qubits.  We may restrict this set to operations which only apply local unitaries and trace out everything but a pair of qubits to obtain
\begin{equation}
F_{LO}(\rho) \geq \max_{U, V} F(\text{tr}_{X} (U \otimes V \rho U^{\dagger} \otimes V^{\dagger}), \proj{\phi_{+}}{\phi_{+}})
\end{equation}
where $\text{tr}_{X}$ means tracing out everything but a pair of qubits.  Since the fidelity cannot decrease when a system is traced out we have
\begin{equation}
F_{LO}(\rho) \geq \max_{U,V} F(U \otimes V \rho U^{\dagger}Ê\otimes V^{\dagger}, \rho^{\prime} \otimes \proj{\phi_{+}}{\phi_{+}})
\end{equation}
for all $\rho^{\prime}$, and in particular for the $\rho^{\prime}$ which maximizes the expression and gives $F_{MY}(\rho)$.  Thus
\begin{equation}
F_{LO}(\rho) \geq F_{MY}(\rho)
\end{equation}

Now suppose that $\rho = \proj{\psi}{\psi}_{AB}$.  We may write an operation in LO as adding a pair of ancillas and a pair of target qubits, applying a pair of unitaries, and tracing out everything but the target qubits.  Thus
\begin{equation}
F_{LO}(\ket{\psi}) = \max_{U, V} F(\text{tr}_{ABX_{a}X_{b}}(U \otimes V\ket{\psi}_{AB} \ket{00}_{X_{a} X_{b}} \ket{00}_{Y_{a}Y_{b}}), \proj{\phi_{+}}{\phi_{+}}_{Y_{a}Y_{b}}) 
\end{equation}
Applying Uhlmann's Theorem, we obtain
\begin{equation}
F_{LO}(\ket{\psi}) = \max_{U, V, \ket{\phi}} \left|\bra{\psi}_{AB} \bra{00}_{X_{a} X_{b}} \bra{00}_{Y_{a}Y_{b}}U^{\dagger} \otimes V^{\dagger} \ket{\phi}\otimes\ket{\phi_{+}}_{Y_{a}Y_{b}}\right|^{2}  
\end{equation}
The right hand side is equal to $F_{MY}(\ket{\psi}\otimes \ket{00} \otimes \ket{00})$ by definition.  This in turn is equal to $F_{MY} (\ket{\psi})$ since the value of $F_{MY}$ for a pure state is only dependent on the Schmidt decomposition, which the product state ancillas do not change.  Thus
\begin{equation}
F_{LO}(\ket{\psi}) = F_{MY}(\ket{\psi}).
\end{equation}
\end{proof}

For mixed states there exist cases with a strict inequality.  For example $F_{MY}(\frac{I}{4}) = \frac{1}{4}$, but $F_{LO}(\frac{I}{4}) = \frac{1}{2}$ since the class $LO$ allows us to replace the state with $\ket{00}$.

The proof of the Lemma implies a construction for the optimal local operations for extract an approximate EPR pair.  First take the singular value decomposition with singular values ordered in a decreasing fashion.  Then pair them up and introduce some new variables to obtain
\begin{equation}
\ket{\psi} = \sum_{j} \sqrt{p_{j}} \left(c_{j} \ket{2j}\ket{2j} + s_{j} \ket{2j +1}\ket{2j+1}\right)
\end{equation}
where $c_{j}^{2} + s_{j}^{2} = 1$ and $\lambda_{2j} = p_{j}c_{j}\, , \,\, \lambda_{2j+1} = p_{j}s_{j}$, further implying that $\sum_{j} p_{j} = 1$.  The decreasing ordering implies that $c_{j}$ and $s_{j}$ are as close together as possible, overall.  We can then think of the state as a direct sum of pairs of qubits, with each pair of qubits close as possible to an EPR pair.  The optimal local operations consist of projecting onto the spaces spanned by $\ket{2j}$ and $\ket{2j+1}$ for various $j$, and then mapping $\ket{2j}$ to $\ket{0}$ and $\ket{2j+1}$ to $\ket{1}$ to obtain a qubit.

\subsubsection{Larger classes of operations}
We may further extend the definition of $F_{MY}$ on qubits using different classes of operations.  Instead of considering local operations of the form $\Phi_{A}\otimes \Phi_{B}$ we may instead consider separable maps or local operations with classical communication (LOCC).  The latter is most interesting to us because they are the largest class of reasonably implementable operations that do not increase entanglement.  (Separable operations are the largest class that do not increase entanglement, but they may require quantum communication to implement.)  Thus we define one more measure of fidelity:

\begin{definition}
\begin{equation}
F_{LOCC}(\rho) = \max_{\Phi} F(\Phi(\rho), \proj{\phi_{+}}{\phi_{+}})
\end{equation}
with $\Phi$ ranging over LOCC maps that take the state $\rho$ to a pair of qubits.
\end{definition}

Since $LOCC$ contains $LO$ we have
\begin{equation}
F_{MY}(\rho) \leq F_{LO} \leq F_{LOCC}.
\end{equation}
%

\section{Bounds for higher dimensions}

\subsection{Bounds for $F_{MY}$}
Recall that for pure states we are able to analytically calculate $F_{MY}$ in terms of the singular values (Schmidt coefficients).  Conveniently, Gisin and Peres  \cite{Gisin:1992:Maximal-violati} have studied the case of determing $S_{max}$ of a pure state in terms of the singular values.  Although they were not able to find an analytic solution, they make the following conjecture

\begin{conjecture}[Gisin, Peres \cite{Gisin:1992:Maximal-violati}]\label{conj:gisinperes}
Let a bipartite state $\ket{\psi}$ be given with singular value decomposition
\begin{equation}
\ket{\psi} = \sum_{j} p_{j} \left(c_{j}\ket{2j}\ket{2j} + s_{j} \ket{2j+1}\ket{2j+1}\right).
\end{equation}
with $p_{0}c_{0} \geq p_{0}s_{0} \geq p_{1}c_{1} \geq p_{1}s_{1} \geq \dots$.  Then
\begin{equation}
S_{max}(\ket{\psi}) = \sum_{j} p_{j} \left(2 \sqrt{1 + 4 c_{j}^{2}s_{j}^{2}}\right).
\end{equation}
\end{conjecture}

Later numerical studies by Liang and Doherty \cite{Liang:2006:Better-Bell-ine} supported the conjecture.

Gisin and Peres give an explicit construction for the measurement operators which achieve their conjectured value of $S_{max}$.  The general idea is much the same as for the optimal local operations for pure states that achieve $F_{LO}$, as discussed in section~\ref{sec:flofloccdef}.  The state is divided up into a direct sum of pairs of qubits, and the measurements are a projection onto one of the summands followed by the optimal measurement for that pair of qubits.  

We now prove that the Gisin-Peres conjecture implies a bound on $F_{MY}$ in terms of $S_{max}$ for pure states.
\begin{lemma}\label{lemma:mybound}
If conjecture~\ref{conj:gisinperes} holds, then for pure bipartite state $\ket{\psi}_{AB}$ we have
\begin{equation}
F_{MY}(\ket{\psi}) \geq  \frac{S_{max}(\ket{\psi}) + 2 \sqrt{2} - 4}{4\left(\sqrt{2} - 1\right)}.
\end{equation}
Furthermore, this bound is tight.
\end{lemma}

\begin{proof}
We may write $\ket{\psi}$ in the singular value decomposition with decreasing singular values $\sqrt{p_{j}} c_{j}$, $\sqrt{p_{j}} s_{j}$.  Let $\ket{\psi_{j}} = c_{j} \ket{00}_{AB} + s_{j}\ket{11}_{AB}$.  Mapping to a different space (possibly adding one more dimension on each side if the original dimension was not even on both sides) $\ket{\psi} = \sum_{j} \sqrt{p_{j}} \ket{\psi_{j}}_{AB} \ket{jj}_{AB}$.  Then the Gisin-Peres conjecture implies
\begin{equation}
S_{max}(\ket{\psi}) = \sum_{j} p_{j} S_{max}\left(\ket{\psi_{j}}_{AB}\right).
\end{equation}
This may be seen as a result of the fact that the Gisin-Peres construction for the optimal measurement strategy is to project onto $\ket{jj}_{AB}$, obtaining the value $j$ on both sides (with probability $p_{j}$), and implement the optimal measurement strategy for $\ket{\psi_{j}}$.  Also, the fact that $S_{max}(c\ket{00} + s\ket{11}) = 2 \sqrt{1 + 4c^{2}s^{2}}$ completes the argument.  This equation may be obtained using the techniques from the proof of Theorem~\ref{theorem:chshqubitmax}.  We set $S_{j} = S_{max}(\ket{\psi_{j}})$. 

Meanwhile, applying Lemma~\ref{lemma:fmypurestates}, we find that $F_{MY}$ may also be written
\begin{equation}
F_{MY}(\ket{\psi}) = \sum_{j} p_{j} F_{MY}(\ket{\psi_{j}}),
\end{equation}
and we set $F_{j} = F_{MY}(\ket{\psi_{j}})$.

Now we are in a position to prove the bound.  Using the bound for qubits we find
\begin{equation}
F_{MY} = \sum_{j} p_{j} F_{MY}(\ket{\psi_{j}}) =  \sum_{j} p_{j}\frac{ 1 + \sqrt{\left(\frac{S_{max}(\ket{\psi_{j}})}{2}\right)^{2} - 1}}{2}.
\end{equation}
We now know that the convex hull of the qubit bound provides our needed bound.  We find the line connecting the extreme points, namely $S=2, F_{MY} = 1/2$ for the state $\ket{00}$ and $S=2 \sqrt{2}, F_{MY} = 1$ for the state $\ket{\phi_{+}}$.  By Lemma~\ref{lemma:qubitboundconcavedown} this line is always below a convex combination of the qubit bound since the qubit bound is concave down.  Hence
\begin{equation}
F_{MY}(\ket{\psi}) \geq   \frac{S_{max}(\ket{\psi}) + 2 \sqrt{2} - 4}{4\left(\sqrt{2} - 1\right)} .
\end{equation}

Now consider the state 
\begin{equation}
\ket{\psi} = \sqrt{p}\ket{00}_{AB} \ket{00}_{AB} + \sqrt{1-p}\ket{\phi_{+}}_{AB} \ket{11}_{AB}
\end{equation}
for some $p$ with $0 \leq p \leq 1$.  The construction for the CHSH measurement operators given by Gisin and Peres give $S_{max}(\ket{\psi})$ to be $2p + (1-p)2 \sqrt{2}$.  Meanwhile, $F_{MY} = p/2 + (1-p)$.  Thus these states saturate the conjectured lower bound.
\end{proof}

\subsection{Bounds for $F_{LOCC}$}
Our main motivation for introducing $F_{LOCC}$ is that we are able to obtain a tight bound, which we are unable to obtain for $F_{LO}$.  The bound obtained is the same is for $F_{MY}$, and the proof is very similar.

\begin{lemma}
Let a bipartite state $\rho$ be given, then
\begin{equation}
F_{LOCC}(\rho) \geq   \frac{S_{max}(\rho) + 2 \sqrt{2} - 4}{4\left(\sqrt{2} - 1\right)} \end{equation}
\end{lemma}
\begin{proof}
We begin with optimal measurements $A_{a}$ and $B_{b}$ for $a,b = 0,1$.  Applying Lemma~\ref{lemma:blockdiagonalization} twice, we obtain a $2 \times 2$ block diagonalization for both sides, which we may turn into a $4 \times 4$  block diagonalization on the state as a whole.  We project $\rho$ onto the blocks, obtaining a direct sum of four dimensional bipartite states $p_{jk}\rho_{jk}$, which we interpret as  pair of qubits.  We choose $p_{jk}$ so that $\tr{\rho_{jk}} = 1$.  The indices $j$ and $k$ indicate the block on the $A$ and $B$ sides respectively.  Note that $\rho$ may have off-diagonal entries as well, so it may not be the case that $\sum_{jk} p_{jk} \rho_{jk}= \rho$, but because the measurements have the same block structure, $S_{max}$ is unaffected by the projection.

We may interpret the measurement as first projecting onto block $(j,k)$ with probability $p_{jk}$ obtaining reduced state $\rho_{jk}$, followed by a two qubit measurement on $\rho_{jk}$.  The measurement on the $(j,k)$th block must be optimal for $\rho_{jk}$, otherwise we could increase $S_{max}(\rho)$.  This fact allows us to write
\begin{equation}
S_{max}(\rho) = \sum_{j,k} S_{max}(\rho_{jk}).
\end{equation}

Meanwhile, to extract the approximate EPR pair using LOCC operations we first project onto block $(j,k)$ obtaining the state $\rho_{jk}$.  We transmit $j$ and $k$ classically, so that the block is identified on both sides, then apply a local change of bases exactly as in the qubit case, allowing us to extract an EPR pair with fidelity $F_{LO}(\rho_{jk})$.  The combined fidelity is
\begin{equation}
F_{LOCC}(\rho) \geq \sum_{jk} F_{LO}(\rho_{jk}).
\end{equation}

From here we follow the latter half of the proof of Lemma~\ref{lemma:mybound}, which we reproduce here using the current notation.  Using the bound for qubits we find
\begin{equation}
F_{LOCC} \geq \sum_{j,k} p_{jk} F_{LO}(\rho_{jk}) =  \sum_{j,k} p_{jk}\frac{ 1 + \sqrt{\left(\frac{\max(S_{max}(\rho_{jk}),2)}{2}\right)^{2} - 1}}{2}.
\end{equation}
We now know that the convex hull of the qubit bound provides our needed bound.  We find the line connecting the extreme points, namely $S=2, F_{LO} = 1/2$ for the state $\ket{00}$ and $S=2 \sqrt{2}, F_{LO} = 1$ for the state $\ket{\phi_{+}}$.  By Lemma~\ref{lemma:qubitboundconcavedown} this line is always below a convex combination of the qubit bound since the qubit bound is concave down.  Hence
\begin{equation}
F_{LOCC}(\rho) \geq   \frac{S_{max}(\rho) + 2 \sqrt{2} - 4}{4\left(\sqrt{2} - 1\right)} .
\end{equation}
\end{proof}

At first it appears as though the bound is tight.  Consider the states
\begin{equation}
\rho = p \proj{00}{00}_{AB} \otimes \proj{00}{00}_{AB} + (1-p) \proj{\phi_{+}}{\phi_{+}}_{AB} \otimes \proj{11}{11}_{AB}
\end{equation}
It appears that the optimal CHSH measurements would be to identify block (0,0) or (1,1) and apply the optimal CHSH measurement for either $\ket{00}$ in the case of block (0,0), or for $\ket{\phi_{+}}$ in the case of $(1,1)$.  Similarly, the optimal LOCC operations seem to be to identify one of these blocks and map to a pair of qubits.  However, we have no proof that either of these strategies is optimal.  It may be the case that the CHSH measurements are optimal and the LOCC operations are not, in which case the bound would not be tight.

\begin{conjecture}
The bound 
\begin{equation}
F_{LOCC}(\rho) \geq   \frac{S_{max}(\rho) + 2 \sqrt{2} - 4}{4\left(\sqrt{2} - 1\right)} 
\end{equation}
is tight, saturated by the states
\begin{equation}
\rho = p \proj{00}{00}_{AB} \otimes \proj{00}{00}_{AB} + (1-p) \proj{\phi_{+}}{\phi_{+}}_{AB} \otimes \proj{11}{11}_{AB}.
\end{equation}
\end{conjecture}

\subsection{Bounds for $F_{LO}$}
For local operations there is a significant problem in using the proof techniques developed above.  In particular, the block diagonalization into 2 qubit states is indexed by two variables $j$ and $k$, with $j$ the result of a projection on the $A$ side and $k$ the result of a projection on the $B$ side.  Since the operations necessary to fix the local bases may depend on both $j$ and $k$, classical communication is in general required to control the operations.  When calculating $F_{LO}$ this becomes important since classical communication is not allowed.  Nonetheless, we are able to obtain a bound.

\begin{lemma}
Let a bipartite state $\rho$  be given.  Then
\begin{equation}
F_{LO}(\rho) \geq \frac{S_{max}(\rho) - 2}{2(\sqrt{2} - 1)}.
\end{equation}
\end{lemma}

The proof of this result is due to our coauthor on \cite{Bardyn:2009:Device-independ}, Serge Massar.  We omit the proof, but instead provide this insight:  the main idea is to use the measurements to define local bases for each two qubit block.  Since the measurements are completely local in nature, this allows for only local operations.  

We conjecture that this bound is not tight.  The reason is simply that the conjectured lower bound for $F_{MY}$, which must also be a lower bound on $F_{LO}$ since $F_{MY} \leq F_{LO}$, is higher than the current bound on $F_{LO}$.  Another fact of interest is that, if the conjectured bound on $F_{MY}$ holds, and the bound on $F_{LOCC}$ is tight, then all three measures would have the same tight lower bound since $F_{MY}$ lower bounds them all and $F_{LOCC}$ upper bounds them.

\chapter{Concluding Remarks}

As we have demonstrated, the field of black box quantum computing is varied and fruitful.  The general circuit testing construction of self-testing allows for wide application of the black box methodology, while a more tailored approach allows us to draw more specific conclusions for DIQKD and black box state characterization.  Additionally, techniques developed here have application in other settings such as foundations and complexity.

There remain important challenges as well.  For all the results presented in this thesis we need to make additional assumptions in order to collect statistics about the devices.  Even for DIQKD, we must assume that the devices have no memory.  Finding a mechanism for collecting statistics without assuming independent and identical trials is the most important open problem in this field.

\ifthenelse{\boolean{shortdraft}}{}
{

\addcontentsline{toc}{chapter}{Bibliography}
\bibliography{Global_Bibliography}

}

\end{document}